\documentclass[epsf]{iopart}
\usepackage{iopams,amsfonts,amsthm,rotating,graphicx,fancybox,fancyhdr,bm,eucal,dsfont,appendix}
\usepackage{datetime,esint,color}
\usepackage{accents,hyperref}
\usepackage[normalem]{ulem}

\newcommand{\dbtilde}[1]{\accentset{\approx}{#1}}

\eqnobysec
\settimeformat{hhmmsstime}
\ddmmyyyydate

\theoremstyle{plain}
\newtheorem{theorem}{Theorem}[section]
\newtheorem{lemma}[theorem]{Lemma}
\newtheorem{proposition}[theorem]{Proposition}
\newtheorem{corollary}[theorem]{Corollary}
\theoremstyle{definition}
\newtheorem{definition}[theorem]{Definition}
\newtheorem{remark}[theorem]{Remark}

\begin{document}

\title{Power spectrum of the circular unitary ensemble}

\author{\hspace{-2cm}Roman Riser$^{1,2}$ and Eugene Kanzieper$^{1,3}$~\footnote{Corresponding author.}}

\address{\hspace{-2cm}$^1$~Department of Mathematics, Holon Institute of Technology, Holon 5810201, Israel\\
         \hspace{-2cm}$^2$~Department of Physics and Research Center for Theoretical Physics\\
         \hspace{-1.85cm} and Astrophysics, University of Haifa, Haifa 3498838, Israel\\
         \hspace{-2cm}$^3$~Department of Physics of Complex Systems, Weizmann Institute of Science,\\
         \hspace{-1.85cm} Rehovot 7610001, Israel
        }
\noindent\newline\newline
\begin{abstract}
We study the power spectrum of eigen-angles of random matrices drawn from the circular unitary ensemble ${\rm CUE}(N)$ and show that it can be evaluated in terms of
either a Fredholm determinant, or a Toeplitz determinant, or a sixth Painlev\'e function. In the limit of infinite-dimensional matrices, $N\rightarrow\infty$, we derive a {\it concise}  parameter-free formula for the power spectrum which involves a fifth Painlev\'e transcendent and interpret it in terms of the ${\rm Sine}_2$ determinantal random point field. Further, we discuss a universality of the predicted power spectrum law and tabulate it \href{http://eugenekanzieper.faculty.hit.ac.il/data.html}{(\textcolor{blue}{follow this url})} for easy use by random-matrix-theory and quantum chaos practitioners.
\noindent\newline\newline\newline\newline\newline\newline
Published in: \href{https://doi.org/10.1016/j.physd.2022.133599}{\textcolor{blue}{Physica D {\bf 444}, 133599 (2023)}}
\end{abstract}
\newpage 
\tableofcontents

\newpage

\section{Introduction} \label{intro}

\subsection{Brief overview and motivation} \label{intro-overview}
The power spectrum characterization of irregular spectra was introduced long ago by A.~Odlyzko in his famous study~\cite{Od-1987} on the spacing distribution between nontrivial zeros of the Riemann zeta function. Yet, this statistical indicator did not receive much attention until it was re-invented by A.~Rela\~{n}o and collaborators~\cite{RGMRF-2002} fifteen years later. Ever since, the power spectrum has emerged as an effective tool for studying both system-specific and universal properties of complex wave and quantum systems on both local and global energy scales (see discussion below). In the context of quantum systems, it reveals~\cite{RGMRF-2002} whether the corresponding classical dynamics is regular or chaotic, or a mixture~\cite{GRRFSVR-2005,SB-2005} of both, and encodes a `degree of chaoticity'~\cite{PRPAG-2018}. In combination with other long- and short-range spectral fluctuation measures, it provides an effective way to identify system symmetries, determine a degree of incompleteness of experimentally measured spectra, and get the clues about systems’ internal dynamics~\cite{GKKMRR-2011}. A comparative analysis of the power spectrum with other statistical measures of spectral fluctuations as well as a brief overview of a power spectrum characterization of real physical systems can be found in Section~1 of Ref.~\cite{ROK-2020}.

A notion of the power spectrum arises naturally when a sequence of ordered discrete energy levels of a quantum system is interpreted as a discrete time random process, with the indices of ordered eigenlevels playing the r\^ole of time. On the physics side, one is normally interested in the power spectrum of {\it unfolded} energy levels in the limit of infinitely long spectral sequences, when emergence of universal laws is anticipated~\cite{GMGW-1998}. To approach this limit, it makes sense to first define and study the power spectrum for arbitrary, {\it not necessarily unfolded}, eigenlevel sequences of finite length.

\begin{definition}\label{d-01}
  Let $\{\theta_1 \le \dots \le \theta_N\}$ be a sequence of {\it ordered} eigenlevels, $N \in {\mathbb N}$, and let
  $\langle \delta\theta_\ell \delta\theta_m \rangle$ be the covariance matrix of level displacements $\delta\theta_\ell = \theta_\ell - \langle \theta_\ell\rangle$ from their mean $\langle \theta_\ell\rangle$. The Fourier transform of the covariance matrix
\begin{eqnarray}\label{ps-def}
    S_N(\omega) = \frac{1}{N\Delta_N^2} \sum_{\ell=1}^N \sum_{m=1}^N \langle \delta \theta_\ell \delta\theta_m \rangle\, e^{i\omega (\ell-m)}, \quad \omega \in {\mathbb R}
\end{eqnarray}
is called the {\it power spectrum of eigenlevels}. Here, $\Delta_N$ is a constant (that may depend on $N$) which acquires a meaning of the mean level spacing for unfolded spectrum. The angular brackets stand for an average over an ensemble of eigenlevel sequences.
\hfill $\blacksquare$
\end{definition}
Since the power spectrum is $2\pi$-periodic, real and even function of $\omega$, it is sufficient to study it in the interval $0 \le \omega \le \omega_{\rm Ny}$, where $\omega_{\rm Ny} = \pi$ is the Nyquist frequency.

Notice that, by tuning the frequency $\omega$ in the power spectrum, one may attend to spectral correlations between either adjacent or distant eigenlevels. Indeed, (i)~at `large' frequencies, $\omega={\mathcal O}(N^0)$, yet below the Nyquist frequency $\omega_{\rm Ny}=\pi$, the distant eigenlevels barely contribute to the power spectrum; characterized by large values of $|\ell-m|$, they produce strongly oscillating terms in Eq.~(\ref{ps-def}) which effectively cancel each other. As the result, $S_N(\omega)$ is mainly shaped by correlations between the nearby levels. (ii)~At low frequencies, $\omega \ll 1$, these oscillations are by far less pronounced thus making a contribution of distant eigenlevels increasingly important. Since their fluctuations are system specific, an infrared part of the power spectrum can be used to uncover nonuniversal aspects of quantum dynamics, see e.g. Ref.~\cite{CLKT-2022}.

Early numerical simulations \cite{RGMRF-2002}, followed by a heuristic description \cite{FGMMRR-2004} of the power spectrum, have revealed that the power spectrum of long eigenlevel sequences discriminates sharply between quantum systems with chaotic and integrable classical dynamics. At low frequencies, two simple power laws, $\sim 1/\omega$ and $\sim 1/\omega^2$, were argued to describe the power spectrum of quantum systems with chaotic and integrable classical dynamics, respectively.

A nonperturbative theory of the power spectrum was developed in Refs.~\cite{ROK-2020,ROK-2017,RK-2021}, where {\it stationarity} of level {\it spacings}~\footnote[1]{A sequence $\{s_1,s_2,\cdots\}$ of level spacings of ordered eigenlevels is said to be stationary, if the average $\langle s_\ell \rangle$ does not depend on the position $\ell$ of the corresponding eigenlevel and the covariance matrix $\langle s_\ell s_m\rangle$ is of the Toeplitz type. For a necessary and sufficient condition of {\it spacings} stationarity, formulated in terms of ordered sequence of {\it eigenlevels}, see Lemma~3.1 in Ref.~\cite{ROK-2020}.} was  the only, but central, assumption. This approach has further been employed to determine the power spectrum for (i) {\it fully chaotic} quantum systems with broken time-reversal symmetry, mimicked by the ``tuned'' circular unitary ensemble~\cite{ROK-2020,ROK-2017}, abbreviated as ${\rm TCUE}(N)$ throughout the paper, and (ii) generic quantum systems with {\it integrable} classical dynamics, described by {\it truncated} random diagonal matrices~\cite{RK-2021}.

While the assumption of stationarity of level spacings can rigorously be justified for a number of spectral models~\footnote[2]{These include: (i) eigenlevels drawn from the ``tuned'' circular ensembles~\cite{ROK-2020,ROK-2017} of random matrices appearing in the context of quantum systems with completely chaotic classical dynamics, (ii) a set of uncorrelated identically distributed eigenlevels~\cite{ROK-2020} and (iii) unfolded  spectra of diagonal random matrices~\cite{RK-2021}; the two latter models are used to account for quantum systems with integrable classical dynamics.} describing quantum systems with chaotic and integrable classical dynamics, generically, a finite-$N$ stationarity of level spacings is the {\it exception} rather than the rule. For one, level spacings sequences in the celebrated triad~\cite{D-1962} of Dyson's circular ensembles do {\it not} possess the required stationarity at finite $N$. Hence, an extension of the theory~\cite{ROK-2020,ROK-2017,RK-2021} to a wider class of spectral models is very much called for.

\subsection{Main results and discussion}\label{m-res-dis}
\noindent
{\bf First result.}---In Section~\ref{exact-theory}, we formulate a nonperturbative theory of the power spectrum for generic, random spectral sequences of finite length $N$ {\it without} assuming stationarity of eigenlevel spacings. Its outcome is summarized in Theorem~\ref{main-technical-theorem} which establishes an exact relation between the power spectrum $S_N(\omega)$ and the moment generating function of the eigenvalue counting function, see Remark~\ref{MGF-CF-remark}. Let us stress that Theorem~\ref{main-technical-theorem} holds well beyond the random-matrix-theory paradigm.
\noindent\newline\newline
{\bf Second result.}---In Section~\ref{power-spectrum-CUE}, we apply this framework to determine the power spectrum in the Dyson's circular unitary ensemble ${\rm CUE}(N)$. In particular, we show that flatness of the Haar measure on $U(N)$ group  reduces our general Theorem~\ref{main-technical-theorem} to its simplified version given by Theorem~\ref{circular-theorem}. Further, we establish three alternative, albeit equivalent, representations of the ${\rm CUE}(N)$ power spectrum: Theorem~\ref{painleve-vi} expresses the power spectrum in terms of a sixth Painlev\'e function. Proposition~\ref{prop-fred} and Proposition~\ref{toep-prop} provide representations in terms of Fredholm and Toeplitz determinants, respectively.
\noindent\newline\newline
{\bf Third (central) result.}---In Sections~\ref{CK-section} and \ref{cue-pv-proof}, we analyze the exact ${\rm CUE}(N)$ solution, given by Theorem~\ref{circular-theorem} and Proposition~\ref{toep-prop}, in the limit of infinite-dimensional unitary matrices to derive a {\it concise} parameter-free formula for the power spectrum which involves a fifth Painlev\'e transcendent. Theorem~\ref{cue-pv} is the {\it central result} of our study.
\begin{theorem}\label{cue-pv}
Let $S_N(\omega)$ denote the power spectrum of the ${\rm CUE}(N)$. For all $0<\omega< \pi$, the limit $\lim_{N\rightarrow\infty} S_N(\omega)$ exists and equals
  \begin{eqnarray} \label{ps-cue-inf}\fl\qquad
    S_\infty(\omega) = \frac{1}{4\pi \sin^2(\omega/2)} {\rm Re\,}\int_{0}^{\infty} d\lambda \, \exp\left(
        \int_{0}^{\lambda}\frac{dt}{t}\, \sigma_0(t; \zeta=1-e^{i\omega})
    \right),
\end{eqnarray}
where $\sigma_0(t;\zeta)$ denotes a family of one-parameter solutions to the fifth Painlev\'e equation ($\sigma$-PV)
\begin{equation} \label{PV-eq-0}
    (t \sigma_0^{\prime\prime})^2 + (t\sigma_0^\prime -\sigma_0) \left(
            t\sigma_0^\prime -\sigma_0 + 4 (\sigma_0^\prime)^2
        \right)= 0,
\end{equation}
which are analytic at $t=0$ and satisfy the boundary condition
\begin{eqnarray}\label{bc-zero-0}
\sigma_0(t; \zeta)=  -\frac{t}{2\pi}\zeta -\left(\frac{t}{2\pi}\right)^2\zeta^2 + {\mathcal O}(t^3) \quad {\rm as} \quad t\rightarrow 0.
\end{eqnarray}
\end{theorem}
\begin{remark}\label{LSD-vs-PS}
It is a remarkable fact that both the power spectrum [Eq.~(\ref{ps-cue-inf})] and the level spacing distribution~\cite{JMMS-1980}
\begin{eqnarray}\label{LSD-PV}
    P_\infty(s) = \frac{d^2}{ds^2} \exp \left(
            \int_{0}^{s} \frac{\sigma_0(t;\zeta=1)}{t} dt
        \right)
\end{eqnarray}
in the infinite-dimensional ${\rm CUE}$ ensembles are governed by fifth Painlev\'e transcendents belonging to the {\it same family} of one-parameter solutions~\cite{JMMS-1980,MCT-1986} to the $\sigma$-PV equation specified in Theorem~\ref{cue-pv}. Yet, these two spectral statistics are {\it profoundly different}: while the level spacing distribution $P_\infty(s)$ is merely determined by the {\it gap formation probability} $E_{1/2\pi,\infty}(0;\lambda)$, the power spectrum $S_\infty(\omega)$ is shaped by the {\it entire set of probabilities} $\{E_{1/2\pi,\infty}(\ell;\lambda)\}_{\ell=0}^\infty$ to find exactly $\ell$ unfolded eigen-angles in a spectral interval of a given length, see Remark~\ref{PV-GFP} for the notation. On a formal level, this conclusion emerges out of the expansion Eq.~(\ref{pv-prob-expansion}) which reduces to the gap formation probability at $\zeta=1$, see Eq.~(\ref{LSD-PV}), but remains an (infinite) series when $\zeta= 1 - e^{i\omega}$, see Eq.~(\ref{ps-cue-inf}), for all $0<\omega<\pi$.
\hfill $\blacksquare$
\end{remark}

To illustrate our main result, we turn to Figs.~\ref{Fig-1-cue} and \ref{Fig-2-cue}, where the parameter-free prediction for the power spectrum, Theorem~\ref{cue-pv}, is confronted with the results of a numerical experiment conducted for the large-$N$ circular unitary ensemble ${\rm CUE}(N)$. Let us stress, that the limiting (red) curve $S_\infty(\omega)$ in both figures displays a {\it genuine} $N\rightarrow\infty$ result obtained by a direct numerical evaluation~\footnote[8]{In Refs.~\cite{ROK-2020, ROK-2017}, computation of the power spectrum was based on the ${\rm dPV}$ representation~\cite{FW-2005} of the exact, finite-$N$ Painlev\'e VI solution (akin to Eq.~(\ref{phin-cue-painleve-6}) of Theorem~\ref{Th-PVI}) taken at $N=10^4$, see Figs.~4 and 5 in Ref.~\cite{ROK-2020}.} of the exponent in Eq.~(\ref{ps-cue-inf}). The agreement between our analytical prediction Eq.~(\ref{ps-cue-inf}) and the results of a numerical experiment is nearly perfect. (Seemingly large absolute fluctuations clearly visible at very small frequencies in the inset in Fig.~\ref{Fig-2-cue}, where the {\it difference} $\delta S_\infty(\omega)$ between the power spectrum and its large singular component $1/2\pi\omega$ is displayed, do not have any significant influence on the relative error in the power spectrum $S_\infty(\omega)$ itself, see Fig.~\ref{Fig-1-cue}.)

\begin{figure}
\includegraphics[width=\textwidth]{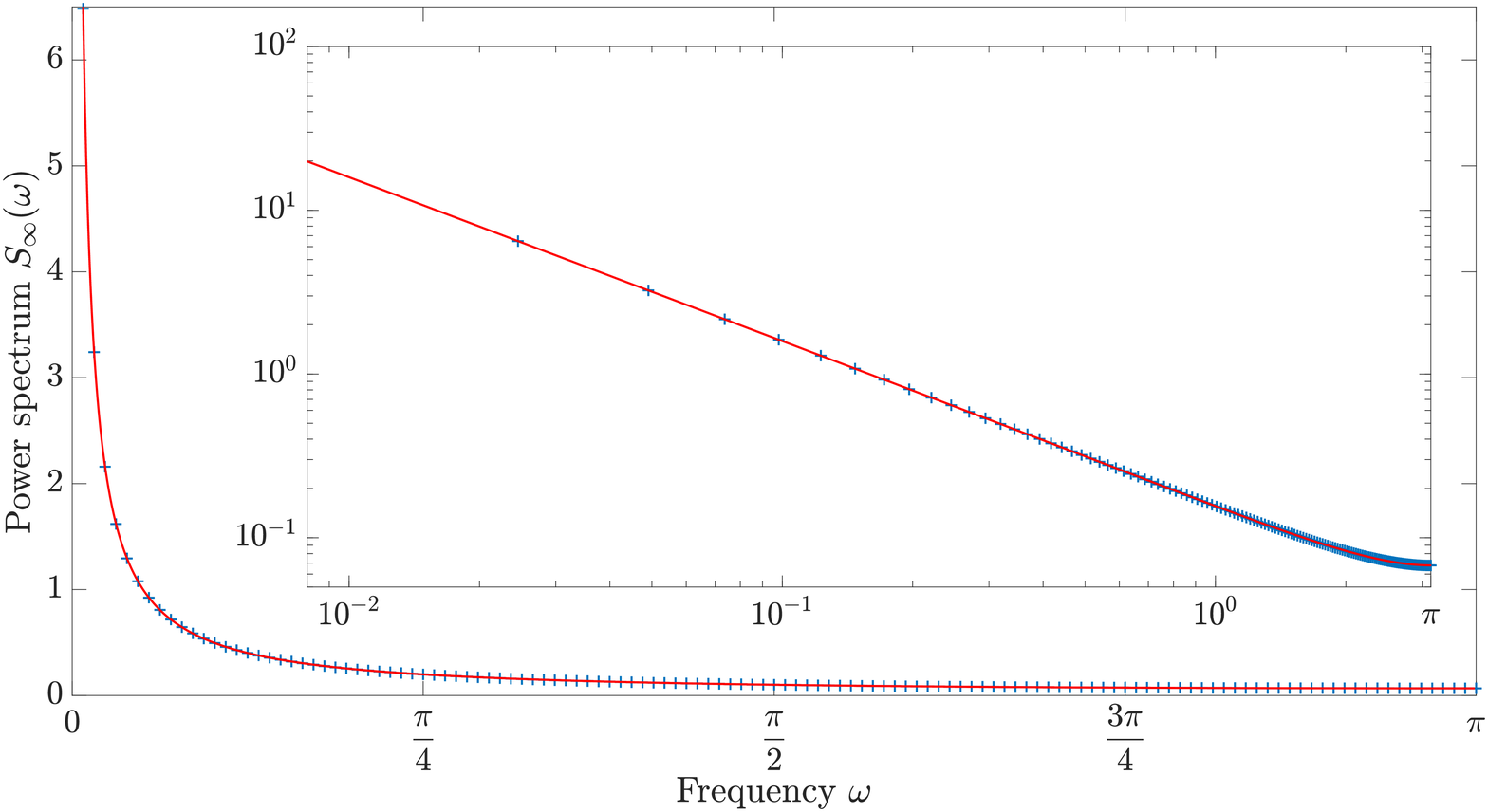}
\caption{A graph for the power spectrum $S_\infty(\omega)$ as a function of frequency. Red curve corresponds to the power spectrum calculated numerically using Eq.~(\ref{ps-cue-inf}), thus representing a genuine $N\rightarrow\infty$ result. Blue crosses show the power spectrum calculated for sequences of $256$ unfolded ${\rm CUE}$ eigen-angles averaged over $10^7$ realizations, borrowed from Ref.~\cite{ROK-2020}. Inset: a log-log plot for the same graph.
\label{Fig-1-cue}}
\end{figure}

\begin{figure}
\includegraphics[width=\textwidth]{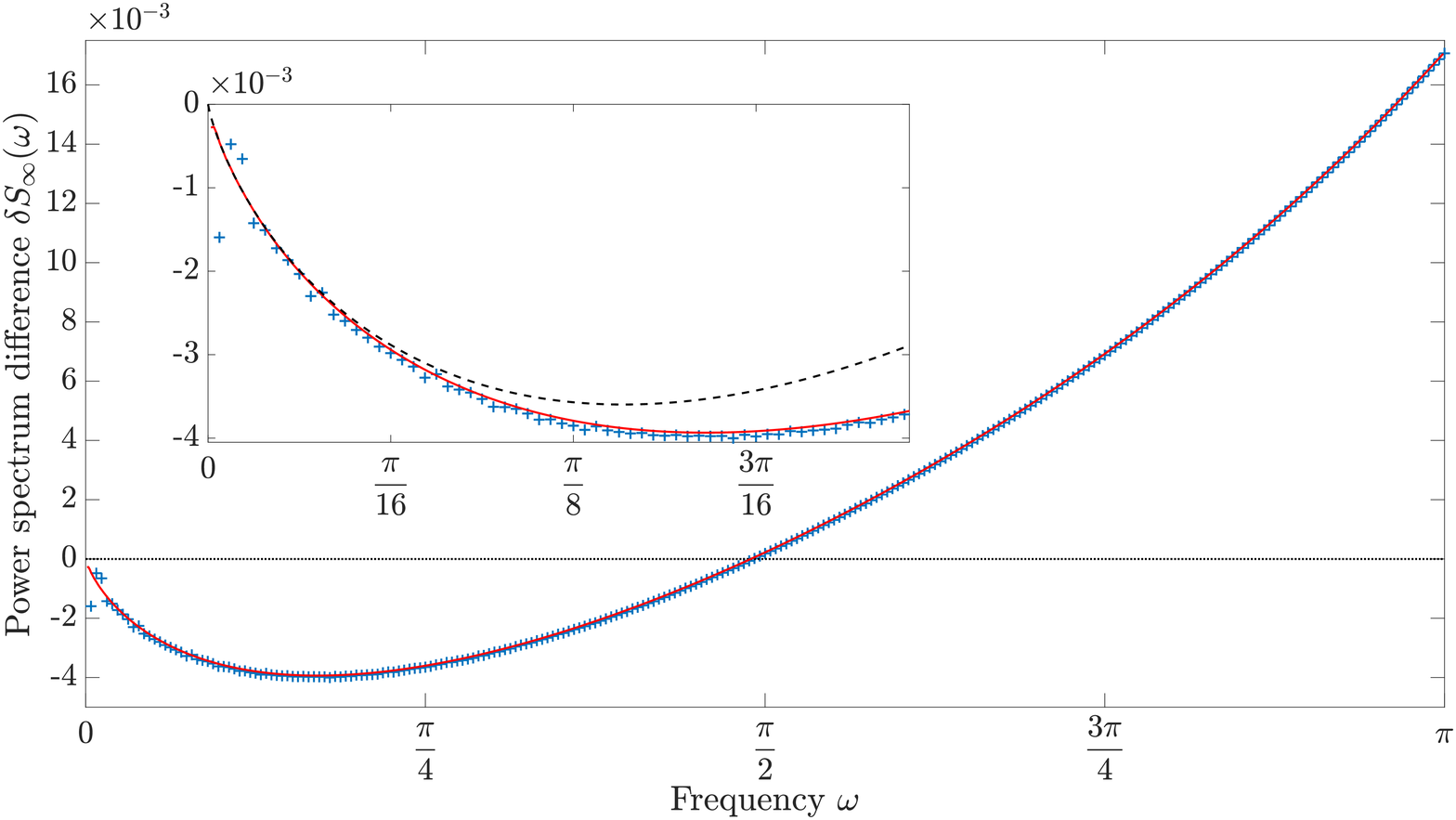}
\caption{Difference $\delta S_\infty(\omega)$ between the power spectrum $S_\infty(\omega)$ and its singular part $1/2\pi\omega$ described by the first term in Eq.~(\ref{S-res-0}).
The singular part of the power spectrum corresponds to $\delta S_\infty(\omega)=0$ as represented by a gray dotted line. Red curve: analytical prediction computed as explained in Fig.~\ref{Fig-1-cue}. Blue crosses: numerical experiment based on $4 \times 10^8$ sequences of $512$ unfolded CUE eigenvalues, borrowed from Ref.~\cite{ROK-2020}. Inset: magnified portion of the same graph for $0\le \omega \le \pi/4$; additional black dashed line displays the difference $\delta S_\infty(\omega)$ calculated using the small-$\omega$ expansion Eq.~(\ref{S-res-0}).
\label{Fig-2-cue}}
\end{figure}
\noindent\newline
{\bf Fourth result.}---In Section~\ref{Sine-2-DPP}, we show that our main result, Theorem~\ref{cue-pv}, can naturally be interpreted in the language of the ${\rm Sine}_2$ determinantal point process~\cite{D-1962-CUE,KS-2009,VV-2020} since the latter describes the bulk scaling limit of unfolded ${\rm CUE}(N)$ eigen-angles as $N\rightarrow\infty$.

\begin{lemma}\label{Sine-2-Lemma}
    Let $n_\rho(\lambda)$ be a random counting function of the $Sine_2$ determinantal point process with the mean local density $\rho=1/2\pi$. For all $0<\omega<\pi$, it holds:
\begin{eqnarray} \label{Sine-2-Lemma-eq}
    S_\infty(\omega) = \frac{1}{4\pi \sin^2(\omega/2)} {\rm Re\,}\int_{0}^{\infty} d\lambda \,\langle
            e^{i\omega n_{1/2\pi}(\lambda)}
    \rangle.
\end{eqnarray}
\end{lemma}
We have used this interpretation in Remark~\ref{LSD-vs-PS} to highlight a profound difference between the two spectral fluctuation measures -- level spacing distribution and the power spectrum. It will further be adopted to extend our knowledge about the ${\rm Sine}_2$ process itself (see the ``sixth result'' below).
\noindent\newline\newline
{\bf Fifth result.}---In Section~\ref{small-omega-section}, we derive a small-$\omega$ expansion of the ${\rm CUE}(\infty)$ power spectrum. The reader is referred to
Proposition~\ref{Th-small-omega} for explicit formula.
\noindent\newline\newline
{\bf Sixth result.}---In Appendix~\ref{A-3}, we combine an interpretation of the power spectrum in terms of the ${\rm Sine}_2$ process (``fourth result'') with a small-$\omega$ analysis of the power spectrum (``fifth result''), to derive an explicit formula [Eqs.~(\ref{kappa-3-new}) and (\ref{f3-integrated})] for the {\it third cumulant} $\kappa_3^{(\rho)}(\lambda)=\langle\!\langle n_\rho^3(\lambda)\rangle\!\rangle$ of the ${\rm Sine}_2$ counting function. This result, which we have not managed to find in the random-matrix-theory literature, adds up to the well-known exact formula~\cite{M-2004} for second cumulant $\kappa_2^{(\rho)}(\lambda)=\langle\!\langle n_\rho^2(\lambda)\rangle\!\rangle$ calculated in the early days of the random matrix theory. The third cumulant, albeit expressed in terms of hypergeometric and Meijer $G$-function, exhibits rather simple asymptotic behavior:
\begin{eqnarray} \fl\quad \label{3rd-cum-as-intro}
    \kappa_3^{(\rho)}(\lambda)= \left\{
                                         \begin{array}{ll}
                                           \displaystyle \rho\lambda +{\mathcal O}(\lambda^2), & \hbox{$\lambda \rightarrow 0$;} \\
                                          \displaystyle \frac{2}{\pi \rho \lambda}\left\{ 1 + \frac{\sin(2\pi\rho\lambda)}{\pi\rho\lambda} \big(
                                            \log(2\pi\rho\lambda)+ \gamma \big)
                                            \right\} +{\mathcal O}(\lambda^{-3}), & \hbox{$\lambda \rightarrow\infty$.}
                                         \end{array}
                                       \right.
\end{eqnarray}
\begin{figure}
\includegraphics[width=\textwidth]{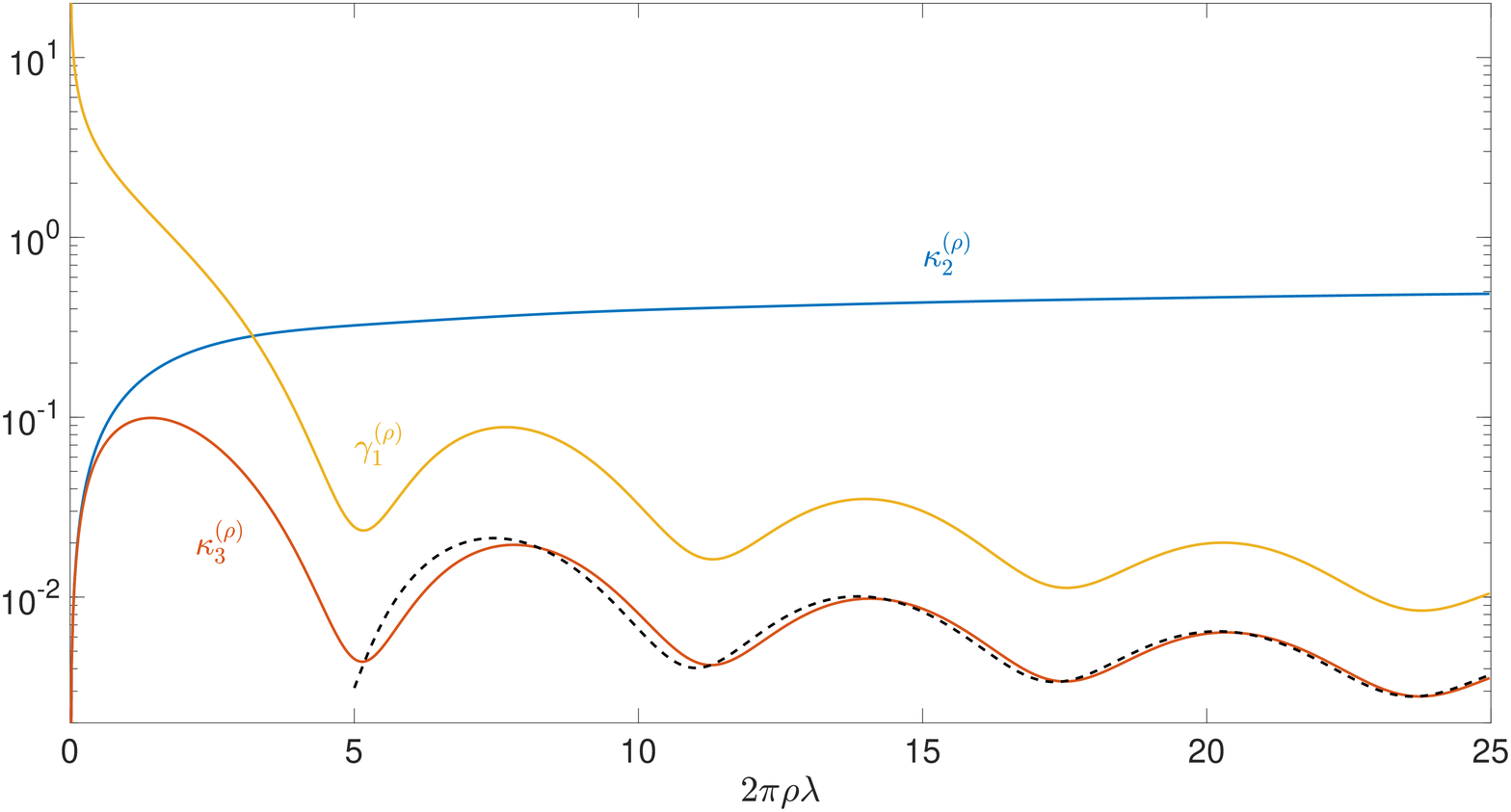}
\caption{A semi-log plot of the second $\kappa_2^{(\rho)}(\lambda)$ and third $\kappa_3^{(\rho)}(\lambda)$ cumulants, and the skewness $\gamma_1^{(\rho)}(\lambda) = \kappa_3^{(\rho)}/(\kappa_2^{(\rho)})^{3/2}$ of the eigenvalue counting function $n_{\rho}(\lambda)$ for the ${\rm Sine}_2$ determinantal random point field with the mean local density $\rho$. Solid blue, red and yellow curves correspond to $\kappa_2^{(\rho)}$ [Eq.~(\ref{2nd-cum})], $\kappa_3^{(\rho)}$ [Eqs.~(\ref{kappa-3-new}) and (\ref{f3-integrated})] and $\gamma_1^{(\rho)}$, respectively. Black dashed line represents a large-$\lambda$ asymptotic behavior of the third cumulant, see Eq.~(\ref{3rd-cum-as-intro}).
\label{Fig-3-cumulants}}
\end{figure}

In Figure~\ref{Fig-3-cumulants}, we plot both the second and the third cumulants alongside the skewness of the ${\rm Sine}_2$ counting function. Remarkably, while the second cumulant (aka number variance) is a monotonically increasing function of the interval length $\lambda$, the third cumulant shows a somewhat unusual oscillatory behavior. The skewness, defined as the ratio $\gamma_1(\lambda) = \kappa_3^{(\rho)}(\lambda)/\kappa_2^{(\rho)}(\lambda)^{3/2}$, displays similar oscillations. Long ago, they were observed in numerical studies of higher-order correlations between nontrivial zeros of the Riemann zeta function, see Ref.~\cite{Od-1987} and Fig.~16.14(a) in Ref.~\cite{M-2004}.
\noindent\newline\newline
{\bf Seventh result.}---Finally, we return to the issue of anticipated universality~\cite{ROK-2020,ROK-2017} of the power spectrum law [Eq.~(\ref{ps-cue-inf})] in the context of random matrix models belonging to the $\beta=2$ Dyson universality class. In Refs.~\cite{ROK-2020,ROK-2017}, a {\it numerical} evidence was presented that the power spectra in large-dimensional ${\rm CUE}$ and ${\rm TCUE}$ ensembles coincide with each other. In Section~\ref{CUE-TCUE-equiv}, we show {\it analytically} that the power spectra in ${\rm CUE}(\infty)$ and ${\rm TCUE}(\infty)$ are indeed described by the same universal law, see Theorem~\ref{tcue-pv}.
\begin{figure}
\includegraphics[width=\textwidth]{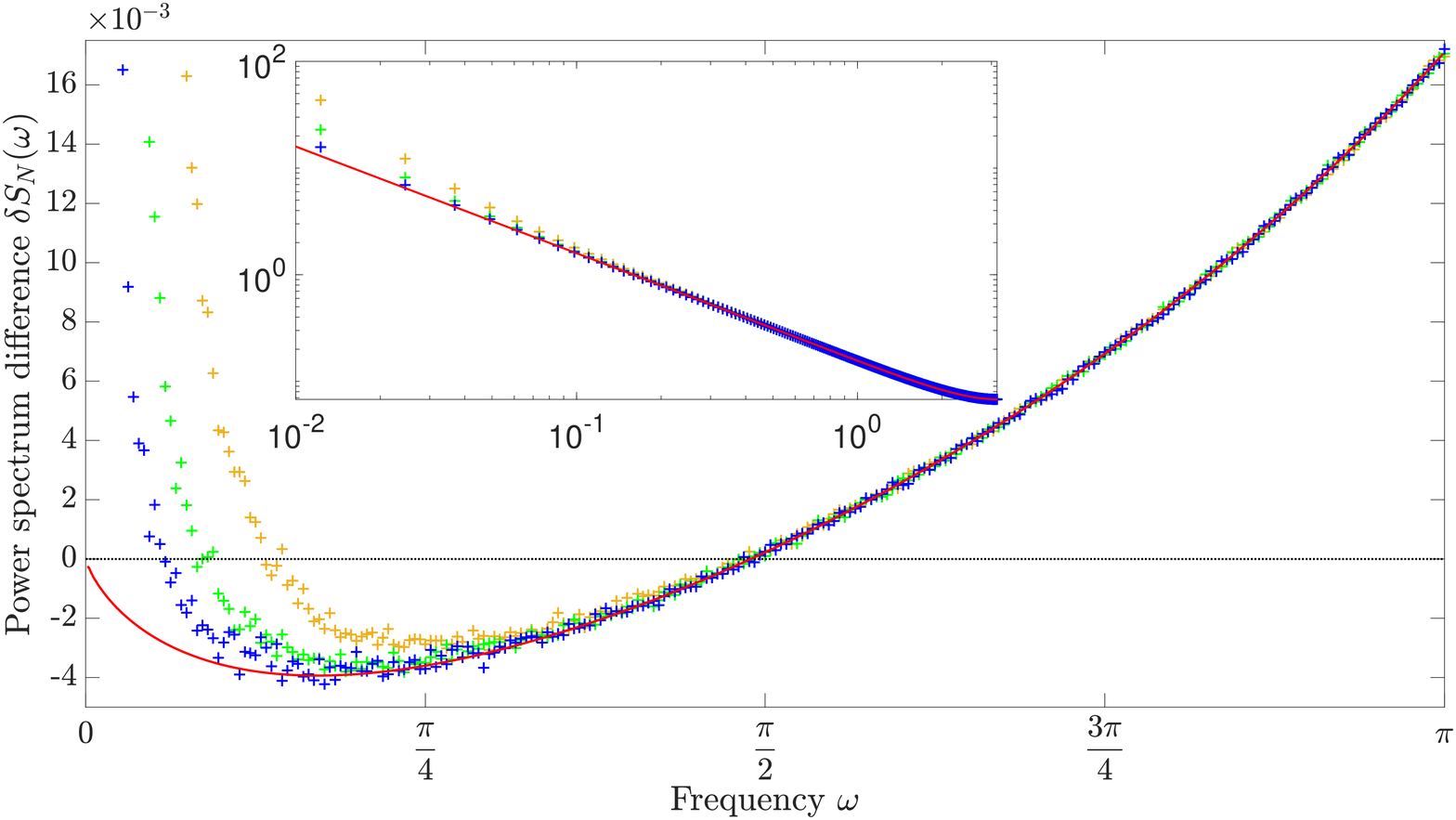}
\caption{Difference $\delta S_N(\omega)$ between the ${\rm GUE}(N)$ power spectrum and the singular part $1/2\pi\omega$ described by the first term in Eq.~(\ref{S-res-0}).
Red curve: analytical prediction computed as explained in Figs.~\ref{Fig-1-cue} and~\ref{Fig-2-cue}. Yellow, green and blue crosses represent a numerical experiment based on $10^6$ sequences of $N=512$ (yellow), $N=2048$ (green) and $N=8192$ (blue) unfolded ${\rm GUE}(N)$ eigenvalues. Inset: a log-log plot for the same data but without subtraction of  the singular part $1/2\pi\omega$.
\label{Fig-4-gue}}
\end{figure}

In addition, we perform an extensive numerical experiment for large-dimensional random matrices drawn from the Gaussian unitary ensemble ${\rm GUE}(N)$. Experimentally produced ${\rm GUE}(N)$ eigenvalues were unfolded with regard to the mean level density computed from all generated samples. Figure~\ref{Fig-4-gue} shows that the power spectrum of unfolded ${\rm GUE}(N)$ eigenvalues converges to the parameter-free ${\rm CUE}(\infty)$ law as the matrix dimension $N$ grows, in concert with the universality conjecture. The discrepancy between experimental data and the theoretical curve is more pronounced at small frequencies, which hints that the convergence of the finite-$N$ ${\rm GUE}$ power spectrum to the universal law is not uniform near $\omega=0$.

Analytical proof of the power spectrum universality remains an open question.
\noindent\newline\newline
{\bf Eighth result.}---We perform a numerical evaluation of the power spectrum law based on Eq.~(\ref{ps-cue-inf}) and tabulate it for easy use by random-matrix-theory and quantum chaos practitioners. The computation of the power spectrum close to the Nyquist frequency was based on the Bornemann code~\cite{B-2010}. The online data available in Ref.~\cite{PS-data-2022} provide the power spectrum $S_\infty(\omega)$ calculated to an absolute and relative errors better than $\delta_{{\rm abs}} =10^{-6}$ and $\delta_{{\rm rel}} =10^{-5}$, respectively.

\section{Power spectrum for general spectral sequences} \label{exact-theory}

\subsection{Preliminaries} \label{preliminaries}
To set the stage, let us consider a sequence $\{0\le \theta_1 \le \dots \le \theta_N < 2\pi\}$ of $N \in {\mathbb N}$ {\it ordered} eigenlevels such that $\{e^{i\theta_j}\}_{j=1}^N$ belong to the unit circle~\footnote[3]{A circular setup is chosen for further convenience. Reformulation for the case of random sequence supported on the real axis is straightforward, see, e.g., Ref.~\cite{ROK-2020}.} and let $P_N(\theta_1,\dots, \theta_N)$ denote their {\it symmetrized} joint probability density function (JPDF)~\footnote[4]{
The normalization is fixed by
$$
\prod_{j=1}^{N}\int_0^{2\pi} \frac{d\theta_j}{2\pi}\, P_N(\theta_1,\dots,\theta_N)=1.
$$} which stays invariant under arbitrary permutation of its arguments.

As soon as the power spectrum is a particular example of the {\it order statistics} (the indices $(\ell,m)$ in Definition~\ref{d-01} label {\it ordered} eigenlevels), it is beneficial to introduce three spectral fluctuation measures:

(i) the probability density function of the $\ell$-th ordered eigenvalue
\begin{eqnarray} \label{PLN} \fl
\qquad
     p_{\ell}(\varphi) =
     \frac{N!}{(N-\ell)!}\frac{1}{(\ell-1)!}
    \left(\int_{0}^{\varphi}\right)^{\ell-1} \left(\int_{\varphi}^{2\pi}\right)^{N-\ell} \, \prod_{j=2}^{N} \frac{d\theta_j}{2\pi}\, P_N(\varphi, \theta_2,\dots,\theta_N);
\end{eqnarray}

(ii) the joint probability density function of the $\ell$-th and $m$-th ordered eigenvalues $(\ell<m)$
\begin{eqnarray} \label{PLMN} \fl
    p_{\ell m}(\varphi,\varphi^\prime) = \frac{N!}{(N-m)!}\frac{1}{(\ell-1)! (m-\ell-1)!}
    \left(\int_{0}^{\varphi}\right)^{\ell-1}
    \left(
        \int_{\varphi}^{\varphi^\prime}
    \right)^{m-\ell-1}
    \left(\int_{\varphi^\prime}^{2\pi}
    \right)^{N-m} \nonumber\\
    \times\,
    \prod_{j=3}^{N} \frac{d\theta_j}{2\pi}\, P_N(\varphi,\varphi^\prime, \theta_3,\dots,\theta_N);
\end{eqnarray}
and

(iii) the probability to find exactly $\ell$ eigenvalues within the interval $(\varphi,\varphi^\prime)$:
\begin{eqnarray} \label{ELN} \fl
 E_{N}(\ell; (\varphi,\varphi^\prime)) = \frac{N!}{\ell!(N-\ell)!}
    \left(\int_{\varphi}^{\varphi^\prime}\right)^\ell \left(\int_{0}^{2\pi} - \int_{\varphi}^{\varphi^\prime}\right)^{N-\ell}
    \prod_{j=1}^{N} \frac{d\theta_j}{2\pi}\, P_N(\theta_1,\dots,\theta_N).
\end{eqnarray}
In both Eqs.~(\ref{PLMN}) and (\ref{ELN}) it is assumed that $\varphi <\varphi^\prime$.

\begin{lemma}\label{Lemma-3GF}
Define the three generating functions
  \begin{eqnarray} \label{GF0-d}
    \Phi_N((\varphi,\varphi^\prime); 1-z) &=& \sum_{\ell=0}^N z^\ell  E_{N}(\ell; (\varphi,\varphi^\prime)), \\
    \label{GF1-d}
    \Phi_N^{(1)}(\varphi;1-z) &=& \sum_{\ell=1}^N z^\ell  p_\ell(\varphi), \\
    \label{GF2-d}
    \Phi_N^{(2)}(\varphi,\varphi^\prime;1-z) &=& \sum_{\ell<m}^N z^{m-\ell}  p_{\ell m}(\varphi,\varphi^\prime),
\end{eqnarray}
where $z \in {\mathbb C}$. The following relations hold:
\begin{eqnarray}
    \label{GF0-rel} \fl \qquad
    \Phi_N((\varphi,\varphi^\prime); 1-z) &=&
    \left( \int_{0}^{2\pi} -(1-z)  \int_{\varphi}^{\varphi^\prime}
    \right)^{N}  \prod_{j=1}^{N} \frac{d\theta_j}{2\pi}\, P_N(\theta_1,\dots,\theta_N), \\
    \label{GF1-rel} \fl \qquad
    \Phi_N^{(1)}(\varphi;1-z) &=&  -z \frac{2\pi}{1-z} \frac{d}{d\varphi} \Bigg(
        \Phi_N((0,\varphi);1-z) - z^N \Bigg), \\
    \label{GF2-rel} \fl \qquad
    \Phi_N^{(2)}(\varphi,\varphi^\prime;1-z) &=&
    -z \left( \frac{2\pi}{1-z}  \right)^2 \frac{\partial^2}{\partial\varphi \partial\varphi^\prime} \Phi_N((\varphi,\varphi^\prime); 1-z).
\end{eqnarray}
\end{lemma}
\begin{proof}
(i) To prove Eq.~(\ref{GF0-rel}), we substitute Eq.~(\ref{ELN}) into Eq.~(\ref{GF0-d}) and perform formal summation by Newton's binomial formula.

(ii) To prove Eq.~(\ref{GF1-rel}), we notice that Eqs.~(\ref{PLN}) and (\ref{ELN}) imply the relation
\begin{eqnarray} \label{ENL-PL}
    \frac{d}{d\varphi} E_{N}(\ell; (0,\varphi)) = \frac{1}{2\pi} \left(
    p_\ell(\varphi) - p_{\ell+1}(\varphi)
    \right)
\end{eqnarray}
which holds for $\ell=0,\cdots,N$ if we formally set $p_0(\varphi) = p_{N+1}(\varphi)=0$. Equation~(\ref{ENL-PL}) is equivalent to
\begin{eqnarray}\label{PL-ENJ}
   p_\ell(\varphi) = -2\pi \sum_{j=0}^{\ell-1} \frac{d}{d\varphi} E_N(j;(0,\varphi)), \quad \ell=1,\dots,N.
\end{eqnarray}
With Eq.~(\ref{GF1-d}) in mind, we multiply Eq.~(\ref{PL-ENJ}) by $z^\ell$ and interchange summation order to derive
\begin{eqnarray} \fl \qquad
    \Phi_N^{(1)}(\varphi;1-z) = \sum_{\ell=1}^N z^\ell  p_\ell(\varphi) = -z \frac{2\pi}{1-z} \frac{d}{d\varphi} \sum_{j=0}^{N}(z^{j}-z^{N})  E_N(j;(0,\varphi)).
\end{eqnarray}
Invoking the normalization condition $\sum_{j=0}^N E_N(j;(0,\varphi)) =1$ completes the proof of Eq.~(\ref{GF1-rel}).

(iii) To prove Eq.~(\ref{GF2-rel}), we substitute Eq.~(\ref{PLMN}) into Eq.~(\ref{GF2-d}) to write down
\begin{eqnarray}
    \Phi_N^{(2)}(\varphi,\varphi^\prime;1-z) = z \sum_{\ell=1}^{N-1} \sum_{\sigma=0}^{N-\ell-1} z^\sigma
    \, p_{\ell,\ell+1+\sigma} (\varphi,\varphi^\prime).
\end{eqnarray}
Performing the inner summation with the help of Eq.~(\ref{PLMN}),
\begin{eqnarray} \fl
\sum_{\sigma=0}^{N-\ell-1} z^\sigma\, p_{\ell,\ell+1+\sigma} (\varphi,\varphi^\prime) = \frac{N!}{(\ell-1)! (N-\ell-1)!} \nonumber\\
    \fl \qquad\qquad
    \times
   \left(\int_{0}^{\varphi}\right)^{\ell-1}
    \left(
        z \int_{\varphi}^{\varphi^\prime} + \int_{\varphi^\prime}^{2\pi}
    \right)^{N-\ell-1}  \prod_{j=3}^{N} \frac{d\theta_j}{2\pi}\, P_N(\varphi,\varphi^\prime, \theta_3,\dots,\theta_N),
\end{eqnarray}
we realize that the outer summation is equally feasible; straightforward algebra yields:
\begin{eqnarray} \fl
    \qquad \Phi_N^{(2)}(\varphi,\varphi^\prime;1-z) = z N(N-1)  \nonumber\\
    \times
    \left( \int_{0}^{2\pi} -(1-z)  \int_{\varphi}^{\varphi^\prime}
    \right)^{N-2}  \prod_{j=3}^{N} \frac{d\theta_j}{2\pi}\, P_N(\varphi,\varphi^\prime, \theta_3,\dots,\theta_N).
\end{eqnarray}
Further, we spot that Eq.~(\ref{GF0-rel}) implies
\begin{eqnarray} \fl \qquad
    N(N-1) \left( \int_{0}^{2\pi} -(1-z)  \int_{\varphi}^{\varphi^\prime}
    \right)^{N-2}  \prod_{j=3}^{N} \frac{d\theta_j}{2\pi}\, P_N(\varphi,\varphi^\prime, \theta_3,\dots,\theta_N)
    \nonumber\\
    \qquad\qquad =
    -\left( \frac{2\pi}{1-z} \right)^2 \frac{\partial^2}{\partial\varphi \partial\varphi^\prime} \Phi_N((\varphi,\varphi^\prime); 1-z).
\end{eqnarray}
\end{proof}

\begin{remark}
    The symmetry relation for the probabilities
    \begin{eqnarray}  \label{compl-ENL}
        E_{N}(\ell; (0,2\pi)\smallsetminus(\varphi,\varphi^\prime)) = E_{N}(N-\ell; (\varphi,\varphi^\prime))
    \end{eqnarray}
    induces the property
    \begin{eqnarray} \label{comp-PhiN}
            \Phi_N((0,2\pi)\smallsetminus(\varphi,\varphi^\prime); 1-z)= z^N \Phi_N((\varphi,\varphi^\prime); 1-z^{-1}).
    \end{eqnarray}
\hfill $\blacksquare$
\end{remark}
\begin{remark}\label{MGF-CF-remark}
  The generating function $\Phi_N((\varphi,\varphi^\prime); 1-e^{i\omega})$ is essentially the moment generating function
  \begin{eqnarray}\label{mgf}
  \Phi_N((\varphi,\varphi^\prime); 1-e^{i\omega}) = \left<
        e^{i \omega {\mathcal N}_N(\varphi,\varphi^\prime)}
  \right>
  \end{eqnarray}
of the (random) eigenvalue counting function
\begin{eqnarray} \label{NCF}
    {\mathcal N}_N(\varphi,\varphi^\prime) =\sum_{j=1}^N \mathds{1}_{\varphi \le \theta_j \le \varphi^\prime}
\end{eqnarray}
which equals the number of eigenvalues in the spectral interval $(\varphi,\varphi^\prime)$. Here, ${\mathds 1}_X$ is the indicator function; it equals unity if $X$ is true; otherwise it brings zero.
\hfill $\blacksquare$
\end{remark}

\subsection{Main theorem for the power spectrum} \label{power-spectrum-phi}

The theorem below does {\it not} assume the stationarity of eigenlevel spacings and does {\it not} require a random-matrix-theory paradigm for its proof.

\begin{theorem}\label{main-technical-theorem}
  Let $\{0\le \theta_1 \le \dots \le \theta_N < 2\pi\}$ be a sequence of $N \in {\mathbb N}$ {\it ordered} eigenlevels confined to the unit circle. The ensemble averaged power spectrum of eigenlevels (Definition~\ref{d-01}) admits the representation
\begin{eqnarray} \fl \label{ps-main}
    S_N(\omega) = \frac{z}{N\Delta_N^2} \left(\frac{2\pi}{1-z}\right)^2 \Bigg\{
    \left| \int_{0}^{2\pi} \frac{d\varphi}{2\pi}
      \left[
        z^{-N}\Phi_N((0,\varphi);1-z) - 1
    \right] \right|^2 \nonumber\\
    \qquad\qquad - 2 {\rm Re}
    \Bigg(
        \int_{0}^{2\pi} \frac{d\varphi^\prime}{2\pi} \int_{0}^{\varphi^\prime} \frac{d\varphi}{2\pi} \,
         \left[ \Phi_N\left( (\varphi,\varphi^\prime); 1-z   \right) -1\right] \nonumber\\
                   \qquad\qquad\qquad -  \int_{0}^{2\pi} \frac{d\varphi}{2\pi} \,\left[ z^{-N} \Phi_N\left( (0,\varphi); 1-z   \right) - 1 \right]
    \Bigg)
    \Bigg\}.
\end{eqnarray}
Here, $\Phi_N\left( (\varphi,\varphi^\prime); 1-z \right)$ is the generating function [Eqs.~(\ref{GF0-d}) and (\ref{mgf})] of the probabilities [Eq.~(\ref{ELN})] to find a given number of eigenlevels in the interval $(\varphi,\varphi^\prime)$; parameter $z$ is set to $z=e^{i\omega}$, and $\,0 < \omega \le \pi$.
\end{theorem}
\begin{proof}
Since the covariance matrix $\langle\delta\theta_\ell \delta\theta_m \rangle$ equals $\langle\theta_\ell \theta_m \rangle - \langle\theta_\ell\rangle \langle \theta_m \rangle$, we rewrite the power spectrum Eq.~(\ref{ps-def}) in the form
  \begin{eqnarray} \label{SN-S1-S2}
    S_N(\omega) = \frac{1}{N \Delta_N^2} \left( \mathcal{S}_2(z) - \left|
        \mathcal{S}_1(z)
    \right|^2\right)
  \end{eqnarray}
which contains two auxiliary generating functions
\begin{eqnarray}
\label{S1-def}
    \mathcal{S}_1(z) &=& \sum_{\ell=1}^{N} \langle\theta_\ell\rangle z^\ell, \\
    \label{S2-def}
    \mathcal{S}_2(z) &=& \sum_{\ell=1}^{N} \sum_{m=1}^{N}\langle\theta_\ell \theta_m\rangle z^{\ell-m}.
  \end{eqnarray}

(i) First, let us focus on the generating function $\mathcal{S}_1(z)$. Making use of Lemma~\ref{Lemma-3GF}, see Eqs.~(\ref{GF1-d}) and (\ref{GF1-rel}), the function $\mathcal{S}_1(z)$ can be transformed as follows:
\begin{eqnarray} \label{S1-aux-1}
\mathcal{S}_1(z) &=& \int_{0}^{2\pi} \frac{d\varphi}{2\pi} \, \varphi
    \, \Phi_N^{(1)}(\varphi;1-z) \nonumber\\
    &=& -z \frac{2\pi}{1-z}\int_{0}^{2\pi} \frac{d\varphi}{2\pi} \, \varphi
    \frac{d}{d\varphi} \left(
        \Phi_N((0,\varphi);1-z) - z^N
    \right).
\end{eqnarray}
Integrating by parts and taking into account that $\Phi_N((0,2\pi);1-z)=z^N$, we derive:
\begin{eqnarray} \label{S1-Phi}
    \mathcal{S}_1(z) = \frac{z}{1-z}\int_{0}^{2\pi} d\varphi
      \left(
        \Phi_N((0,\varphi);1-z) - z^N
    \right).
\end{eqnarray}

(ii) To deal with $\mathcal{S}_2(z)$, we divide it into three parts:
\begin{eqnarray} \label{S2-aux-1}
    \mathcal{S}_2(z) = \int_{0}^{2\pi} \frac{d\varphi}{2\pi} \varphi^2 \varrho_N(\varphi) + Q_N(z) + Q_N(z^{-1}),
\end{eqnarray}
where $\varrho_N(\theta) = \sum_{\ell=1}^{N}\langle \delta(\theta-\theta_\ell)\rangle$ is the mean density of eigenlevels whilst
\begin{eqnarray} \fl \label{QN-aux-1}
    \qquad
   Q_N(z) = \sum_{\ell < m}^N \langle\theta_\ell \theta_m\rangle z^{m-\ell} = \int_{0}^{2\pi}\int_{0}^{2\pi} \frac{d\varphi d\varphi^\prime}{(2\pi)^2}
   \varphi\varphi^\prime \, \Phi_N^{(2)}(\varphi,\varphi^\prime;1-z).
\end{eqnarray}
This is a consequence of Eq.~(\ref{GF2-d}). Notice that the function $\Phi_N^{(2)}(\varphi,\varphi^\prime;1-z)$ vanishes for $\varphi>\varphi^\prime$. To handle the double integral in Eq.~(\ref{QN-aux-1}), we apply the identity Eq.~(\ref{GF2-rel}) to obtain
\begin{eqnarray} \label{QN-aux-2}
   Q_N(z) = -\frac{z}{(1-z)^2} \int_{0}^{2\pi} d\varphi^\prime\, \varphi^\prime\, \mathcal{I}_N(\varphi^\prime),
\end{eqnarray}
where
\begin{eqnarray} \label{QN-aux-1-IN}
            \mathcal{I}_N(\varphi^\prime)=\int_{0}^{\varphi^\prime} d\varphi\, \varphi
   \frac{\partial^2}{\partial\varphi \partial\varphi^\prime} \Phi_N((\varphi,\varphi^\prime); 1-z).
\end{eqnarray}
Integration by parts reduces it to
\begin{eqnarray} \fl \label{inner-int} \quad
    \mathcal{I}_N(\varphi^\prime)
       = \varphi^\prime \left(\frac{\partial}{\partial\varphi^\prime} \Phi_N((\varphi,\varphi^\prime); 1-z)\right) \Bigg|_{\varphi = \varphi^\prime}
       - \int_{0}^{\varphi^\prime} d\varphi\, \left(\frac{\partial}{\partial\varphi^\prime} \Phi_N((\varphi,\varphi^\prime); 1-z)\right) \nonumber\\
{}
\end{eqnarray}
and further to
\begin{eqnarray} \fl \qquad \label{QN-aux-5}
    \mathcal{I}_N(\varphi^\prime)= 1 - \frac{1-z}{2\pi} \varphi^\prime\, \varrho_N(\varphi^\prime) -
    \frac{\partial}{\partial\varphi^\prime}\left( \int_{0}^{\varphi^\prime} d\varphi\,  \Phi_N((\varphi,\varphi^\prime); 1-z)\right)
\end{eqnarray}
after spotting the two identities:
\begin{eqnarray} \label{QN-aux-3}
    \frac{\partial}{\partial\varphi^\prime} \Phi_N((\varphi,\varphi^\prime); 1-z)\Bigg|_{\varphi = \varphi^\prime} = -\frac{1-z}{2\pi} \varrho_N(\varphi^\prime),
\end{eqnarray}
see Eq.~(\ref{GF0-rel}), and
\begin{eqnarray} \label{QN-aux-4} \fl \quad
   \int_{0}^{\varphi^\prime} d\varphi \left( \frac{\partial}{\partial\varphi^\prime} \Phi_N((\varphi,\varphi^\prime); 1-z) \right) =
    \frac{\partial}{\partial\varphi^\prime}\left( \int_{0}^{\varphi^\prime} d\varphi\,  \Phi_N((\varphi,\varphi^\prime); 1-z)\right)
    - 1.
\end{eqnarray}
In order to determine the function $Q_N(z)$, given by Eq.~(\ref{QN-aux-2}), we now focus on the integral
\begin{eqnarray}  \label{QN-aux-6} \fl \qquad
    \int_{0}^{2\pi} d\varphi^\prime \,\varphi^\prime \, \mathcal{I}_N(\varphi^\prime) = 2\pi^2 -(1-z) \int_{0}^{2\pi} \frac{d\varphi^\prime}{2\pi}
    \, {\varphi^\prime}^2 \varrho_N(\varphi^\prime) \nonumber\\
    \qquad
    -
    \int_{0}^{2\pi} d\varphi^\prime \,\varphi^\prime \,
    \frac{\partial}{\partial\varphi^\prime}\left( \int_{0}^{\varphi^\prime} d\varphi\,  \Phi_N((\varphi,\varphi^\prime); 1-z)\right),
\end{eqnarray}
see Eq.~(\ref{QN-aux-5}). To handle the double integral therein, we calculate it by parts
\begin{eqnarray} \label{QN-aux-7}\fl
   \int_{0}^{2\pi} d\varphi^\prime \varphi^\prime \frac{\partial}{\partial\varphi^\prime}\left( \int_{0}^{\varphi^\prime} d\varphi  \Phi_N((\varphi,\varphi^\prime); 1-z)\right)
     = 2\pi \,  \int_{0}^{2\pi} d\varphi\,  \Phi_N((\varphi,2\pi); 1-z) \nonumber\\
     \qquad
     - \int_{0}^{2\pi} d\varphi^\prime \left(
        \int_{0}^{\varphi^\prime} d\varphi \, \Phi_N((\varphi,\varphi^\prime); 1-z)
    \right)
\end{eqnarray}
and consult Eq.~(\ref{comp-PhiN}) to observe the identity
\begin{eqnarray} \label{QN-aux-8}
    \Phi_N((\varphi,2\pi); 1-z) = z^N \Phi_N\left((0,\varphi); 1-z^{-1}\right).
\end{eqnarray}
This brings:
\begin{eqnarray} \fl  \label{outer-int}
  \int_{0}^{2\pi} d\varphi^\prime \,\varphi^\prime \, \mathcal{I}_N(\varphi^\prime) =  2\pi^2 -(1-z) \int_{0}^{2\pi} \frac{d\varphi^\prime}{2\pi}
    \, {\varphi^\prime}^2 \varrho_N(\varphi^\prime) \nonumber\\
    \fl \quad
    - 2\pi z^N \int_{0}^{2\pi} d\varphi \, \Phi_N((0,\varphi); 1-z^{-1}) +
      \int_{0}^{2\pi} d\varphi^\prime \left(
        \int_{0}^{\varphi^\prime} d\varphi \, \Phi_N((\varphi,\varphi^\prime); 1-z)\right).
\end{eqnarray}
Combining this result with Eq.~(\ref{QN-aux-2}), we derive:
\begin{eqnarray}  \label{QN-final} \fl \qquad
    Q_N(z) =  \frac{z}{1-z} \int_{0}^{2\pi} \frac{d\varphi}{2\pi}
    \, \varphi^2 \varrho_N(\varphi) \nonumber\\
    \fl \qquad\qquad\qquad
    +\, z \left( \frac{2\pi}{1-z}  \right)^2 \Bigg\{
    \int_{0}^{2\pi} \frac{d\varphi}{2\pi}\,  \left[ z^N\Phi_N((0,\varphi); 1-z^{-1})-1\right] \nonumber\\
    \qquad
    - \int_{0}^{2\pi} \frac{d\varphi^\prime}{2\pi} \left(
        \int_{0}^{\varphi^\prime} \frac{d\varphi}{2\pi} \, \left[ \Phi_N((\varphi,\varphi^\prime); 1-z)-1\right]
    \right)\Bigg\}.
\end{eqnarray}
Substitution to Eq.~(\ref{S2-aux-1}) yields
\begin{eqnarray} \fl \quad \label{S2-final}
    \mathcal{S}_2(z) = z \left(
    \frac{2\pi}{1-z}
    \right)^2 \Bigg\{ \int_{0}^{2\pi} \frac{d\varphi}{2\pi} \, \left[  z^N\Phi_N((0,\varphi); 1-z^{-1})-1\right]  \nonumber\\
    \qquad\quad
    +   \int_{0}^{2\pi} \frac{d\varphi}{2\pi} \, \left[ z^{-N}\Phi_N((0,\varphi); 1-z)-1\right] \nonumber\\
    \qquad\quad
    - \int_{0}^{2\pi} \frac{d\varphi^\prime}{2\pi} \left(
        \int_{0}^{\varphi^\prime} \frac{d\varphi}{2\pi} \, \big[ \Phi_N((\varphi,\varphi^\prime); 1-z)-1\big]
    \right)\nonumber\\
    \qquad\quad
    - \int_{0}^{2\pi} \frac{d\varphi^\prime}{2\pi} \left(
        \int_{0}^{\varphi^\prime} \frac{d\varphi}{2\pi} \, \big[ \Phi_N((\varphi,\varphi^\prime); 1-z^{-1})-1\big]
    \right)\Bigg\}.
\end{eqnarray}
Confining $z$ to the unit circle $|z|=1$, the above reduces to
\begin{eqnarray} \fl \quad \label{S2-final-circle}
    \mathcal{S}_2(z) = 2 z \left(
    \frac{2\pi}{1-z}
    \right)^2 {\rm Re\,} \Bigg\{  \int_{0}^{2\pi} \frac{d\varphi}{2\pi} \,\left[ z^{-N} \Phi_N((0,\varphi); 1-z)-1\right] \nonumber\\
    \qquad\quad
    - \int_{0}^{2\pi} \frac{d\varphi^\prime}{2\pi} \left(
        \int_{0}^{\varphi^\prime} \frac{d\varphi}{2\pi} \, \big[ \Phi_N((\varphi,\varphi^\prime); 1-z) -1 \big]
    \right)\Bigg\}.
\end{eqnarray}
Substitution of Eqs.~(\ref{S1-Phi}) and (\ref{S2-final-circle}) into Eq.~(\ref{SN-S1-S2}) completes the proof.
\end{proof}

\section{Power spectrum for ${\rm CUE}(N)$} \label{power-spectrum-CUE}

\subsection{Basic definitions}
 Below we apply Theorem~\ref{main-technical-theorem} to describe the power spectra for the circular unitary ensemble ${\rm CUE}(N)$ defined by the JPDF of the eigen-angles
 $\{\theta_j\}_{j=1}^N \in [0,2\pi)$ in the form~\cite{M-2004}
 \begin{eqnarray}\label{jpdf-cue}
    P_N(\theta_1,\dots,\theta_N) = \frac{1}{N!} \prod_{1\le j<k \le N} \left|
        e^{i \theta_j} -  e^{i \theta_k}
    \right|^2.
\end{eqnarray}
This JPDF is invariant under simultaneous, uniform translation of all eigen-angles.

(i) The ${\rm CUE}(N)$ eigen-angles form a determinantal point process since their $\ell$-point correlation function
\begin{eqnarray} \label{RLN-CUE}
    R_{\ell,N}(\theta_1,\dots,\theta_\ell) = \frac{N!}{(N-\ell)!} \left(\prod_{j=\ell+1}^{N} \int_{0}^{2\pi} \frac{d\theta_j}{2\pi}
    \right) \, P_N(\theta_1,\dots,\theta_N)
\end{eqnarray}
admits the determinantal representation~\cite{M-2004}
\begin{eqnarray}\label{RLN-det}
    R_{\ell,N} (\theta_1,\dots,\theta_\ell) = {\det}_{1\le j,k \le \ell} \left[
        K_{N} (\theta_j -\theta_k)
    \right],
\end{eqnarray}
where
\begin{eqnarray} \label{CUEN-Kernel}
    K_{N}(\theta) = \frac{\sin[N(\theta/2)]}{\sin(\theta/2)}
\end{eqnarray}
is the ${\rm CUE}(N)$ scalar kernel satisfying the reproducing property
\begin{eqnarray} \label{rep-prop-CUE}
    \int_{0}^{2\pi} \frac{d\phi}{2\pi} K_N(\theta-\phi)\,
    K_N(\phi-\theta^\prime) = K_N(\theta -\theta^\prime).
\end{eqnarray}

(ii) The moment generating function [Eq.~(\ref{GF0-d})] of the eigen-angle counting function ${\mathcal N}_N(0,\varphi)$ defined in Eq.~(\ref{NCF}) is given by the multiple integral
\begin{eqnarray} \fl\quad \label{mgf-CUEN}
    \Phi_N((0,\varphi); 1-z) = \frac{1}{N!}
    \left( \int_{0}^{2\pi} -(1-z)  \int_{0}^{\varphi}
    \right)^{N}  \prod_{j=1}^{N} \frac{d\theta_j}{2\pi}\, \prod_{1\le k<\ell \le N} \left|
        e^{i \theta_k} -  e^{i \theta_\ell}
    \right|^2,
\end{eqnarray}
see Lemma~\ref{Lemma-3GF}.

\subsection{Master formula}
\begin{theorem}\label{circular-theorem}
  Let $\{0\le \theta_1 \le \dots \le \theta_N < 2\pi\}$ be a sequence of $N \in {\mathbb N}$ {\it ordered} fluctuating eigen-angles drawn from the ${\rm CUE}(N)$. For all $\,0 < \omega \le \pi$, the ensemble averaged power spectrum of eigen-angles admits the representation
\begin{eqnarray} \fl \label{ps-circular}
\qquad
    S_N(\omega) = \frac {z N}{(1-z)^2} \Bigg\{
    \left| \int_{0}^{2\pi} \frac{d\varphi}{2\pi}
      \, \Phi_N((0,\varphi);1-z) \right|^2 \nonumber\\
    \qquad\qquad
    - 2 {\rm Re}\,
        \int_{0}^{2\pi} \frac{d\varphi}{2\pi} \left( 1 -\frac{\varphi}{2\pi}\right) \Phi_N\left( (0,\varphi); 1-z   \right)
    \Bigg\}.
\end{eqnarray}
Here, $\Phi_N\left( (0, \varphi); 1-z \right)$ is the generating function [Eq.~(\ref{GF0-d})] of the probabilities [Eq.~(\ref{ELN})] to find a given number of eigen-angles in the interval $(0,\varphi)$ and the parameter $z$ is set to $z=e^{i\omega}$.
\end{theorem}
\begin{proof}
Invariance of the JPDF Eq.~(\ref{jpdf-cue}) under uniform translation $\{\theta_j \mapsto \theta_j-\delta\}_{j=1}^N$ ensures the property
\begin{eqnarray}\label{Phi-N-uniform}
    \Phi_N\left( (\varphi, \varphi^\prime); 1-z \right) = \Phi_N\left( (0, \varphi^\prime -\varphi); 1-z \right),
\end{eqnarray}
see Eq.~(\ref{GF0-rel}). Substituting Eq.~(\ref{Phi-N-uniform}) into Eq.~(\ref{ps-main}) and setting $\Delta_N = 2\pi/N$, we reproduce, after some algebra, the desired Eq.~(\ref{ps-circular}).

\end{proof}

\begin{remark}\label{rem-symmetry}
Spectral correlation functions in the ${\rm CUE}(N)$ posses the mirror symmetry: fluctuations of eigen-angles on two arcs $[0,\pi)$ and $[\pi,2\pi)$ are related to each other
\begin{eqnarray}
    R_{\ell,N} (\theta_1,\dots,\theta_\ell) = R_{\ell,N} (2\pi-\theta_1,\dots,2\pi-\theta_\ell).
\end{eqnarray}
It is beneficial to explicitly unveil this symmetry in the master formula Eq.~(\ref{ps-circular}) by reducing integration domains therein to $[0,\pi]$.
\hfill $\blacksquare$
\end{remark}

\begin{corollary} \label{ps-master-symmetry}
In the notation of Theorem~\ref{circular-theorem}, the ${\rm CUE}(N)$ power spectrum admits the representation
 \begin{eqnarray} \fl \qquad \label{ps-lemma}
    S_N(\omega) = \frac {2 z}{(1-z)^2} \Bigg\{
                - {\rm Re}\,\Big(
                I_{N,0}(z)
                - (1-z^{-N})I_{N,1}(z)
                \Big)
                  \nonumber\\
                \fl \qquad\qquad\qquad\qquad\qquad\qquad \qquad
                + \frac{2}{N} \left({\rm Re}\,\big[ z^{-N/2}I_{N,0}(z)\big]\right)^2
                 \Bigg\},
\end{eqnarray}
where
\begin{eqnarray} \label{ps-alter-master}
    I_{N,q}(z) = N \int_{0}^{\pi} \frac{d\varphi}{2\pi} \left(
    \frac{\varphi}{2\pi}
    \right)^q \Phi_N\left( (0,\varphi); 1-z   \right)
\end{eqnarray}
with $q=0,1$.
\end{corollary}
\begin{proof}
To work out the mirror symmetry explicitly, we split the first integral in Eq.~(\ref{ps-circular}) into two
\begin{eqnarray} \fl \quad
    \int_{0}^{2\pi}\frac{d\varphi}{2\pi} \Phi_N((0,\varphi); 1-z) =
    \int_{0}^{\pi}\frac{d\varphi}{2\pi} \Phi_N((0,\varphi); 1-z) +
    \int_{\pi}^{2\pi}\frac{d\varphi}{2\pi} \Phi_N((0,\varphi); 1-z) \nonumber\\
    {}
\end{eqnarray}
and further change the integration variable in the second integral, $\varphi^\prime=2\pi-\varphi$, thus reducing the above to
\begin{eqnarray}
        \int_{0}^{\pi}\frac{d\varphi}{2\pi} \Phi_N((0,\varphi); 1-z) +
    \int_{0}^{\pi}\frac{d\varphi}{2\pi} \Phi_N((0,2\pi-\varphi); 1-z).
\end{eqnarray}
Making use of the symmetry relation [see Eqs.~(\ref{QN-aux-8}) and (\ref{Phi-N-uniform})]
\begin{eqnarray}  \label{phin-sym-rel}
   \Phi_{N}((0,2\pi-\varphi);1-z) = z^{N}\, \overline{\Phi_{N}((0,\varphi); 1-z)},
\end{eqnarray}
where $|z|=1$, we derive:
\begin{eqnarray} \fl \qquad\qquad\quad \label{int-1-sym}
    N \int_{0}^{2\pi}\frac{d\varphi}{2\pi} \Phi_N((0,\varphi); 1-z) =I_{N,0}(z) + z^N \overline{I_{N,0}(z)}.
\end{eqnarray}
The same strategy, applied to the second integral in Eq.~(\ref{ps-circular}), yields
\begin{eqnarray} \fl \qquad \label{int-2-sym}
    N \int_{0}^{2\pi}\frac{d\varphi}{2\pi} \left( 1 - \frac{\varphi}{2\pi}\right) \Phi_N((0,\varphi); 1-z)
    =I_{N,0}(z) - I_{N,1}(z) + z^N \,\overline{I_{N,1}(z)}.
\end{eqnarray}
Substitution of Eqs.~(\ref{int-1-sym}) and (\ref{int-2-sym}) into Eq.~(\ref{ps-circular}) completes the proof.
\end{proof}

\begin{remark}
  Theorem~\ref{circular-theorem} and Corollary~\ref{ps-master-symmetry} obviously hold for the circular orthogonal ${\rm COE}(N)$ and circular symplectic ${\rm CSE}(N)$ ensembles, with the generating function $\Phi_N((0,\varphi); 1-z)$ adjusted appropriately.
\hfill $\blacksquare$
\end{remark}

\subsection{Painlev\'e VI representation} \label{painleve-vi}

\begin{theorem}\label{Th-PVI}
    Let $\{0\le \theta_1 \le \cdots \le \theta_N < 2\pi\}$ be a sequence of $N \in {\mathbb N}$ {\it ordered} fluctuating eigen-angles drawn from the ${\rm CUE}(N)$. For all $\,0 < \omega \le \pi$, the ensemble averaged power spectrum of eigen-angles admits the exact representation Eq.~(\ref{ps-circular}) of Theorem~\ref{circular-theorem}, where
\begin{equation} \label{phin-cue-painleve-6}
    \Phi_N((0,\varphi);\zeta) = \exp \left(
            -\int_{\cot(\varphi/2)}^{\infty} \frac{dt}{1+t^2} {\sigma}_N(t;\zeta)\right).
\end{equation}
Here, $\zeta=1-z$ and $z=e^{i\omega}$. The Painlev\'e VI function ${\sigma}_N(t;\zeta)$ is a solution to the nonlinear equation
\begin{eqnarray} \label{pvi} \fl
    \qquad\qquad \left( (1+t^2)\,{\sigma}_N^{\prime\prime} \right)^2 + 4 {\sigma}_N^\prime ({\sigma}_N - t {\sigma}_N^\prime)^2
    + 4 ({\sigma}_N^\prime)^2 \left(
        {\sigma}_N^\prime + N^2
    \right) = 0
\end{eqnarray}
satisfying the boundary condition
\begin{eqnarray} \label{pvi-bc}
\fl \qquad\qquad
    {\sigma}_N(t;\zeta) = -\frac{N\zeta}{2\pi} - \frac{N^2 \zeta^2}{2\pi^2 t}  +  \frac{N\zeta}{12\pi^3 t^2}
    \left(
        N^2(\pi^2-6\zeta^2) -\pi^2
    \right)
    +
    {\mathcal O}(t^{-3})
\end{eqnarray}
as $t\rightarrow \infty$.
\end{theorem}

\begin{proof}
By virtue of Eqs.~(\ref{GF0-rel}) and (\ref{jpdf-cue}), the generating function $\Phi_N((0,\varphi);\zeta)$ in Theorem~\ref{circular-theorem} admits a multidimensional-integral representation
\begin{eqnarray} \fl\qquad \label{Phin-cue}
    \Phi_N((0,\varphi);\zeta) = \frac{1}{N!}
    \prod_{j=1}^N \left( \int_0^{2\pi} - \zeta \int_0^\varphi \right) \frac{d\theta_j}{2\pi}
    \prod_{1 \le j < k \le N}^{}
    \left| e^{i\theta_j} - e^{i\theta_k} \right|^2.
\end{eqnarray}
To express it in terms of a Painlev\'e VI function, we follow Ref.~\cite{FW-2004} to rewrite Eq.~(\ref{Phin-cue}) as an eigenvalue integral over the Cauchy measure. This is achieved through the change of variables
\begin{eqnarray}
    e^{i\theta_j}=\frac{i\lambda_j-1}{i\lambda_j+1}.
\end{eqnarray}
Thanks to the identities
\begin{eqnarray}
    \prod_{j=1}^N \frac{d\theta_j}{2\pi}=\frac{(-1)^N}{\pi^N}\prod_{j=1}^N\frac{d\lambda_j}{1+\lambda_j^2}
\end{eqnarray}
and
\begin{eqnarray} \fl \qquad
    \prod_{1\le j<k\le N}\left|\;e^{i\theta_j}-e^{i\theta_k}\right|^2
    =2^{N(N-1)} \prod_{j=1}^N\frac{1}{(1+\lambda_j^2)^{N-1}}
    \prod_{1 \le j < k \le N}^{}
    \left| \lambda_j - \lambda_k \right|^2,
\end{eqnarray}
we reduce Eq.~(\ref{Phin-cue}) to the multidimensional integral
\begin{eqnarray} \label{FNFF-cauchy} \fl\quad
\Phi_N((0,\varphi);\zeta)=
    \frac{2^{N(N-1)}}{\pi^N} \prod_{j=1}^N
        \left(\int_{-\infty}^{\infty}-\zeta\int_{\cot(\varphi/2)}^\infty \right)
            \frac{d\lambda_j}{(1+\lambda_j^2)^{N}}
                \prod_{1 \le j < k \le N}^{}
    \left| \lambda_j - \lambda_k \right|^2. \nonumber\\
    {}
\end{eqnarray}
Its Painlev\'e VI representation can be read off from Chapter 8, \S~8.3.1 of Ref.~\cite{PF-book} upon setting $\mu=0$, $\alpha=N$ and $s=\cot (\varphi/2)$ in Eq.~(8.71) therein. This results in Eqs.~(\ref{phin-cue-painleve-6}) and (\ref{pvi}). For a detailed derivation of the boundary condition Eq.~(\ref{pvi-bc}), the reader is referred to Appendix~\ref{A-1}.
\end{proof}

\subsection{Fredholm determinant representation}

\begin{proposition} \label{prop-fred}
    Let $\{0\le \theta_1 \le \cdots \le \theta_N < 2\pi\}$ be a sequence of $N \in {\mathbb N}$ {\it ordered} fluctuating eigen-angles drawn from the ${\rm CUE}(N)$. For all $\,0 < \omega \le \pi$, the ensemble averaged power spectrum of eigen-angles admits the exact representation Eq.~(\ref{ps-circular}) of Theorem~\ref{circular-theorem}, where
\begin{eqnarray} \label{Phi-FD}
    \Phi_N((0,\varphi);\zeta) = {\rm det} \big[
        \mathds{1} - \zeta  \hat{K}_N^{(0,\varphi)} \big].
\end{eqnarray}
Here, $\hat{K}_N^{(0,\varphi)}$ is an integral operator defined by
\begin{eqnarray} \label{Phi-FD-SN}
    \big[\hat{K}_N^{(0,\varphi)} f\big] (\theta_1) = \int_{0}^{\varphi} \frac{d\theta_2}{2\pi} K_N(\theta_1 - \theta_2) \, f(\theta_2),
\end{eqnarray}
and ${K}_{N}(\theta)$ is the ${\rm CUE}(N)$ two-point scalar kernel defined by Eq.~(\ref{CUEN-Kernel}), $\zeta=1-z$ with $z=e^{i\omega}$.
\end{proposition}

\begin{proof}
Combine the representation Eq.~(\ref{ps-tcue-2A}) with the determinantal formula Eq.~(\ref{RLN-det}) to obtain
\begin{eqnarray} \fl \qquad \label{FD-01}
    \Phi_N((0,\varphi);\zeta) = 1+ \sum_{\ell=1}^{N} \frac{(-\zeta)^\ell}{\ell!} \left(
    \prod_{j=1}^\ell \int_0^\varphi\frac{d\theta_j}{2\pi} \right)\,
    {\det}_{1\le i,j \le \ell} \left[
        K_{N} (\theta_i - \theta_j)
    \right],
\end{eqnarray}
where the scalar kernel $K_N(\theta)$ is given by Eq.~(\ref{CUEN-Kernel}). Equation~(\ref{FD-01}) coincides with
a definition of the Fredholm determinant Eq.~(\ref{Phi-FD}).
\end{proof}

\subsection{Toeplitz determinant representation}

To analyse the power spectrum in the limit $N \rightarrow \infty$, it is beneficial to represent the generating function $\Phi_N((0,\varphi);\zeta)$ entering the exact solution Eq.~(\ref{ps-circular}) in the form of a Toeplitz determinant.

\begin{lemma} \label{claeys-multiple-int}
Consider the Toeplitz determinant
\begin{eqnarray}\label{Tn} \fl \quad
    D_N(t;\alpha_1,\alpha_2;\beta_1,\beta_2) = {\rm det}_{0\le j,k \le N-1} \left(
        \int_{0}^{2\pi} \frac{d\theta}{2\pi} \, e^{i(j-k)\theta}f_t\left( e^{i\theta}; \alpha_1,\alpha_2;\beta_1,\beta_2\right)
    \right)
\end{eqnarray}
associated with the Fisher-Hartwig symbol~\cite{CK-2015}
 \begin{eqnarray} \label{ft-def}
    f_t(z; \alpha_1,\alpha_2;\beta_1,\beta_2) = z^{\beta_1+\beta_2} \prod_{\ell=1}^{2} |z-z_\ell|^{2\alpha_\ell} g_{z_\ell,\beta_\ell}(z) \, z_\ell^{-\beta_\ell},
 \end{eqnarray}
where $z=e^{i\theta}$, $\theta \in [0,2\pi)$, $z_1=e^{it}$, $z_2=e^{i(2\pi-t)}$, $0<t<\pi$, and
\begin{eqnarray} \label{g-symbol-def}
    g_{z_\ell, \beta_\ell} (z) = \left\{
                                         \begin{array}{ll}
                                           e^{ i\pi \beta_\ell}, & \hbox{$0 \le {\rm arg\,} z < {\rm arg\,} z_{\ell}$} \\
                                           e^{- i\pi \beta_\ell}, & \hbox{${\rm arg\,} z_{\ell} \le {\rm arg\,} z < 2\pi$}
                                         \end{array}
                                       \right.
\end{eqnarray}
The parameters $\alpha_1$ (at $z_1$) and $\alpha_2$ (at $z_2$) describe power- or root-type singularities; $\beta_1$ and $\beta_2$ describe jump
discontinuities. It is assumed that ${\rm Re\,}\alpha_\ell>-1/2$ and $\beta_\ell \in \mathbb{C}$ for $\ell=1,2$. The following multiple integral representation holds:
\begin{eqnarray} \fl \quad \label{m-stat-lemma}
        D_N(t;\alpha_1,\alpha_2;\beta_1,\beta_2) = e^{-2i t N \beta_1} e^{-i\pi N (\beta_1+\beta_2)} \nonumber\\
        \fl \qquad\qquad\qquad\quad
        \times
        \frac{1}{N!} \prod_{j=1}^{N}\left(
            \int_{0}^{2\pi} - \left(1-e^{2i\pi\beta_1}\right) \int_{0}^{2t}
        \right) \frac{d\theta_j}{2\pi}\, e^{i\theta_j(\beta_1+\beta_2)} \nonumber\\
        \fl \qquad\qquad\qquad\quad
        \times
         \left| e^{2it}-e^{i\theta_j} \right|^{2\alpha_1}  \left| 1-e^{i\theta_j} \right|^{2\alpha_2}
         \left|
         \Delta_N(e^{i\theta})
         \right|^2,
\end{eqnarray}
where
\begin{eqnarray}
    \Delta_N \big( e^{i\theta}\big) = \prod_{1\le j< k\le N} \big(  e^{i\theta_k} - e^{i\theta_j}\big)
\end{eqnarray}
is the Vandermonde determinant.
\end{lemma}
\begin{proof}
Disentangling the product $g_{z_1,\beta_1}(z)\,g_{z_2,\beta_2}(z)$ in Eq.~(\ref{ft-def}), we notice that the single integral in Eq.~(\ref{Tn}) equals
\begin{eqnarray} \fl\quad \label{SI-1}
        e^{-it(\beta_1-\beta_2)} e^{-2 i\pi\beta_2}
        \left( e^{i\pi (\beta_1+\beta_2)} \int_{0}^{t} + e^{i\pi (\beta_2 - \beta_1)} \int_{t}^{2\pi - t}
        + e^{-i\pi (\beta_1+\beta_2)} \int_{2\pi-t}^{2\pi} \right)\nonumber\\
        \fl \qquad \qquad \qquad
        \times
         \, \frac{d\theta}{2\pi}  e^{i(j-k)\theta}  e^{i\theta(\beta_1+\beta_2)}
         \big| e^{i\theta}-e^{it} \big|^{2\alpha_1}  \big| e^{i\theta} - e^{i(2\pi-t)} \big|^{2\alpha_2}.
\end{eqnarray}
Changing the integration variables $\theta^\prime = \theta+t$ in the first and second integral, and separately $\theta^\prime = \theta+t -2\pi$ in the third integral, we reduce Eq.~(\ref{SI-1}) to
\begin{eqnarray} \fl\quad \label{SI-2}
        e^{-2it\beta_1} e^{- i\pi(\beta_1+ \beta_2)} e^{-it(j-k)}
        \left( e^{2i\pi \beta_1} \int_{0}^{2t} + \int_{2t}^{2\pi}
        \right)\nonumber\\
        \fl \qquad \qquad \qquad
        \times
         \, \frac{d\theta^\prime}{2\pi}  e^{i(j-k)\theta^\prime}  e^{i\theta^\prime(\beta_1+\beta_2)}
         \big| e^{i\theta^\prime}-e^{2it} \big|^{2\alpha_1}  \big| 1 - e^{i\theta^\prime} \big|^{2\alpha_2}
\end{eqnarray}
and, after adding and subtracting an integral from $0$ to $2t$, further to
\begin{eqnarray} \fl\quad \label{SI-3}
        e^{-2it\beta_1} e^{- i\pi(\beta_1+ \beta_2)} e^{-it(j-k)}
        \left( \int_{0}^{2\pi} - \left( 1 -e^{2i\pi \beta_1}\right) \int_{0}^{2t}
        \right)\nonumber\\
        \fl \qquad \qquad \qquad
        \times
         \, \frac{d\theta^\prime}{2\pi}  e^{i(j-k)\theta^\prime}  e^{i\theta^\prime(\beta_1+\beta_2)}
         \big| e^{i\theta^\prime}-e^{2it} \big|^{2\alpha_1}  \big| 1 - e^{i\theta^\prime} \big|^{2\alpha_2}.
\end{eqnarray}
Substitution of Eq.~(\ref{SI-3}) back to Eq.~(\ref{Tn}) yields
\begin{eqnarray} \fl \quad \label{m-stat-lemma-00}
        D_N(t;\alpha_1,\alpha_2;\beta_1,\beta_2) = e^{-2i t N \beta_1} e^{-i\pi N (\beta_1+\beta_2)} \nonumber\\
        \fl \qquad\qquad\qquad
        \times {\rm det}_{0\le j,k\le N-1}
        \Bigg(
         \left( \int_{0}^{2\pi} - \left( 1 -e^{2i\pi \beta_1}\right) \int_{0}^{2t}\right) \frac{d\theta}{2\pi} \,
         \nonumber\\
        \fl \qquad \qquad\qquad\qquad\qquad
        \times
        e^{i(j-k)\theta}  e^{i\theta(\beta_1+\beta_2)}
         \big| e^{i\theta}-e^{2it} \big|^{2\alpha_1}  \big| 1 - e^{i\theta} \big|^{2\alpha_2}\Bigg).
\end{eqnarray}
Finally, applying Andr\'eief's formula \cite{A-1883,dB-1955},
\begin{eqnarray} \fl \quad \label{adbf}
    {\rm det}_{0\le j,k \le N-1}\left(
        \int_{{\mathcal L}} \frac{d\theta}{2\pi}\, w(\theta) e^{i(j-k)\theta}
    \right)
    = \frac{1}{N!}  \left( \prod_{j=1}^N \int_{{\mathcal L}}\frac{d\theta_j}{2\pi} w(\theta_j) \right) \big|
        \Delta_n(e^{i\theta})
    \big|^2
\end{eqnarray}
to Eq.~(\ref{m-stat-lemma-00}), we derive Eq.~(\ref{m-stat-lemma}).
\end{proof}

\begin{proposition}\label{toep-prop}
    Let $\{0\le \theta_1 \le \cdots \le \theta_N < 2\pi\}$ be a sequence of $N \in {\mathbb N}$ {\it ordered} fluctuating eigen-angles drawn from the ${\rm CUE}(N)$. For all $\,0 < \omega \le \pi$, the ensemble averaged power spectrum of eigen-angles admits the exact representation Eq.~(\ref{ps-circular}) of Theorem~\ref{circular-theorem}, where
\begin{eqnarray} \label{GF-toeplitz-1}
    \Phi_N((0,\varphi); \zeta) = e^{i\varphi \tilde{\omega} N} \, D_N(\varphi/2;0,0;\tilde\omega,-\tilde\omega).
\end{eqnarray}
Here, $\tilde{\omega} = \omega/2\pi$, and $D_N$ is the Toeplitz determinant defined in Eq.~(\ref{Tn}).
\end{proposition}
\begin{proof}
Comparing Eq.~(\ref{Phin-cue}) with Eq.~(\ref{m-stat-lemma}), we identify $t=\varphi/2$, $\alpha_1=\alpha_2=0$, $\beta_1= -\beta_2 = \tilde\omega$. This completes the proof.
\end{proof}

\section{Power spectrum for ${\rm CUE}(\infty)$}\label{CUE-INF}

In the limit $N \rightarrow \infty$, the exact solution for the ${\rm CUE}(N)$ power spectrum presented in Theorem~\ref{circular-theorem} should converge to a universal, parameter-free law. To determine it, we shall perform an asymptotic analysis of the exact solution Eqs.~(\ref{ps-circular}) with the generating function $\Phi_N((0,\varphi); 1-z)$ being represented as a Toeplitz determinant in Proposition \ref{toep-prop}.

As a matter of fact, an alternative representation of the power spectrum, given by Corollary~\ref{ps-master-symmetry}, is a better starting point for the asymptotic analysis as Eq.~(\ref{ps-lemma}) fully explores the mirror symmetry of the ${\rm CUE}(N)$ spectra as discussed in Remark~\ref{rem-symmetry}. We shall prove that, as $N\rightarrow\infty$, only the first term in Eq.~(\ref{ps-lemma}) will survive
\begin{eqnarray} \label{SN-anticipation} \fl\qquad
    S_\infty(\omega) = \lim_{N\rightarrow\infty} S_N(\omega) = -\frac{2z}{(1-z)^2}\,{\rm Re} \lim_{N\rightarrow\infty} N\int_{0}^{\pi} \frac{d\varphi}{2\pi} \,\Phi_N((0,\varphi);1-z),
\end{eqnarray}
where $0<\omega<\pi$. A rigorous proof of this statement, together with the explicit calculation of the limit in Eq.~(\ref{SN-anticipation}), is the main objective of this section.

\subsection{Uniform asymptotics of the Toeplitz determinant} \label{CK-section}

To perform the integrals Eq.~(\ref{ps-alter-master}) in Eq.~(\ref{ps-lemma}) in the limit $N \rightarrow \infty$, {\it uniform} asymptotics of the Toeplitz determinant in the r.h.s.~of~Eq.~(\ref{GF-toeplitz-1}) are required in the subtle regime of two merging singularities. In our case, the Fisher-Hartwig symbol associated with $D_N(\varphi/2; 0,0;\tilde\omega,-\tilde\omega)$, has two jump type singularities as $\alpha_1=\alpha_2=0$, see Proposition~\ref{toep-prop}. Relevant uniform asymptotics were studied in great detail by Claeys and Krasovsky \cite{CK-2015}. For the Toeplitz determinant in Eq.~(\ref{GF-toeplitz-1}), their Theorems~1.1, ~1.5 and~1.8 hold.
\begin{theorem}\label{thm:CK}
For $0\le \varphi <\varphi_0$, where $\varphi_0$ is sufficiently small \footnote[2]{In fact, here $\varphi_0 =2\pi -\epsilon$ with $\epsilon >0$.}, the following asymptotic expansion holds {\it uniformly} as $N \rightarrow \infty$
\begin{eqnarray} \fl \label{eq:Edge-1}
\quad
	\log D_N(\varphi/2;0,0;\tilde\omega,-\tilde\omega) \nonumber\\
\fl \qquad\qquad
     = - i N \tilde{\omega} \varphi - 2\tilde{\omega}^2 \log \left(
	\frac{\sin(\varphi/2)}{\varphi/2}\right)
	+ \int_{0}^{- i N \varphi} \frac{ds}{s}\, \sigma(s) + {\mathcal O}\left(N^{-1+ 2 \tilde{\omega}}\right)
\end{eqnarray}
so that
\begin{eqnarray} \label{eq:Edge-2} \fl\qquad
	\Phi_N((0,\varphi);\zeta) =  \left(
	\frac{\sin(\varphi/2)}{\varphi/2}
	\right)^{-2\tilde{\omega}^2} \exp\left(
	\int_{0}^{- i N \varphi} \frac{ds}{s}\, \sigma(s)
	\right) \left( 1 + {\mathcal O}(N^{-1+ 2 \tilde{\omega}}) \right).
\end{eqnarray}
Here, $\sigma(s)$ is the fifth Painlev\'e transcendent defined as the solution to the nonlinear equation
\begin{equation} \label{eq:PV-eq-Th}
	s^2 (\sigma^{\prime\prime})^2 = \left(\sigma - s \sigma^\prime + 2 (\sigma^\prime)^2 \right)^2 - 4 (\sigma^\prime)^4
\end{equation}
subject to the boundary conditions~\footnote[3]{In Theorem~1.1 of Ref.~\cite{CK-2015}, the boundary condition as $s\rightarrow - i0_+$ is of the form ${\mathcal O}(|s|\log|s|)$. For the particular Fisher-Hartwig symbol associated with $D_N(\varphi/2; 0,0;\tilde\omega,-\tilde\omega)$, the boundary condition at zero admits a power series expansion in $s$ whose coefficients can be calculated recursively, see Appendix~\ref{A-2}.}
\begin{eqnarray}\label{eq:bc-zero-Th} \fl \qquad\qquad
	\sigma(s)= -\frac{s}{2\pi}\big( 1 - e^{i \omega}\big) - \frac{s^2}{4\pi^2}\big( 1 - e^{i \omega}\big)^2 + {\mathcal O}\left(
	s^3
	\right)\quad {\rm as} \quad s\rightarrow - i0_+
\end{eqnarray}
and
\begin{eqnarray}\label{eq:bc-inf-Th}\fl \qquad\qquad
	\sigma(s) = -{\tilde \omega} s - 2{\tilde \omega}^2 + \frac{s\gamma(s)}{1+\gamma(s)} + {\mathcal O}\left(
	|s|^{-1+2{\tilde \omega}}
	\right)
 \quad {\rm as} \quad  s\rightarrow - i\infty.
\end{eqnarray}
The function $\gamma(s)$ in Eq.~(\ref{eq:bc-inf-Th}) equals~\footnote[5]{Here, we have corrected a coefficient in the appearance of $\gamma(t)$ in Ref.~\cite{CK-2015}.
Equation~(\ref{eq:gamma-Th}) is compatible with an asymptotic expansion in Ref.~\cite{MCT-1986}. Numerical calculations provide independent support to this result.}
\begin{eqnarray}\label{eq:gamma-Th}
	\gamma(s) = \left|
	s
	\right|^{2(-1+2\tilde{\omega})} e^{-i |s|} \frac{\Gamma^2(1-\tilde{\omega}) }{\Gamma^2(\tilde{\omega})}.
\end{eqnarray}
The above holds for $0 \le \tilde{\omega} < 1/2$.
\end{theorem}

\begin{remark}\label{Rem-no-poles}
Following Ref.~\cite{CK-2015}, we notice that in Eqs.~(\ref{eq:Edge-1}) and (\ref{eq:Edge-2}) the path of integration in the complex $s$-plane should be chosen to avoid a finite number of poles $\{ s_j \in - i {\mathbb R}_+ \}$ of $\sigma(s)$ corresponding to zeros $\{ \varphi_j = i s_j/N \}$ in the asymptotics of the Toeplitz determinant $D_N(\varphi/2;0,0;\tilde\omega,-\tilde\omega)$. In Lemma~\ref{Lemma-zeros} and Remark~\ref{Remark-zeros} we show that $\{ s_j \}$ is the empty set provided $0\le \tilde\omega <1/2$. At $\tilde\omega=1/2$, when Theorem~\ref{thm:CK} does not hold, numerical analysis reveals that there are poles in $\sigma(s)$.
\hfill $\blacksquare$
\end{remark}

\begin{remark}
The following global integral condition for $\sigma(s)$ holds, see Ref.~\cite{CK-2015}:
\begin{eqnarray} \label{eq:global}
	\lim_{T\rightarrow +\infty} \left(
	\int_{0}^{-i T} \frac{ds}{s} \, \sigma(s) - i\tilde{\omega} T +2\tilde{\omega}^2 \log T
	\right) = \log G_{\tilde{\omega}},
\end{eqnarray}
where
\begin{eqnarray}
    G_{\tilde\omega} = G^2(1+\tilde\omega)G^2(1-\tilde\omega)
\end{eqnarray}
with $G(\dots)$ being the Barnes' $G$-function.
\hfill $\blacksquare$
\end{remark}

\subsection{Proof of Theorem~\ref{cue-pv}}\label{cue-pv-proof}

We start with two technical Lemmas.

\begin{lemma}\label{lemma-tech}
Let the function
\begin{eqnarray} \label{eq:f}
	\mathcal{F}_N(\varphi;z) =\left[1- \left(1-z^{-N}\right) \frac{\varphi}{2\pi}\right] \left(\frac{\sin(\varphi/2)}{\varphi/2}\right)^{-2\tilde{\omega}^2}
\end{eqnarray}
be defined on the unit circle $z=e^{2 i \pi \tilde{\omega}}$, and let $\sigma(s)$ be the Painlev\'e V transcendent specified in Theorem~\ref{thm:CK}. As $N\rightarrow\infty$, the following estimate holds
\begin{eqnarray}\fl \label{est-2}
 \qquad N  \int_{0}^\pi d\varphi\, \left\{	1- \mathcal{F}_N(\varphi;z)\right\} \exp\left(\int_0^{-i N\varphi} \frac{ds}{s} \sigma(s)\right)
    = {\mathcal O}(N^{-2\tilde{\omega}^2})
\end{eqnarray}
for $\tilde\omega \in (0, 1/2)$.
\end{lemma}
\begin{proof}
To estimate Eq.~(\ref{est-2}), we split the integral therein into two parts
\begin{eqnarray}\fl \label{eq:R1}
\qquad \mathcal{E}_N^{(1)}(z) = N \int_{0}^{\Omega(N)/N} d\varphi\,
\left\{	1- \mathcal{F}_N(\varphi;z)\right\} \exp\left(\int_0^{-i N\varphi} \frac{ds}{s} \sigma(s)\right)
\end{eqnarray}
and
\begin{eqnarray}\fl \label{eq:R2}
\qquad	\mathcal{E}_N^{(2)}(z) =
N \int_{\Omega(N)/N}^\pi d\varphi\,
\left\{	1- \mathcal{F}_N(\varphi;z)\right\} \exp\left(\int_0^{-i N\varphi} \frac{ds}{s} \sigma(s)\right),
\end{eqnarray}
where $\Omega(x)$ is any positive, smooth for large $x$ function such that $\Omega(N) \rightarrow \infty$ whilst $\Omega(N)/N \rightarrow 0$
as $N \rightarrow \infty$.

(i) To treat the integral Eq.~(\ref{eq:R1}), we change the integration variable $\varphi=\lambda/N$,
\begin{eqnarray}\fl \label{eq:R1-2}
	\qquad \mathcal{E}_N^{(1)}(z) = \int_{0}^{\Omega(N)} d\lambda\,
\left\{	1- \mathcal{F}_N\left(\lambda/N;z\right)\right\}
\exp\left(\int_0^{-i \lambda} \frac{ds}{s} \sigma(s)\right),
\end{eqnarray}
and make use of Eq.~(\ref{eq:global}) to realize that
 \begin{eqnarray} \label{eq:R1-inf}
 	 \exp\left(\int_0^{-i \lambda} \frac{ds}{s} \sigma(s)\right)=\mathcal{O}(\lambda^{-2\tilde\omega^2}) \quad {\rm as} \quad \lambda\rightarrow\infty.
 \end{eqnarray}
Equivalently,
   \begin{eqnarray} \label{eq:R1-inf-C}
  \sup_{\lambda\in [0,\infty)} \left| \lambda^{2\tilde{\omega}^2}	\exp\left(\int_0^{-i \lambda} \frac{ds}{s} \sigma(s)\right)\right|=C,
  \end{eqnarray}
where $C$ is a constant which depends on $\tilde{\omega}$ but obviously not on $N$. Next, since $\{1-\mathcal{F}_N(\lambda/N;z)\}=\mathcal{O}(\lambda/N)$ for $\lambda/N\rightarrow 0$, the integrand in Eq.~(\ref{eq:R1-2}) is uniformly bounded by $\mathcal{O}(\Omega(N)^{1 -2\tilde\omega^2}/N)$ for $\lambda\in [0,\Omega(N)]$. We then
conclude that
\begin{eqnarray}\label{EN1}
    \mathcal{E}_N^{(1)}(z) = \mathcal{O}\left( \frac{\Omega(N)^{2-2\tilde\omega^2}}{N} \right).
\end{eqnarray}

(ii) To treat the integral Eq.~(\ref{eq:R2}), we notice that $N \varphi \ge \Omega(N)$ within the integration domain. Since $\Omega(N) \rightarrow \infty$ as $N\rightarrow\infty$, we may use the asymptotic behavior Eq.~(\ref{eq:global}) to derive
\begin{eqnarray} \label{eq:R2-2} \fl
	\qquad	\mathcal{E}_N^{(2)}(z) = G_{\tilde\omega}N^{1-2\tilde\omega^2} \int_{\Omega(N)/N}^\pi d\varphi\,e^{i\tilde\omega  N\varphi} \frac{1- \mathcal{F}_N(\varphi;z)}{\varphi^{2\tilde\omega^2}}  \left(1+o(1) \right) \nonumber\\
 \fl \qquad \qquad \quad
 = \frac{G_{\tilde\omega}}{i\tilde\omega} N^{-2\tilde\omega^2} \Bigg[ \left. e^{i\tilde\omega  N\varphi} \frac{1- \mathcal{F}_N(\varphi;z)}{\varphi^{2\tilde\omega^2}} \right|_{\varphi=\Omega(N)/N}^{\varphi=\pi} \nonumber \\
  \fl \qquad\qquad\qquad \qquad\quad - \int_{\Omega(N)/N}^\pi d\varphi\, e^{i\tilde\omega  N\varphi} \frac{d}{d\varphi}
  \left(\frac{1- \mathcal{F}_N(\varphi;z)}{\varphi^{2\tilde\omega^2}}\right)\Bigg] \left(1+o(1) \right).
\end{eqnarray}
Here, we have used integration by parts in the second step.

As $N\rightarrow\infty$, the boundary term (in square brackets) at $\varphi=\pi$ yields a contribution of order $\mathcal{O}(1)$; yet, its contribution at $\varphi=\Omega(N)/N$ is of order
$$
    {\mathcal O}\left( \left( \frac{\Omega(N)}{N}\right)^{1-2\tilde\omega^2} \right)
$$
and can thus safely be neglected.

To handle the integral in the last line of Eq.~(\ref{eq:R2-2}), we split it into two
\begin{eqnarray}\fl \label{eq:R2-int}
\qquad \int_{\Omega(N)/N}^\pi d\varphi\, e^{i\tilde\omega  N\varphi} \frac{d}{d\varphi}
  \frac{1- \mathcal{F}_N(\varphi;z)}{\varphi^{2\tilde\omega^2}} \nonumber\\
  \fl \qquad\qquad = \frac{1-z^{-N}}{2\pi} \int_{\Omega(N)/N}^\pi d\varphi\, e^{i\tilde\omega  N\varphi} \frac{d}{d\varphi}
  \left(\frac{\varphi}{(2\sin(\varphi/2))^{2\tilde\omega^2}}\right)
  \nonumber\\
\fl \qquad \qquad \quad
+\int_{\Omega(N)/N}^\pi d\varphi\,
e^{i\tilde\omega  N\varphi} \frac{d}{d\varphi} \left(
            \frac{1}{\varphi^{2\tilde\omega^2}} -  \frac{1}{(2\sin(\varphi/2))^{2\tilde\omega^2}}
\right).
\end{eqnarray}
The former integral can be estimated as follows:
\begin{eqnarray}\fl \label{eq:R2-int1}
	\qquad\left| \int_{\Omega(N)/N}^\pi d\varphi\, e^{i\tilde\omega  N\varphi} \frac{d}{d\varphi}\!\left( \frac{\varphi}{(2\sin(\varphi/2))^{2\tilde\omega^2}}\right) \right| \nonumber\\
\le  \int_{0}^\pi d\varphi\, \left|\frac{d}{d\varphi}\!\left( \frac{\varphi}{(2\sin(\varphi/2))^{2\tilde\omega^2}}\right)\right|\nonumber\\
	=\left. \frac{\varphi}{(2\sin(\varphi/2))^{2\tilde\omega^2}}\right|_{\varphi=0}^{\varphi=\pi}= \frac{\pi}{2^{2\tilde\omega^2}}.
\end{eqnarray}
As for the latter, we have:
\begin{eqnarray}\fl \label{eq:R2-int2}
	\qquad\left| \int_{\Omega(N)/N}^\pi d\varphi\, e^{i\tilde\omega  N\varphi} \frac{d}{d\varphi}\!\left(\frac{1}{\varphi^{2\tilde\omega^2}}- \frac{1}{(2\sin(\varphi/2))^{2\tilde\omega^2}} \right)\right| \nonumber\\
	\le  \int_{0}^\pi d\varphi\, \left|\frac{d}{d\varphi}\!\left(\frac{1}{\varphi^{2\tilde\omega^2}}- \frac{1}{(2\sin(\varphi/2))^{2\tilde\omega^2}} \right)\right|\nonumber\\
	=\left.\frac{1}{(2\sin(\varphi/2))^{2\tilde\omega^2}} - \frac{1}{\varphi^{2\tilde\omega^2}} \right|_{\varphi=0}^{\varphi=\pi}= \frac{1}{2^{2\tilde\omega^2}}-\frac{1}{\pi^{2\tilde\omega^2}}.
\end{eqnarray}
Here, we have used the fact that the derivatives in both Eqs.~(\ref{eq:R2-int1}) and (\ref{eq:R2-int2}) are positive on $[0,\pi]$. Further, consulting with Eq.~(\ref{eq:R2-2}), we conclude that
\begin{eqnarray}\label{EN2}
\mathcal{E}_N^{(2)}(z) = \mathcal{O}(N^{-2\tilde\omega^2} ).
\end{eqnarray}
Choosing, e.g., $\Omega(x)=\log x$ and comparing Eqs.~(\ref{EN1}) and (\ref{EN2}), we conclude that the contribution ${\mathcal E}_N^{(1)}$ is subleading as compared to ${\mathcal E}_N^{(2)}$. This completes the proof.
\end{proof}

\begin{lemma}\label{lemma-tech-2}
  Let $\sigma(s)$ be the Painlev\'e V transcendent specified in Theorem~\ref{thm:CK}. As $N\rightarrow\infty$, the following estimate holds
  \begin{eqnarray} \label{IN0-lemma} \fl \qquad
    N \int_{0}^{\pi} d\varphi \left(\frac{\sin(\varphi/2)}{\varphi/2}\right)^{-2\tilde{\omega}^2}
     \exp\left(\int_0^{-i N\varphi} \frac{ds}{s} \sigma(s)\right) \nonumber\\
     \qquad = \int_{0}^{\infty} d\lambda
     \exp\left(\int_0^{-i \lambda} \frac{ds}{s} \sigma(s)\right) + {\mathcal O}\left(
     \Omega(N)^{-2\tilde\omega^2}
     \right)
\end{eqnarray}
in the domain $\tilde\omega \in (0, 1/2)$, in which the integral in the r.h.s. is convergent; $\Omega(x)$ is any positive, smooth for large $x$ function such that $\Omega(N) \rightarrow \infty$ whilst $\Omega(N)/N \rightarrow 0$
as $N \rightarrow \infty$.
\end{lemma}
\begin{proof}
Similarly to the proof of Lemma~\ref{lemma-tech}, we split the integral in Eq.~(\ref{IN0-lemma}) into two parts
\begin{eqnarray}\label{EN1-tech} \fl \qquad
    \tilde{\mathcal{E}}_N^{(1)}(z) =  N \int_{0}^{\Omega(N)/N} d\varphi \left(\frac{\sin(\varphi/2)}{\varphi/2}\right)^{-2\tilde{\omega}^2}
     \exp\left(\int_0^{-i N\varphi} \frac{ds}{s} \sigma(s)\right)
\end{eqnarray}
and
\begin{eqnarray}\label{EN2-tech}\fl \qquad
    \tilde{\mathcal{E}}_N^{(2)}(z) =  N \int_{\Omega(N)/N}^\pi d\varphi \left(\frac{\sin(\varphi/2)}{\varphi/2}\right)^{-2\tilde{\omega}^2}
     \exp\left(\int_0^{-i N\varphi} \frac{ds}{s} \sigma(s)\right).
\end{eqnarray}

(i) To estimate $\tilde{\mathcal{E}}_N^{(1)}(z)$, we change the integration variable $\varphi=\lambda/N$,
\begin{eqnarray}\label{EN1-tech-2} \fl \qquad
    \tilde{\mathcal{E}}_N^{(1)}(z) =  \int_{0}^{\Omega(N)} d\lambda \left(\frac{\sin(\lambda N/2)}{\lambda N/2}\right)^{-2\tilde{\omega}^2}
     \exp\left(\int_0^{-i \lambda} \frac{ds}{s} \sigma(s)\right),
\end{eqnarray}
and further observe that
\begin{eqnarray}
     \left(\frac{\sin(\lambda/2N)}{\lambda/2N}\right)^{-2\tilde{\omega}^2} = 1 +{\mathcal O}\big((\Omega(N)/N)^2\big)
\end{eqnarray}
uniformly for $\lambda \in [0,\Omega(N)]$. This yields
\begin{eqnarray}\label{EN1-tech-3} \fl \qquad
    \tilde{\mathcal{E}}_N^{(1)}(z) =  \int_{0}^{\Omega(N)} d\lambda
     \exp\left(\int_0^{-i \lambda} \frac{ds}{s} \sigma(s)\right) \left( 1 + {\mathcal O}\big((\Omega(N)/N)^2\big) \right).
\end{eqnarray}
As $N\rightarrow\infty$, the upper integration bound $\Omega(N)$ in Eq.~(\ref{EN1-tech-3}) extends to infinity. Since, in view of Eq.~(\ref{eq:global}), such an integral exists, we have
\begin{eqnarray}\label{EN1-tech-33} \fl \quad
    \tilde{\mathcal{E}}_N^{(1)}(z) =  \left\{\int_{0}^{\infty} d\lambda
     \exp\left(\int_0^{-i \lambda} \frac{ds}{s} \sigma(s)\right) - \int_{\Omega(N)}^\infty d\lambda
     \exp\left(\int_0^{-i \lambda} \frac{ds}{s} \sigma(s)\right) \right\} \nonumber\\
     \times\left( 1 + {\mathcal O}\big((\Omega(N)/N)^2\big) \right).
\end{eqnarray}
Since $\Omega(N) \rightarrow \infty$ as $N\rightarrow\infty$, we may use the asymptotic behavior of the second integral therein, see Eq.~(\ref{eq:global}),
\begin{eqnarray}
    \int_{\Omega(N)}^\infty d\lambda
     \exp\left(\int_0^{-i \lambda} \frac{ds}{s} \sigma(s)\right) = {\mathcal O}\left(
     \Omega(N)^{-2\tilde\omega^2}
     \right).
\end{eqnarray}
This leads us to conclude that
\begin{eqnarray} \label{EN1-tilde}
\tilde{\mathcal{E}}_N^{(1)}(z) = \int_{0}^{\infty} d\lambda
     \exp\left(\int_0^{-i \lambda} \frac{ds}{s} \sigma(s)\right) + {\mathcal O}\left(
     \Omega(N)^{-2\tilde\omega^2}
     \right)
\end{eqnarray}
as $N \rightarrow \infty$.

(ii) To treat the integral Eq.~(\ref{EN2-tech}), we notice that $N \varphi \ge \Omega(N)$ within the integration domain. Since $\Omega(N) \rightarrow \infty$ as $N\rightarrow\infty$, we may use the asymptotic behavior Eq.~(\ref{eq:global}) to derive
\begin{eqnarray} \label{EN2-tech-2} \fl
	\qquad	\tilde{\mathcal{E}}_N^{(2)}(z) = G_{\tilde\omega}N^{1-2\tilde\omega^2} \int_{\Omega(N)/N}^\pi d\varphi\,e^{i\tilde\omega  N\varphi}
            \frac{1}{(2\sin(\varphi/2))^{2\tilde\omega^2}}  \left(1+o(1) \right) \nonumber\\
 \fl \qquad \qquad \quad
 = \frac{G_{\tilde\omega}}{i\tilde\omega} N^{-2\tilde\omega^2} \Bigg[ \left. e^{i\tilde\omega  N\varphi} \frac{1}{(2\sin(\varphi/2))^{2\tilde\omega^2}} \right|_{\varphi=\Omega(N)/N}^{\varphi=\pi} \nonumber \\
  \fl \qquad\qquad\qquad \quad - \int_{\Omega(N)/N}^\pi d\varphi\, e^{i\tilde\omega  N\varphi} \frac{d}{d\varphi}
  \left(\frac{1}{(2\sin(\varphi/2))^{2\tilde\omega^2}}\right)\Bigg] \left(1+o(1) \right).
\end{eqnarray}
Here, we have used integration by parts in the second step.

As $N\rightarrow\infty$, the boundary term (taken together with the multiplicative pre-factor $N^{-2\tilde\omega^2}$) is of order ${\mathcal O}(\Omega(N)^{-2\tilde\omega^2})$. The integral in the last line of Eq.~(\ref{EN2-tech-2}) can readily be estimated as follows:
\begin{eqnarray}
    \fl \qquad
    \left|\int_{\Omega(N)/N}^\pi d\varphi\, e^{i\tilde\omega  N\varphi} \frac{d}{d\varphi}
  \left(\frac{1}{(2\sin(\varphi/2))^{2\tilde\omega^2}}\right)\right| \nonumber\\
  \fl \qquad\qquad
  \le \int_{\Omega(N)/N}^\pi d\varphi\, \left| \frac{d}{d\varphi}
  \left(\frac{1}{(2\sin(\varphi/2))^{2\tilde\omega^2}}\right)\right| ={\mathcal O}\left( \left( \frac{\Omega(N)}{N}\right)^{-2\tilde\omega^2}\right).
\end{eqnarray}
Taken together with the multiplicative pre-factor $N^{-2\tilde\omega^2}$ in Eq.~(\ref{EN2-tech-2}), its contribution to $\tilde{\mathcal{E}}_N^{(2)}(z)$ is of order ${\mathcal O}(\Omega(N)^{-2\tilde\omega^2})$, too. Hence,
\begin{eqnarray}\label{EN2-tilde}
\tilde{\mathcal{E}}_N^{(2)}(z)={\mathcal O}(\Omega(N)^{-2\tilde\omega^2}).
\end{eqnarray}
Comparing Eqs.~(\ref{EN1-tilde}) and (\ref{EN2-tilde}), we obtain Eq.~(\ref{IN0-lemma}).
\end{proof}
\noindent\newline
{\bf Proof of Theorem~\ref{cue-pv}}.---Equipped with Lemmas~\ref{lemma-tech} and \ref{lemma-tech-2}, we are ready to prove Theorem~\ref{cue-pv}. Equation (\ref{ps-lemma}) of Corollary~\ref{ps-master-symmetry} is the starting point.

(i) By virtue of Eq.~(\ref{ps-alter-master}), the term in the first line of Eq.~(\ref{ps-lemma}) equals
\begin{eqnarray} \label{ps-03} \fl \qquad
     T_1(z;N) = - N {\rm Re\,} \int_{0}^{\pi} \frac{d\varphi}{2\pi}\, \Big[ 1- \big( 1- z^{-N} \big)\frac{\varphi}{2\pi}
     \Big]\, \Phi_{N}((0,\varphi);1-z).
\end{eqnarray}
As $N\rightarrow\infty$, the uniform asymptotic expansion Eq.~(\ref{eq:Edge-2}) of Theorem~\ref{thm:CK} can be used to deduce that
\begin{eqnarray} \label{ps-03-2} \fl \quad
     T_1(z;N) = - N {\rm Re\,} \int_{0}^{\pi} d\varphi\, {\mathcal F}_N(\varphi;z)
     \, \exp\left(\int_0^{-i N\varphi} \frac{ds}{s} \sigma(s)\right)
     \left( 1 + {\mathcal O}(N^{-1+2\tilde\omega})
     \right),
\end{eqnarray}
where the function ${\mathcal F}_N(\varphi;z)$ is defined in Eq.~(\ref{eq:f}) of Lemma~\ref{lemma-tech}. According to this lemma, replacing
${\mathcal F}_N(\varphi;z)$ with unity in the integrand produces an error of order ${\mathcal O}(N^{-2\tilde\omega^2})$. Doing so and further changing the integration variable to $\lambda=\varphi N$, we obtain
\begin{eqnarray} \label{ps-03-3} \fl\quad
     T_1(z;N) = - {\rm Re\,} \int_{0}^{\pi N} d\lambda\,
     \, \exp\left(\int_0^{-i \lambda} \frac{ds}{s} \sigma(s)\right)
     \left( 1 + {\mathcal O}(N^{-1+2\tilde\omega})
     \right) + {\mathcal O}(N^{-2\tilde\omega^2}). \nonumber\\
     {}
\end{eqnarray}
Equation~(\ref{ps-03-3}) is particularly suitable for calculating the $N\rightarrow \infty$ limit; it is obtained by a plain replacement of the upper bound $\pi N$ with $\infty$.
Such a replacement is clearly justified as, according to Eq.~(\ref{eq:global}), it introduces an error of the order ${\mathcal O}(N^{-2\tilde\omega^2})$. Therefore, we are led to conclude that
\begin{eqnarray} \label{ps-03-limit}
    \lim_{N\rightarrow\infty} T_1(z;N) = - {\rm Re\,} \int_{0}^{\infty} d\lambda\,
     \, \exp\left(\int_0^{\lambda} \frac{dt}{t} \sigma_0(t;\zeta)\right).
\end{eqnarray}
Here, $\sigma_0(t;\zeta)=\sigma(s=-it)$ with $\zeta=1-e^{i\omega}$ is the fifth Painlev\'e transcendent that obeys Eq.~(\ref{PV-eq-0}) subject to the boundary condition Eq.~(\ref{bc-zero-0}).

(ii) It remains to prove that the term in the second line of Eq.~(\ref{ps-lemma}) can be neglected as $N\rightarrow\infty$. To this end,
we consider
\begin{eqnarray} \fl \label{eq:IN0-2} \quad
    I_{N,0}(z) = N \int_{0}^{\pi} \frac{d\varphi}{2\pi}\,
	\Phi_{N}((0,\varphi);1-z) \nonumber\\
\fl\qquad\quad= N \int_{0}^{\pi} \frac{d\varphi}{2\pi}\, \left(
	\frac{\sin(\varphi/2)}{\varphi/2}
	\right)^{-2\tilde{\omega}^2} \exp\left(
	\int_{0}^{- i N \varphi} \frac{ds}{s}\, \sigma(s)
	\right) \left( 1 + {\mathcal O}(N^{-1+ 2 \tilde{\omega}}) \right),
\end{eqnarray}
where the second line follows from substitution of a uniform asymptotics Eq.~(\ref{eq:Edge-2}). Lemma~\ref{lemma-tech-2} implies that $I_{N,0}(z)=\mathcal{O}(1)$ so that the overall contribution of the term in the second line of Eq.~(\ref{ps-lemma}) is of order ${\mathcal O}(N^{-1})$.

This conclusion combined with Eq.~(\ref{ps-03-limit}) allows us to evaluate the limit $N\rightarrow\infty$ in Eq.~(\ref{ps-lemma}); the result reads:
\begin{eqnarray}\label{final-eq-proof}
    \lim_{N\rightarrow\infty} S_N(\omega) = -\frac{z}{\pi(1-z)^2} {\rm Re\,} \int_{0}^{\infty} d\lambda\,
     \, \exp\left(\int_0^{\lambda} \frac{dt}{t} \sigma_0(t)\right).
\end{eqnarray}
Substitution $z=e^{i\omega}$ completes the proof. By virtue of Lemma~\ref{lemma-tech-2}, Eq.~(\ref{final-eq-proof}) is consistent with Eq.~(\ref{SN-anticipation}) stated in the opening of Section~\ref{CUE-INF}.
\hfill $\square$

\subsection{Relation to the ${\rm Sine}_2$ determinantal point process}\label{Sine-2-DPP}

Remarkably, the parameter-free power spectrum law $S_\infty(\omega)$, stated in Theorem~\ref{cue-pv} and proved in Section~\ref{cue-pv-proof}, can be interpreted in terms of the (random) counting function $n_{\rho}(\lambda)$ of the ${\rm Sine}_2$ determinantal point process~\cite{D-1962-CUE,KS-2009,VV-2020} with the mean local density $\rho=1/2\pi$, as spelled out in Lemma~\ref{Sine-2-Lemma}.
\noindent\newline\newline
{\bf Proof of Lemma~\ref{Sine-2-Lemma}.}---To appreciate the ideas behind Eq.~(\ref{Sine-2-Lemma-eq}), let us define the moment generating function (MGF)
\begin{eqnarray}  \label{sine-mgf}
    \Phi_{\rho,\infty}\big(\lambda;1-e^{i\omega}\big) = \langle e^{i\omega n_{\rho}(\lambda)}\rangle
\end{eqnarray}
of the counting function $n_{\rho}(\lambda)$.

First, setting $\zeta=1-e^{i\omega}$, we quote the relation
\begin{eqnarray} \label{relation-to-sine-process}\fl\quad
        \Phi_{\rho,\infty}\big(\lambda;1-e^{i\omega}\big) = \langle e^{i\omega n_\rho(\lambda)}\rangle =  {\rm det\,} \Big[
        \mathds{1} - \zeta \hat{K}_{\rho,\infty}^{(0,\lambda)}
    \Big] = \exp\left(\int_0^{2\pi \rho \lambda} \frac{dt}{t} \sigma_0(t;\zeta)\right).\nonumber\\
    {}
\end{eqnarray}
Going back to the pioneering work~\cite{JMMS-1980} by Jimbo, Miwa, M\^{o}ri and Sato, this key identity expresses the MGF Eq.~(\ref{sine-mgf}) in terms of

(i) the Fredholm determinant associated with an integral operator $\hat{K}_{\rho,\infty}^{(0,\lambda)}$, defined by
\begin{eqnarray}\label{sine-operator}
    \left[\hat{K}_{\rho,\infty}^{(0,\lambda)} f\right](x) = \int_{0}^{\lambda} K_{\rho,\infty} (x-y)\, f(y)\, dy,
\end{eqnarray}
and, further, in terms of

(ii) the fifth Painlev\'e transcendent $\sigma_0(t;\zeta)$ defined in Theorem~\ref{cue-pv}. Here,
\begin{eqnarray} \label{Sine-Kernel}
    K_{\rho, \infty}(x) =
    \frac{\sin(\pi \rho x)}{\pi x}
\end{eqnarray}
is the sine kernel satisfying the reproducing property
\begin{eqnarray} \label{rep-prop-Sine}
    \int_{-\infty}^{+\infty} dw K_{\rho,\infty}(x-w)\,
    K_{\rho, \infty}(w-y) = K_{\rho,\infty}(x-y).
\end{eqnarray}

Second, we make use of the fact that the ${\rm Sine}_2$ process with the mean local density $\rho$ can be considered as the $N\rightarrow\infty$ limit of the finite-$N$ point process $\{N\theta_1/2\pi\rho,\dots,N\theta_{L(N)}/2\pi \rho\}$ of the length $L(N)$ built upon ${\rm CUE}(N)$ eigen-angles $\{\theta_j\}_{j=1}^N$, see Eq.~(\ref{jpdf-cue}). Here, $L(N)$ is an increasing sequence of integers such that $L(N)\rightarrow \infty$ whilst $L(N)/N \rightarrow 0$ as $N \rightarrow \infty$.

Indeed, following Dyson~\cite{D-1962-CUE}, the $\ell$-point correlation function of the ${\rm Sine}_2$ process admits a determinantal representation
\begin{eqnarray}\label{L-point-det} \fl \qquad\qquad\quad
    \rho_\ell(x_1,\dots,x_\ell) = \lim_{N\rightarrow \infty} \left( \frac{\rho}{N}\right)^\ell R_{\ell,N}\left( \frac{2\pi \rho}{N}x_1,\dots,\frac{2\pi \rho}{N}x_\ell\right) \nonumber\\
    \fl \qquad\qquad\qquad\qquad\quad\quad = {\,} {\det}_{1\le j,k \le \ell} \left[
        K_{\rho, \infty} (x_j -x_k)
    \right],
\end{eqnarray}
see Eqs.~(\ref{RLN-det}) and (\ref{CUEN-Kernel}). Here, $K_{\rho,\infty}(x)$ is the sine kernel [Eq.~(\ref{Sine-Kernel})] naturally arising as the $N\rightarrow\infty$ limit of the ${\rm CUE}(N)$ kernel:
\begin{eqnarray} \label{Sine-Kernel-limit}
     \lim_{N\rightarrow \infty}\frac{\rho}{N} K_N\left( \frac{2\pi \rho}{N} x \right)= \frac{\sin(\pi \rho x)}{\pi x} = K_{\rho,\infty}(x).
\end{eqnarray}
Consequently, the aforementioned (random) counting function $n_\rho(\lambda)$ of such a ${\rm Sine}_2$ process can also be introduced through the limiting procedure
\begin{eqnarray} \label{CF-rho}
        n_\rho(\lambda) = \lim_{N\rightarrow\infty} {\mathcal N}_N\left( 0,\frac{2\pi \rho}{N} \lambda \right),
\end{eqnarray}
where ${\mathcal N}_N\left( 0, \varphi\right)$ is the ${\rm CUE}(N)$ counting function [Eq.~(\ref{NCF})]. This, in turn, implies that the MGF for the ${\rm Sine}_2$ process [Eq.~(\ref{sine-mgf})] can be related to the ${\rm CUE}(N)$ generating function [Eq.~(\ref{mgf-CUEN})] through the limiting procedure
\begin{eqnarray} \fl \qquad\label{Def-Inf}
    \Phi_{\rho, \infty}(\lambda;1-z) = \langle e^{i\omega n_\rho(\lambda)}\rangle = \lim_{N\rightarrow\infty}
        \left< e^{i\omega {\mathcal N}_N(0,2\pi\rho \lambda/N)}\right>\nonumber\\
        \fl\qquad\qquad\qquad =\lim_{N\rightarrow\infty} \Phi_N\left(
        \Big(0,\frac{2\pi \rho}{N}\lambda\Big); 1-z
        \right),
\end{eqnarray}
where $z=e^{i\omega}$. Substituting the expansion Eq.~(\ref{FD-01}) into the r.h.s. of Eq.~(\ref{Def-Inf}) and making use of Propositions~9.1.1 and 9.1.2 of Ref.~\cite{PF-book}, we evaluate the limit $N\rightarrow \infty$ in Eq.~(\ref{Def-Inf}) to derive
\begin{eqnarray} \fl  \label{FD-inf}
    \Phi_{\rho,\infty}(\lambda;1-z) = \langle e^{i\omega n_\rho(\lambda)}\rangle \nonumber\\
    = 1+
    \sum_{\ell=1}^{\infty} \frac{(-\zeta)^\ell}{\ell!} \left(
    \prod_{j=1}^\ell \int_0^\lambda dx_j \right)\,
    {\det}_{1\le i,j \le \ell} \left[
        K_{\rho,\infty} (x_i - x_j)
    \right].
\end{eqnarray}
Here, the scalar kernel $K_{\rho,\infty}(x)$ is the sine kernel Eq.~(\ref{Sine-Kernel}). The series in Eq.~(\ref{FD-inf}) can readily be recognized as the Fredholm determinant:
\begin{eqnarray}  \label{FD-inf-FD}
    \Phi_{\rho,\infty}(\lambda;1-z) = \langle e^{i\omega n_\rho(\lambda)}\rangle =
    {\rm det\,} \left[
        \mathds{1} - \zeta \hat{K}_{\rho,\infty}^{(0,\lambda)}
    \right],
\end{eqnarray}
where $\hat{K}_{\rho,\infty}^{(0,\lambda)}$ is an integral operator defined by Eq.~(\ref{sine-operator}). This explains the second equality in Eq.~(\ref{relation-to-sine-process}). The third equality expressing the Fredholm determinant in terms of the fifth Painlev\'e transcendent $\sigma_0(t;\zeta)$ is the celebrated result of Ref.~\cite{JMMS-1980}, see also Refs.~\cite{JMU-1981,M-1992,TW-1993}.

Comparison of Eq.~(\ref{ps-cue-inf}) from Theorem~\ref{cue-pv} with Eq.~(\ref{relation-to-sine-process}) taken at $\rho=1/2\pi$ completes the proof.
\hfill $\square$

\begin{remark}\label{PV-GFP}
The identity Eq.~(\ref{relation-to-sine-process}) can be interpreted in terms of the probabilities $E_{\rho,\infty}(\ell;\lambda)$ to find exactly $\ell$ eigenvalues in the interval $(0,\lambda)$ of an unfolded spectrum of the ${\rm CUE}(\infty)$ ensemble. Indeed, calculating the average $\langle e^{i\omega n_\rho(\lambda)}\rangle$, we observe:
\begin{eqnarray}\fl\qquad \quad \label{pv-prob-expansion}
    \Phi_{\rho,\infty}\big(\lambda;\zeta) =
    \exp\left(\int_0^{2\pi \rho \lambda} \frac{dt}{t} \sigma_0(t;\zeta)\right) = \sum_{\ell=0}^{\infty} (1-\zeta)^\ell E_{\rho,\infty}(\ell;\lambda),
\end{eqnarray}
where $\zeta = 1-e^{i\omega}$.
\hfill $\blacksquare$
\end{remark}

\begin{remark}
  Equation~(\ref{CF-rho}) implies the relation
  \begin{eqnarray} \label{CF-diff-densities}
    n_\rho(\lambda) = n_{1/2\pi} (2\pi \rho\lambda)
  \end{eqnarray}
  so that
  \begin{eqnarray} \label{Phi-rho}
    \Phi_{\rho,\infty}(\lambda;1-z) = \Phi_{1/2\pi,\infty}(2\pi\rho\lambda;1-z)
  \end{eqnarray}
and
\begin{eqnarray} \label{E-rho}
    E_{\rho,\infty}(\ell;\lambda) = E_{1/2\pi,\infty}(\ell; 2\pi\rho\lambda).
\end{eqnarray}
\hfill $\blacksquare$
\end{remark}
\noindent
\par
In the rest of this section, we deal with a technical lemma required to justify the claim made in Remark~\ref{Rem-no-poles}.

\begin{lemma}\label{Lemma-zeros}
The function $\Phi_{\rho,\infty}\left(\lambda;1-e^{i\omega}\right)\neq 0$ for all $\lambda\ge 0$ and $0\le \omega <\pi$.
\end{lemma}
\begin{proof}
  Since the function $\Phi_{\rho,\infty}\left(\lambda;1-e^{i\omega}\right)$ admits the Fredholm determinant representation Eq.~(\ref{FD-inf-FD}), it can equivalently be expressed
  in terms of the infinite product~\cite{M-1992}
  \begin{eqnarray}
    \Phi_{\rho,\infty}\left(\lambda;1-e^{i\omega}\right) = \prod_{j=0}^{\infty}\Big[1-
    \big(1-e^{i\omega}\big) \Lambda_j^{(\rho)}(\lambda)\Big],
  \end{eqnarray}
where $\Lambda_j^{(\rho)}(\lambda)$ are eigenvalues of the integral operator $\hat{K}_{\rho,\infty}^{(0,\lambda)}$ defined by Eq.~(\ref{sine-operator}). Since
$0\le \Lambda_j^{(\rho)}(\lambda) \le 1$, see e.g. Ref.~\cite{M-1992}, the product can nullify only for $\omega=\pi+2\pi k$, where $k$ is an integer. Hence, it cannot nullify for $0\le \omega <\pi$.
\end{proof}

\begin{remark}\label{Remark-zeros}
As $N\rightarrow\infty$, the Toeplitz determinant
\begin{eqnarray}
    D_N\left( \frac{i s}{2N};0,0;\tilde\omega,-\tilde\omega\right) \neq 0
\end{eqnarray}
for all $s \in -i{\mathbb R}_+$ and $0\le \tilde\omega < 1/2$. This readily follows from Eq.~(\ref{GF-toeplitz-1}) of Proposition~\ref{toep-prop}, Eq.~(\ref{Def-Inf}) and Lemma~\ref{Lemma-zeros}.
\hfill $\blacksquare$
\end{remark}

\subsection{Small-$\omega$ expansion of the universal law}\label{small-omega-section}

Below, the universal power spectrum law [Eq.~(\ref{ps-cue-inf})],
\begin{eqnarray} \label{ps-final-1}
    S_\infty(\omega) = \frac{1}{4\pi \sin^2(\omega/2)}{\rm Re\,}\int_{0}^{\infty} d\lambda \, \exp\left(
        \int_{0}^{\lambda}\frac{dt}{t}\,  \sigma_0(t;\zeta)
    \right),
\end{eqnarray}
will be studied in the vicinity of $\omega=0$. We remind that $\zeta=1-e^{i\omega}$. An alternative representation of the power spectrum in terms of the moment generating function of the counting function $n_{1/2\pi}(\lambda)$ of the ${\rm Sine}_2$ point process, introduced in Section~\ref{Sine-2-DPP} [Eq.~(\ref{sine-mgf})], combined with the cumulant expansion
\begin{eqnarray} \label{cum-exp} \fl\qquad\qquad
    \log \Phi_{1/2\pi,\infty}\big(\lambda; 1-e^{i\omega}\big) = \int_{0}^{\lambda}\frac{dt}{t}\,  \sigma_0(t;\zeta)
    =  \sum_{\ell=1}^\infty \frac{(i\omega)^\ell}{\ell!} \,\kappa_\ell(\lambda)
\end{eqnarray}
will play an essential r\^ole in the forthcoming analysis. Here $\kappa_\ell(\lambda) = \langle\!\langle  n_{1/2\pi}^\ell(\lambda)\rangle\!\rangle$ is the $\ell$-th cumulant of the counting function $n_{1/2\pi}(\lambda)$, see Eq.~(\ref{CF-rho}). Their detailed analysis is presented in Appendix~\ref{A-3-2}.

The main result of this section is stated below.

\begin{proposition}
\label{Th-small-omega}
  As $\omega\rightarrow 0$, the following expansion holds:
  \begin{eqnarray} \label{S-res-0}
    S_\infty (\omega) = \frac{1}{4\pi^2 \tilde\omega} + \frac{1}{2\pi^2} \tilde\omega \log \tilde\omega +\frac{\tilde\omega}{12}
    + {\mathcal O}(\tilde\omega^2),
\end{eqnarray}
where $\tilde\omega = \omega/2\pi$.
\end{proposition}
\begin{proof}
A small-$\omega$ expansion of the fifth Painlev\'e transcendent $\sigma_0(t;\zeta)$, studied in detail in Appendix~\ref{A-3}, lays the basis for our proof. Adopting notations of Eqs.~(\ref{sigma-expan-app}), (\ref{Fk-integrated}) and (\ref{f2-integrated-reg}), we have
\begin{eqnarray} \fl\qquad \label{s-omega}
    \int_{0}^{\lambda}\frac{dt}{t}\,  \sigma_0(t;\zeta) =
     i\tilde\omega \lambda + \tilde\omega^2 \tilde{\mathcal F}_2(\lambda) -2 \tilde\omega^2 \log\lambda
     + \sum_{k=3}^{\infty}\tilde\omega^k {\mathcal F}_k(\lambda).
\end{eqnarray}
Notice that we have explicitly factored out the function $\tilde{\mathcal F}_2(\lambda)$ which, contrary to ${\mathcal F}_2(\lambda)$, does {\it not} diverge as $\lambda\rightarrow \infty$, see a discussion of the logarithmic divergency of ${\mathcal F}_2(\lambda)$ above Eq.~(\ref{f2-integrated-reg}).

In the next step, we would like to take the small-$\omega$ expansion Eq.~(\ref{s-omega}) into the outer integral of Eq.~(\ref{ps-final-1}) and further expand a part of the exponent which stays bounded as $\lambda\rightarrow\infty$,
\begin{eqnarray} \fl\qquad \label{exp-bounded-exp}
	\exp\left(\tilde\omega^2 \tilde{\mathcal F}_2(\lambda) +  \tilde\omega^3 {\mathcal F}_3(\lambda)+ \omega^4 {\mathcal F}_4(\lambda) + {\mathcal O}(\tilde\omega^5)  \right)\nonumber\\
    \fl\qquad\qquad\qquad
	= 1 + \tilde\omega^2 \tilde{\mathcal F}_2(\lambda) + \tilde\omega^3 {\mathcal F}_3(\lambda)+ \frac{\tilde\omega^4}{2} \tilde{\mathcal F}^2_2(\lambda) + \tilde\omega^4 {\mathcal F}_4(\lambda) + {\mathcal O}(\tilde\omega^5).
\end{eqnarray}
To argue that we are allowed to do so, we split the outer integral into two parts, the first one from $0$ to $1$, and the second one from $1$ to $\infty$. For the latter part, we are allowed to make the expansion of the exponent inside the integral since it is uniformly bounded for $\lambda\in[1,\infty)$. For the former part, the Lebesgue's dominated convergence theorem is at work, since there exists a constant $c>0$ such that the integrand is dominated by $c \lambda^{-1/2}$ for all $\lambda\in[0,1]$ and all $\tilde\omega\in[0,1/2]$. This argument also works for higher order terms in the expansion.

Denoting the real part of the outer integral in Eq.~(\ref{ps-final-1}) as $J(\tilde\omega)$, we represent the power spectrum, Eq.~(\ref{ps-final-1}), in the form
\begin{eqnarray}\label{SwJ}
    S_\infty(\omega) = \frac{1}{4\pi \sin^2(\pi\tilde\omega)}\, J(\tilde\omega),
\end{eqnarray}
where
\begin{eqnarray} \label{eq:expansionJ}\fl\qquad
    J(\tilde\omega) = {\rm Re\,} \int_{0}^{\infty} d\lambda \, e^{i \tilde\omega \lambda} \lambda^{-2\tilde\omega^2} \nonumber\\
    \fl\qquad\qquad\qquad \times
    \Big(
    1 + \tilde\omega^2 \tilde{\mathcal F}_2(\lambda) + \tilde\omega^3 {\mathcal F}_3(\lambda)
    + \frac{\tilde\omega^4}{2} \tilde{\mathcal F}^2_2(\lambda) +  \tilde\omega^4 {\mathcal F}_4(\lambda) + {\mathcal O}(\tilde\omega^5)
    \Big).
\end{eqnarray}
Here we have used Eqs.~(\ref{s-omega}) and (\ref{exp-bounded-exp}).

Since the functions $\tilde{{\mathcal F}}_2$ and ${\mathcal F}_4$ are real valued [see Eqs.~(\ref{Fk-integrated}), (\ref{f2-part}) and (\ref{f4t-def})] whilst ${\mathcal F}_3$ is purely imaginary [see Eqs.~(\ref{Fk-integrated}) and (\ref{f3-part})], the above reduces to
\begin{eqnarray} \fl \qquad \label{J-omega}
    J(\tilde\omega) =  \int_{0}^{\infty} d\lambda \,\frac{\cos(\tilde\omega \lambda)}{\lambda^{2\tilde\omega^2}}
    \Big(
    1 + \tilde\omega^2 \tilde{\mathcal F}_2(\lambda) + \frac{\tilde\omega^4}{2} \tilde{\mathcal F}^2_2(\lambda)
    + \tilde\omega^4 {\mathcal F}_4(\lambda)
    \Big) \nonumber\\
    \fl \qquad\qquad\qquad\qquad
    -   \int_{0}^{\infty} d\lambda \, \frac{\sin(\tilde\omega \lambda)}{\lambda^{2\tilde\omega^2}}
    \Big( \tilde\omega^3 {\rm Im\,}{\mathcal F}_3(\lambda) + {\mathcal O}(\tilde\omega^5)\Big).
\end{eqnarray}
Below we consider the above integrals, one by one, keeping the terms up to ${\mathcal O}(\tilde\omega^4)$.

(i) The first integral in Eq.~(\ref{J-omega}) can be evaluated explicitly
\begin{eqnarray}  \label{J1}
    J_1(\tilde\omega) =  \int_{0}^{\infty} d\lambda \, \frac{\cos(\tilde\omega \lambda)}{\lambda^{2\tilde\omega^2}}
    =  \tilde\omega^{-1+2\tilde\omega^2} \sin(\pi \tilde\omega^2) \Gamma(1-2\tilde\omega^2)
\end{eqnarray}
and further expanded, as $\tilde\omega\rightarrow 0$, to bring
\begin{eqnarray}  \label{J1-exp}
    J_1(\tilde\omega) = \pi \tilde\omega + 2 \pi \left(
        \gamma + \log \tilde\omega
    \right) \tilde\omega^3 + {\mathcal O}(\tilde\omega^4).
\end{eqnarray}
Here, $\gamma$ is the Euler's constant.

(ii) The second integral in Eq.~(\ref{J-omega}),
\begin{eqnarray} \label{J2}
    J_2(\tilde\omega) =  \int_{0}^{\infty} d\lambda \, \frac{\cos(\tilde\omega \lambda)}{\lambda^{2\tilde\omega^2}}\, \tilde{{\mathcal F}}_2(\lambda),
\end{eqnarray}
to be multiplied by $\tilde\omega^2$ [see Eq.~(\ref{J-omega})], should be expanded up to remainder terms of order ${\mathcal O}(\tilde\omega^2)$.

To determine a small-$\tilde\omega$ expansion of $J_2(\tilde\omega)$, one has first to regularize the integral above,
\begin{eqnarray}  \label{J2a}
    J_2(\tilde\omega) = \delta J_2(\tilde\omega)
    + \tilde{{\mathcal F}}_2(\infty)\, J_1(\tilde\omega),
\end{eqnarray}
where
\begin{eqnarray}\label{J2b}
   \delta J_2(\tilde\omega) = \int_{0}^{\infty} d\lambda \, \frac{\cos(\tilde\omega \lambda)}{\lambda^{2\tilde\omega^2}} \left( \tilde{{\mathcal F}}_2(\lambda) - \tilde{{\mathcal F}}_2(\infty)\right)
\end{eqnarray}
stays finite as $\tilde\omega \rightarrow 0$, and $\tilde{{\mathcal F}}_2(\infty)=-2(1+\gamma)$, see Eq.~(\ref{F2tilde-inf}). Substituting $\tilde{{\mathcal F}}_2(\lambda)$ from Eq.~(\ref{f2-integrated-reg}) and (\ref{f2-integrated}), we perform the integral in Eq.~(\ref{J2b}) to derive:
\begin{eqnarray} \fl \quad \label{dJ2-closed}
    \delta J_2(\tilde\omega) = \Gamma \left(1-2 \tilde\omega^2\right) \sin \left(\pi  \tilde\omega^2\right)
    \Bigg\{ (1+\tilde\omega)^{2 \tilde\omega^2-1} + (1-\tilde\omega)^{2 \tilde\omega^2-1}
    \nonumber\\
    \fl \qquad \qquad\qquad
    - \frac{1-2 \tilde\omega^2}{1- \tilde\omega^2}
    \,{}_3F_2\left(1-\tilde\omega^2,1-\tilde\omega^2,\frac{3}{2}-\tilde\omega^2;\frac{1}{2},2-\tilde\omega^2;\tilde\omega^2\right)\nonumber\\
     \fl \qquad \qquad\qquad
       -  \frac{2}{1-2 \tilde\omega^2} \, {}_3F_2\left(\frac{1}{2}-\tilde\omega^2,\frac{1}{2}-\tilde\omega^2,1-\tilde\omega^2;\frac{1}{2},\frac{3}{2}-\tilde\omega^2;\tilde\omega^2\right)
       \Bigg\}.
\end{eqnarray}
Here, ${}_p F_q$ is the hypergeometric function. A small-$\tilde\omega$ expansion of Eq.~(\ref{dJ2-closed}) yields
\begin{eqnarray}
    \delta J_2(\tilde\omega) = {\mathcal O}(\tilde\omega^2).
\end{eqnarray}
Returning to Eqs.~(\ref{J2a}) and (\ref{J1-exp}), and having in mind Eq.~(\ref{F2tilde-inf}), we end up with the sought estimate
\begin{eqnarray}\label{J2-exp}
    \tilde\omega^2 J_2(\tilde\omega) = -2\pi (1+\gamma) \tilde\omega^3 + {\mathcal O}(\tilde\omega^4),
\end{eqnarray}
see Eq.~(\ref{J-omega}), and compare it to Eq.~(\ref{J1-exp}).

(iii) A small-$\tilde\omega$ expansion of the third integral in Eq.~(\ref{J-omega}),
\begin{eqnarray} \fl \qquad\qquad\qquad \label{J3}
    J_3(\tilde\omega) =  \int_{0}^{\infty} d\lambda \, \frac{\cos(\tilde\omega \lambda)}{\lambda^{2\tilde\omega^2}} \tilde{{\mathcal F}}^2_2(\lambda),
\end{eqnarray}
to be multiplied by $\tilde\omega^4$, can be studied along the same lines. First we regularize the integral,
\begin{eqnarray} \fl \qquad \label{J3-a}
    J_3(\tilde\omega) =  \int_{0}^{\infty} d\lambda \, \frac{\cos(\tilde\omega \lambda)}{\lambda^{2\tilde\omega^2}}
    \left(
    \tilde{{\mathcal F}}^2_2(\lambda)- \tilde{{\mathcal F}}^2_2(\infty)\right) + \tilde{{\mathcal F}}^2_2(\infty) J_1(\tilde\omega),
\end{eqnarray}
compare to Eqs.~(\ref{J2a}) and (\ref{J2b}). Second, we notice that as soon as
\begin{eqnarray}
    \tilde{{\mathcal F}}_2^2(\lambda)-  \tilde{{\mathcal F}}_2^2(\infty) = {\mathcal O}\left( \frac{1}{\lambda^2}\right),
\end{eqnarray}
see Eqs.~(\ref{f2-integral-as-inf}), (\ref{f2-integrated-reg}) and (\ref{F2tilde-inf}), the integral in Eq.~(\ref{J3-a}) stays finite as $\tilde\omega\rightarrow 0$~\footnote[3]{To take the limit $\tilde\omega \rightarrow 0$ into the integral, we use the Lebesgue's dominated convergence theorem by spotting that the integrand is dominated by the integrable function
\begin{eqnarray}
    f({\lambda}) = \left\{
                   \begin{array}{ll}
                     c\lambda^{-1/2}, & \hbox{$0\le \lambda \le 1$;} \\
                     c\lambda^{-1-\epsilon}, & \hbox{$\lambda > 1$}
                   \end{array}
                 \right. \nonumber
\end{eqnarray}
for all $\lambda\ge 0$ and all $\tilde\omega \in [0,1/2]$, where $c$ is a sufficiently large and $\epsilon$ is sufficiently small constant} yielding a contribution ${\mathcal O}(\tilde\omega^0)$. The second term in Eq.~(\ref{J3-a}) is of order ${\mathcal O}(\tilde\omega)$, see Eq.~(\ref{J1-exp}). Therefore, we conclude that
\begin{eqnarray}\label{J3-exp}
    \tilde\omega^4 J_3(\tilde\omega) = {\mathcal O}(\tilde\omega^4).
\end{eqnarray}

(iv) To study the fourth integral in Eq.~(\ref{J-omega}),
\begin{eqnarray} \label{J4}
    J_4(\tilde\omega) =  \int_{0}^{\infty} d\lambda \, \frac{\cos(\tilde\omega \lambda)}{\lambda^{2\tilde\omega^2}}\, {\mathcal F}_4(\lambda),
\end{eqnarray}
to be multiplied by $\tilde\omega^4$ [see Eq.~(\ref{J-omega})], we first re-write it as a sum of two terms:
\begin{eqnarray} \fl \qquad \label{J4-1}
    J_4(\tilde\omega) =  \int_{0}^{\infty} d\lambda \, \frac{\cos(\tilde\omega \lambda)}{\lambda^{2\tilde\omega^2}}
     \left({\mathcal F}_4(\lambda)-{\mathcal F}_4(\infty) \right) + {\mathcal F}_4(\infty) J_1(\tilde\omega)
\end{eqnarray}
Since [Eqs.~(\ref{f4-integral-as-inf}) and (\ref{F2k-inf})]
\begin{eqnarray}
    {\mathcal F}_4(\lambda) - {\mathcal F}_4(\infty) = {\mathcal O}\left(
        \left( \frac{\log \lambda}{\lambda} \right)^2
    \right) \quad {\rm as} \quad \lambda\rightarrow \infty,
\end{eqnarray}
we may take the limit $\tilde\omega\rightarrow 0$~\footnotemark[3] inside the integral to get
\begin{eqnarray} \label{J4-a}
    J_4(0) =  \int_{0}^{\infty} d\lambda \,\left( {\mathcal F}_4(\lambda) - {\mathcal F}_4(\infty) \right).
\end{eqnarray}
Since the above integral exists, we conclude that
\begin{eqnarray}\label{J4-exp}
    \tilde\omega^4 J_4(\tilde\omega) = {\mathcal O}(\tilde\omega^4).
\end{eqnarray}

(v) The fifth integral in Eq.~(\ref{J-omega}),
\begin{eqnarray} \label{J5}
    J_5(\tilde\omega) =
    \int_{0}^{\infty} d\lambda \, \frac{\sin(\tilde\omega \lambda)}{\lambda^{2\tilde\omega^2}}
     {\rm Im\,}{\mathcal F}_3(\lambda),
\end{eqnarray}
to be multiplied by $\tilde\omega^3$ [see Eq.~(\ref{J-omega})], should be expanded up to remainder terms of order ${\mathcal O}(\tilde\omega)$.

To determine the low order terms of its small-$\tilde\omega$ expansion, we rewrite the above in the regularized form
\begin{eqnarray} \label{J5-a} \fl \qquad
    \frac{J_5(\tilde\omega)}{\tilde\omega} =
    \int_{0}^{\infty} d\lambda \, \frac{\sin(\tilde\omega \lambda)}{\lambda^{2\tilde\omega^2}\,{\tilde\omega}}
     \left( {\rm Im\,}{\mathcal F}_3(\lambda)
     + \left(\frac{4}{\lambda} + \frac{8}{\lambda^2} (\gamma+\log \lambda) \sin \lambda\right)
     \right) \nonumber\\
     \fl\qquad\qquad\qquad
     - \frac{1}{\tilde\omega}\int_{0}^{\infty} d\lambda \, \frac{\sin(\tilde\omega \lambda)}{\lambda^{2\tilde\omega^2}}
     \left( \frac{4}{\lambda} + \frac{8}{\lambda^2} (\gamma+\log \lambda) \sin \lambda
     \right).
\end{eqnarray}
The second integral in Eq.~(\ref{J5-a}) can be evaluated explicitly
\begin{eqnarray} \label{J5-explicit-part} \fl \qquad
    \int_{0}^{\infty} d\lambda \, \frac{\sin(\tilde\omega \lambda)}{\lambda^{2\tilde\omega^2}}
     \left( \frac{4}{\lambda} + \frac{8}{\lambda^2} (\gamma+\log \lambda) \sin \lambda
     \right) \nonumber\\
     \fl \qquad\qquad
     = \frac{2\pi}{\cos \left(\pi  \tilde\omega^2\right) \Gamma \left(2+2 \tilde\omega^2\right)}
     \Bigg\{ \left(1+2 \tilde\omega^2\right) {\tilde\omega}^{2 \tilde\omega^2} \nonumber\\
     \fl \qquad \qquad
     + \left((1+\tilde\omega)^{1+2 \tilde\omega^2}-(1-\tilde\omega)^{1+ 2 \tilde\omega^2}\right) \left(\psi\left(2+2 \tilde\omega^2\right)+\gamma -\frac{\pi}{2} \tan \left(\pi  \tilde\omega^2\right)\right) \nonumber\\
     \fl \qquad \qquad
     - \left((1+\tilde\omega)^{1+2 \tilde\omega^2} \log (1+\tilde\omega)-(1-\tilde\omega)^{1+ 2\tilde\omega^2} \log (1-\tilde\omega)\right)
     \Bigg\}
\end{eqnarray}
and further expanded in small $\tilde\omega$. (Here, $\psi(\cdots)$ is the digamma function.) This yields:
\begin{eqnarray} \label{J5-explicit-part-exp}  \fl\qquad
    \int_{0}^{\infty} d\lambda \, \frac{\sin(\tilde\omega \lambda)}{\lambda^{2\tilde\omega^2}}
     \left( \frac{4}{\lambda} + \frac{8}{\lambda^2} (\gamma+\log \lambda) \sin \lambda
     \right) = 2\pi + {\mathcal O}(\tilde\omega^2 \log\tilde\omega),
\end{eqnarray}
which is the contribution of the second integral in Eq.~(\ref{J5-a}) into $J_5(\tilde\omega)$.

Coming back to the first integral in Eq.~(\ref{J5-a}), we notice that the multiplicative function in its integrand exhibits a sufficiently quick decay [see Eq.~(\ref{f3-integral-as-inf})]
\begin{eqnarray}\label{f-est-3-1}
    {\rm Im}\,{{\mathcal F}}_3(\lambda) + \left(\frac{4}{\lambda} + \frac{8}{\lambda^2} (\gamma+\log \lambda) \sin \lambda\right) =
    {\mathcal O}\left( \frac{\log\lambda}{\lambda^3} \right)
\end{eqnarray}
as $\lambda\rightarrow\infty$. This observation allows~\footnotemark[3] us to take the limit $\tilde\omega\rightarrow 0$ inside the integral. As soon as
\begin{eqnarray}
  \lim_{\tilde\omega\rightarrow 0}\frac{\sin(\tilde\omega \lambda)}{\lambda^{2\tilde\omega^2}\,{\tilde\omega}} = \lambda,
\end{eqnarray}
we find that the leading contribution of the first integral into $J_5(\tilde\omega)$ equals
\begin{eqnarray} \label{J5-b} \fl \qquad
    \tilde\omega \int_{0}^{\infty} d\lambda \, \lambda
     \left( {\rm Im\,}{\mathcal F}_3(\lambda)
     + \left(\frac{4}{\lambda} + \frac{8}{\lambda^2} (\gamma+\log \lambda) \sin \lambda\right)
     \right) = {\mathcal O}(\tilde\omega)
\end{eqnarray}
since the integral above converges.

Combining Eq.~(\ref{J5-explicit-part-exp}) with Eq.~(\ref{J5-b}), we conclude that
\begin{eqnarray}\label{J5-exp}
    \tilde\omega^3 J_5(\tilde\omega) = - 2\pi \tilde\omega^3 + {\mathcal O}(\tilde\omega^4).
\end{eqnarray}

(vi)~To estimate contributions of higher-order terms in the small-$\omega$ expansion in Eq.~(\ref{eq:expansionJ}), we introduce a set of integrals
\begin{eqnarray} \label{eq:Gk}
	\mathcal{G}_k(\tilde\omega)={\rm Re\,}  \int_{0}^{\infty} d\lambda \, e^{i\tilde{\omega}\lambda}\lambda^{-2\tilde{\omega}^2} \mathcal{F}_k(\lambda),
\end{eqnarray}
where $k \ge 5$. Their contribution to $J(\tilde{\omega})$ equals $\tilde{\omega}^k \mathcal{G}_k(\tilde\omega)$. (The contribution of such terms for $k\le 4$ have already been discussed.)

To study a small-$\omega$ behavior of $\mathcal{G}_k(\tilde\omega)$, we split the integral Eq.~(\ref{eq:Gk}) in the following way (which is only necessary for even $k$ when $\mathcal{F}_k(\infty)$ is non-zero):
\begin{eqnarray} \label{eq:Gk2} \fl\qquad
\mathcal{G}_k(\tilde\omega)=\mathcal{F}_k(\infty) J_1(\tilde{\omega})+{\rm Re\,} \int_{0}^{1} d\lambda \, e^{i\tilde{\omega}\lambda}\lambda^{-2\tilde{\omega}^2} (\mathcal{F}_k(\lambda)-\mathcal{F}_k(\infty)) \nonumber\\
\fl\qquad\qquad\qquad	
    +{\rm Re\,} \int_{1}^{\infty} d\lambda \, e^{i\tilde{\omega}\lambda}\lambda^{-2\tilde{\omega}^2} (\mathcal{F}_k(\lambda)-\mathcal{F}_k(\infty))  \nonumber\\
\fl\qquad\qquad	
    =\mathcal{F}_k(\infty) J_1(\tilde{\omega})+{\rm Re\,} \int_{0}^{1} d\lambda \, e^{i\tilde{\omega}\lambda}\lambda^{-2\tilde{\omega}^2} (\mathcal{F}_k(\lambda)-\mathcal{F}_k(\infty)) \nonumber\\
\fl\qquad\qquad\qquad	
     -{\rm Re\,}\left[\frac{e^{i\tilde\omega}}{i\tilde\omega}\left(\mathcal{F}_k(1)-\mathcal{F}_k(\infty)\right)\right]\nonumber\\
\fl\qquad\qquad\qquad	
     -{\rm Re\,} \int_{1}^{\infty} d\lambda \, \frac{e^{i\tilde{\omega}\lambda}}{i\tilde\omega}\lambda^{-2\tilde{\omega}^2-1}
     \bigg[2\tilde\omega^2 (\mathcal{F}_k(\infty)-\mathcal{F}_k(\lambda))+f_k(\lambda) \bigg].
\end{eqnarray}
Here, we have used integration by parts in the second step. Next, we consider the limit
\begin{eqnarray} \label{eq:limGk} \fl \qquad\quad
	\lim_{\tilde\omega\rightarrow 0} \tilde\omega\mathcal{G}_k(\tilde\omega)= -{\rm Im\,}\left(\mathcal{F}_k(1)-\mathcal{F}_k(\infty)\right)-{\rm Im\,} \int_{1}^{\infty} d\lambda \,   \frac{f_k(\lambda)}{\lambda} =0,
\end{eqnarray}
where we have used Eq.~(\ref{J1-exp}) and Eqs.~(\ref{eq:Fk0}), (\ref{eq:Fkinf}) and (\ref{eq:fkinf}) which show that we are allowed to take the limit into the integral~\footnotemark[3]. Hence, for $k\ge 5$, we are led to conclude that $\tilde\omega^k\mathcal{G}_k(\tilde\omega)=o(\tilde\omega^{k-1})$ and, therefore, these can be neglected in the expansion of $J(\omega)$ that we are interested in to the order $\mathcal{O}(\tilde\omega^4)$ only.

Finally, we are left to discuss the terms in the expansion of Eq.~(\ref{eq:expansionJ}) which contain products of several $\mathcal{F}_k$'s with various $k$'s. To this end, we introduce the integrals
\begin{eqnarray} \label{eq:Gk-comb} \fl \qquad\quad
	\mathcal{G}_{k_1,\ldots,k_j}(\tilde\omega)={\rm Re\,}  \int_{0}^{\infty} d\lambda \, e^{i\tilde{\omega}\lambda}\lambda^{-2\tilde{\omega}^2}\widehat{\mathcal{F}}_{k_1}(\lambda)\widehat{\mathcal{F}}_{k_2}(\lambda)\cdots\widehat{\mathcal{F}}_{k_j}(\lambda),
\end{eqnarray}
where $j\ge 2$ and $k_1, k_2, \ldots, k_j\ge 2$, and $\widehat{\mathcal{F}}_{\ell}=\tilde{\mathcal{F}}_{\ell}$ for $\ell=2$ and $\widehat{\mathcal{F}}_{\ell}=\mathcal{F}_{\ell}$ otherwise. Restricting the indices in Eq.~(\ref{eq:Gk-comb}) to $k=\sum_{\ell=1}^j k_\ell$, we may use techniques similar to those in paragraph (iv) to show that contributions of such terms to $J(\tilde\omega)$ are of the order $\tilde\omega^k \mathcal{G}_{k_1,\ldots,k_j}(\tilde\omega)=\mathcal{O}(\omega^{k-1})$; therefore, for $k\ge5$, they can safely be neglected.

Summarizing the calculation (i) to (vi), we conclude that
\begin{eqnarray} \fl\qquad \label{Jw}
    J(\tilde\omega) =
        J_1(\tilde\omega) + \tilde\omega^2 J_2(\tilde\omega) + \frac{\tilde\omega^4}{2} J_3(\tilde\omega)
        + \tilde\omega^4 J_4(\tilde\omega) - \tilde\omega^3 J_5(\tilde\omega)+\mathcal{O}(\tilde\omega^4)
\end{eqnarray}
admits the small-$\omega$ expansion in the form
\begin{eqnarray} \label{Jw-exp}
    J(\tilde\omega) = \pi \tilde\omega + 2\pi \tilde\omega^3 \log\tilde \omega + {\mathcal O}(\tilde\omega^4) \quad {\rm as} \quad \tilde\omega\rightarrow 0,
\end{eqnarray}
see Eqs.~(\ref{J1-exp}), (\ref{J2-exp}), (\ref{J3-exp}), (\ref{J4-exp}) and (\ref{J5-exp}). Finally, substituting Eq.~(\ref{Jw-exp}) into Eq.~(\ref{SwJ}) with expanded denominator
\begin{eqnarray}
    \frac{1}{4\pi \sin^2(\pi\tilde\omega)} = \frac{1}{4 \pi ^3 \tilde\omega^2}+\frac{1}{12 \pi }+O\left(\tilde\omega^2\right),
\end{eqnarray}
we recover the sought result Eq.~(\ref{S-res-0}).
\end{proof}

\section{Equivalence of ${\rm TCUE}(\infty)$ and ${\rm CUE}(\infty)$ power spectra}\label{CUE-TCUE-equiv}

In our previous studies~\cite{ROK-2020,ROK-2017}, we presented a {\it numerical} evidence that the power spectrum in the large-dimensional ${\rm CUE}(N)$ ensemble is very well described by the theoretical power spectrum derived for ${\rm TCUE}(\infty)$. This has naturally been attributed to the anticipated universality phenomenon which should emerge in the limit $N \rightarrow \infty$. Below, we show {\it analytically} that the power spectra in ${\rm CUE}(\infty)$ and ${\rm TCUE}(\infty)$ are indeed described by the same universal law, see Theorem~\ref{tcue-pv}.

\subsection{Power spectrum for ${\rm TCUE}(N)$: Alternative representation}

The ``tuned'' circular unitary ensemble ${\rm TCUE}(N)$ is obtained from the traditional circular unitary ensemble ${\rm CUE}(N+1)$ by conditioning its lowest eigen-angle to stay at zero. This procedure produces the ${\rm TCUE}(N)$ joint probability density of $N$ eigen-angles $\{\theta_j\}_{j=1}^N \in [0,2\pi)$ of the form~\cite{ROK-2020}
\begin{equation}\label{T-CUE} \fl \qquad
    P_N^{\rm TCUE}(\theta_1,\dots,\theta_N) = \frac{1}{(N+1)!} \prod_{1 \le j < k \le N}^{}
    \left| e^{i\theta_j} - e^{i\theta_k} \right|^2
    \prod_{j=1}^{N} \left| 1 - e^{i\theta_j}\right|^2.
\end{equation}
The normalization is fixed by
\begin{eqnarray}\label{TCUE-norm}
\prod_{j=1}^{N}\int_0^{2\pi} \frac{d\theta_j}{2\pi}\, P_{N}^{\rm TCUE}(\theta_1,\dots,\theta_N)=1.
\end{eqnarray}
Such a seemingly minor tuning of ${\rm CUE}(N+1)$ to ${\rm TCUE}(N)$ induces {\it stationarity of level spacings} in ${\rm TCUE}(N)$ for any $N \in {\mathbb N}$, the property which is {\it not} shared by the traditional circular unitary ensemble.

Stationarity of level spacings inherent in ${\rm TCUE}(N)$ allowed us to prove the following theorem (Theorem~2.7 in Ref.~\cite{ROK-2020}).
\begin{proposition}\label{tuned-circular-theorem}
  Let $\{0\le \theta_1 \le \dots \le \theta_N < 2\pi\}$ be a sequence of $N \in {\mathbb N}$ {\it ordered} fluctuating eigen-angles drawn from the ${\rm TCUE}(N)$. For all $\,0 <  \omega \le \pi$, the ensemble averaged power spectrum of eigen-angles admits the representation
\begin{eqnarray} \fl \label{ps-tuned-circular}
\quad
    S_N^{\rm TCUE}(\omega) =
    - \frac{(N+1)^2}{\pi N} \frac{z}{(1-z)^2}{\rm Re} \left\{
    (1-z) \left( N- z \frac{\partial}{\partial z}\right) - z^{-N}\right\}
        \nonumber\\
        \times \int_0^{2\pi} \frac{d\varphi}{2\pi} \,\varphi \, \Phi_N^{\rm TCUE}((0,\varphi);1-z) - \dbtilde{S}_N(\omega)
\end{eqnarray}
with
\begin{eqnarray} \label{ps-tcue-3} \fl\qquad
    \dbtilde{S}_N(\omega) =
    \frac{z}{(1-z)^2} \left\{ 1 - 2{\rm Re\,}\frac{z^{N+1}}{1-z} + \frac{1}{N}
    \left( 1 + \frac{2z}{(1-z)^2}  \left(1-{\rm Re\,}z^N\right)\right)\right\}.
\end{eqnarray}
Here, $z=e^{i\omega}$ whilst $\Phi_N^{\rm TCUE}((0,\varphi);\zeta)$ is the generating function [Eq.~(\ref{GF0-d})] of the probabilities to find a given number of ${\rm TCUE}(N)$-eigen-angles in the interval $(0,\varphi)$. It equals
\begin{equation} \label{phin-tcue-painleve-6}
    \Phi_N^{\rm TCUE}((0,\varphi);\zeta) = \exp \left(
            -\int_{\cot(\varphi/2)}^{\infty} \frac{dt}{1+t^2} \left( \tilde{\sigma}_N(t;\zeta) + t \right)
        \right),
\end{equation}
where the six Painlev\'e function $\tilde{\sigma}_N(t;\zeta)$ satisfies the nonlinear equation
\begin{eqnarray} \label{pvi-tcue} \fl
    \qquad \left( (1+t^2)\,\tilde{\sigma}_N^{\prime\prime} \right)^2 + 4 \tilde{\sigma}_N^\prime (\tilde{\sigma}_N - t \tilde{\sigma}_N^\prime)^2
    + 4 (\tilde{\sigma}_N^\prime+1)^2 \left(
        \tilde{\sigma}_N^\prime + (N+1)^2
    \right) = 0
\end{eqnarray}
and the boundary condition ($\zeta=1-z$)
\begin{eqnarray} \label{pvi-bc-tcue}
    \tilde{\sigma}_N(t;\zeta) = -t + \frac{N(N+1)(N+2)}{3\pi t^2} \zeta + {\mathcal O}(t^{-4})
\end{eqnarray}
as $t\rightarrow \infty$.
\end{proposition}

To analyse the power spectrum $S_N^{\rm TCUE}(\omega)$ in the limit $N\rightarrow \infty$, which is the final goal of this Section, it is useful to study exact relations between the generating functions
\begin{eqnarray} \label{Phi-N-TCUE} \fl \qquad\quad
    \Phi_N^{\rm TCUE}((0,\varphi);\zeta) =
    \frac{1}{(N+1)!}
    \prod_{j=1}^N \left( \int_0^{2\pi} - \zeta \int_0^\varphi \right) \frac{d\theta_j}{2\pi} \nonumber\\
    \fl\qquad\qquad\qquad\qquad\quad
    \times \prod_{j=1}^{N} \left| 1 - e^{i\theta_j}\right|^2
    \prod_{1 \le j < k \le N}^{}
    \left| e^{i\theta_j} - e^{i\theta_k} \right|^2
\end{eqnarray}
and
\begin{eqnarray} \fl\qquad \label{Phi-N-CUE}
    \Phi_N^{\rm CUE}((0,\varphi);\zeta) = \frac{1}{N!}
    \prod_{j=1}^N \left( \int_0^{2\pi} - \zeta \int_0^\varphi \right) \frac{d\theta_j}{2\pi}
    \prod_{1 \le j < k \le N}^{}
    \left| e^{i\theta_j} - e^{i\theta_k} \right|^2,
\end{eqnarray}
the latter being previously denoted $\Phi_N((0,\varphi);\zeta)$, see Eq.~(\ref{Phin-cue}).

Our proofs of Lemma~\ref{Lemma-Phi} and Corollary~\ref{cor-PVI} were inspired by the ideas outlined in Refs.~\cite{FW-2004,WF-2000}.
\begin{lemma}\label{Lemma-Phi}
Let $\Phi_N^{\rm CUE}((0,\varphi);\zeta)$ and $\Phi_N^{\rm TCUE}((0,\varphi);\zeta)$ be generating functions as specified above. Two following relations hold:
\begin{eqnarray} \label{prop-claim}
 \frac{d}{d\varphi} \Phi_N^{\rm CUE}((0,\varphi);\zeta) = -\frac{N}{2\pi} \zeta \Phi_{N-1}^{\rm TCUE}((0,\varphi);\zeta)
\end{eqnarray}
and
\begin{eqnarray} \fl \label{prop-rel} \qquad\quad
    \int_{0}^{2\pi} \frac{d\varphi}{2\pi}\, \varphi \Phi_N^{\rm TCUE}((0,\varphi); \zeta) \nonumber\\
    \fl \qquad\qquad\qquad
    = \frac{2\pi}{N+1}\frac{1}{\zeta}
    \left\{
    \int_{0}^{2\pi} \frac{d\varphi}{2\pi}\, \Phi_{N+1}^{\rm CUE}((0,\varphi); \zeta)
    -  (1-\zeta)^{N+1}\right\}.
\end{eqnarray}
\end{lemma}
\begin{proof}
Differentiating Eq.~(\ref{Phi-N-CUE}) with respect to $\varphi$, we observe
\begin{eqnarray} \label{CUE-N-der} \fl \quad
    \frac{d}{d\varphi}\Phi_N^{\rm CUE}((0,\varphi); \zeta) = -\frac{\zeta N}{2\pi} \nonumber\\
    \fl
    \qquad
    \times\frac{1}{N!}\prod_{j=1}^{N-1} \left( \int_{0}^{2\pi} - \zeta \int_{0}^{\varphi}\right) \frac{d\theta_j}{2\pi}
    \prod_{j=1}^{N-1}\left|e^{i\varphi} -e^{i\theta_j}\right|^2
        \prod_{1 \le j<k\le N-1}^{}\left|e^{i\theta_j} -e^{i\theta_k}\right|^2.
\end{eqnarray}
Changing all integration variables $\theta_j^\prime = \varphi-\theta_j$, and having in mind the $2\pi$-periodicity of the integrand, we reduce the integral above to the form
\begin{eqnarray} \label{CUE-N-der-2}
    \left( \int_{0}^{2\pi} - \zeta \int_{0}^{\varphi}\right) \frac{d\theta_j^\prime}{2\pi}
    \prod_{j=1}^{N-1}\left|1 -e^{i\theta_j^\prime}\right|^2
        \prod_{1 \le j<k\le N-1}^{}\left|e^{i\theta_j^\prime} -e^{i\theta_k^\prime}\right|^2.
\end{eqnarray}
Identifying this object as $\Phi_{N-1}^{\rm TCUE}((0,\varphi);\zeta)$, we conclude the proof of Eq.~(\ref{prop-claim}).

To prove the relation Eq.~(\ref{prop-rel}), we make use of Eq.~(\ref{prop-claim}) taken at $N+1$ and perform integration by parts in its l.h.s. Further, we spot that
$\Phi_{N+1}^{\rm CUE}((0,2\pi); \zeta)=(1-\zeta)^{N+1}$ to end the proof.
\end{proof}

\begin{corollary}\label{cor-PVI}
Let $\sigma_N(t;\zeta)$ and $\tilde\sigma_N(t;\zeta)$ be the solutions to the $\sigma$-Painlev\'e VI equations
\begin{eqnarray}\label{p6-1}\fl \qquad
    \left((1+t^2) \sigma_N^{\prime\prime}\right)^2 + 4 \sigma_N^\prime \left(
        \sigma_N - t\sigma_N^\prime
    \right)^2 + 4 \left(\sigma_N^\prime\right)^2 \left(\sigma_N^\prime + N^2 \right) = 0
\end{eqnarray}
and
\begin{eqnarray} \label{p6-2} \fl
    \qquad \left( (1+t^2)\,\tilde{\sigma}_N^{\prime\prime} \right)^2 + 4 \tilde{\sigma}_N^\prime (\tilde{\sigma}_N - t \tilde{\sigma}_N^\prime)^2
    + 4 (\tilde{\sigma}_N^\prime+1)^2 \left(
        \tilde{\sigma}_N^\prime + (N+1)^2
    \right) = 0
\end{eqnarray}
subject to the boundary conditions Eqs.~(\ref{pvi-bc}) and (\ref{pvi-bc-tcue}), respectively. It holds:
\begin{eqnarray} \label{p6-rel-2}
    \tilde\sigma_{N-1}(t) + t = (1+t^2) \frac{\sigma_N^\prime(t)}{\sigma_N(t)} + \sigma_N(t).
\end{eqnarray}
\end{corollary}

\begin{proof}
Substitute $\Phi_N^{\rm CUE}((0,\varphi); \zeta)$ given by Eq.~(\ref{phin-cue-painleve-6}) and $\Phi_N^{\rm TCUE}((0,\varphi); \zeta)$ given by Eq.~(\ref{phin-tcue-painleve-6})
into Eq.~(\ref{prop-claim}) to derive
\begin{eqnarray} \label{p6-rel-1}
\exp \left(
    -\int_{s}^{\infty} \frac{dt}{1+t^2} \, \left( \tilde\sigma_{N-1}(t) + t \right)
\right)\nonumber\\
\qquad \qquad
= \frac{\pi}{\zeta N}\, \sigma_N(s)
    \exp \left(
    -\int_{s}^{\infty} \frac{dt}{1+t^2} \, \sigma_N(t)
\right).
\end{eqnarray}
Next, differentiate it with respect to $s$ and use Eq.~(\ref{p6-rel-1}) again to cancel exponential terms. This brings the sought Eq.~(\ref{p6-rel-2}).
\end{proof}

Lemma~\ref{Lemma-Phi}, combined with Proposition~\ref{tuned-circular-theorem}, produces an alternative representation of the ${\rm TCUE}(N)$ power spectrum, as formulated below.

\begin{corollary}
Ensemble averaged power spectrum for the ${\rm TCUE}(N)$ eigen-angles admits the representation
\begin{eqnarray}  \label{ps-tcue-reduced} \fl \qquad
    S_N^{\rm{TCUE}}(\omega) =  -\frac{2(N+1)}{N} \frac{z}{(1-z)^2} \nonumber\\
    \times{\rm Re} \left( N- z \frac{\partial}{\partial z}  - \frac{z+z^{-N}}{1-z}\right)
        \int_0^{2\pi} \frac{d\varphi}{2\pi} \, \Phi_{N+1}^{\rm{CUE}}((0,\varphi);1-z) \nonumber\\
        \qquad \qquad \qquad \qquad -  \frac{2}{N} \frac{z^2}{(1-z)^4} {\rm Re\,}\left(1-z^{N+1}\right).
\end{eqnarray}
\end{corollary}
This representation will serve a basis for performing an asymptotic analysis of the ${\rm TCUE}(N)$ power spectrum as $N\rightarrow \infty$.

\subsection{Equivalence of power spectra in ${\rm TCUE}(\infty)$ and ${\rm CUE}(\infty)$}

\begin{theorem}\label{tcue-pv}
Let $S_N^{\rm TCUE}(\omega)$ denote the power spectrum of the ${\rm TCUE}(N)$. For all $0<\omega< \pi$, the limit $\lim_{N\rightarrow\infty} S_N^{\rm TCUE}(\omega)$ exists and equals
  \begin{eqnarray} \label{ps-tcue-inf}
    S_\infty^{\rm TCUE}(\omega) = S_\infty^{\rm CUE}(\omega),
\end{eqnarray}
where $S_\infty^{\rm CUE}(\omega)$ is defined in Theorem~\ref{cue-pv}.
\end{theorem}

\begin{proof}
We use the notation Eq.~(\ref{ps-alter-master}) to rewrite Eq.~(\ref{ps-tcue-reduced}), as $N\rightarrow\infty$ in the form
\begin{eqnarray}  \label{ps-tcue-reduced-st-1} \fl \quad
    S_N^{\rm{TCUE}}(\omega) =  -\frac{2}{N} \frac{z}{(1-z)^2} \nonumber\\
    \times{\rm Re} \left( N- z \frac{\partial}{\partial z}  - \frac{z+z^{-N}}{1-z}\right)
        \, \left(I_{N+1,0}(z) + z^{N+1} \overline{I_{N+1,0}(z)} \right) + {\mathcal O}(N^{-1}).\nonumber\\
\end{eqnarray}
To derive Eq.~(\ref{ps-tcue-reduced-st-1}), we have split the integration domain in Eq.~(\ref{ps-tcue-reduced}) into two intervals $(0,\pi) \cup (\pi, 2\pi)$ and made use of the symmetry relation Eq.~(\ref{comp-PhiN}), see the proof of Theorem~\ref{cue-pv} for a similar calculation. The ${\mathcal O}(N^{-1})$ term in Eq.~(\ref{ps-tcue-reduced-st-1}) originates from the last term in Eq.~(\ref{ps-tcue-reduced}).

It follows from Eq.~(\ref{eq:IN0-2}) and Lemma~\ref{lemma-tech-2} that
\begin{eqnarray} \label{LN0-T2}
    I_{N,0}(z) = L_\infty(z)\left( 1 + {\mathcal O}\big(
     \Omega(N)^{-2\tilde\omega^2}\big)\right),
\end{eqnarray}
where
\begin{eqnarray} \label{LN0-T3}
    L_\infty(z)= \frac{1}{2\pi}
    \int_{0}^{\infty} d\lambda
     \exp\left(\int_0^{-i \lambda} \frac{ds}{s} \sigma(s)\right).
\end{eqnarray}
Finally, we substitute Eq.~(\ref{LN0-T2}) into Eq.~(\ref{ps-tcue-reduced-st-1}) to derive
\begin{eqnarray}
    S_N^{\rm TCUE}(\omega) = -\frac{2z}{(1-z)^2}\, {\rm Re\,} L_\infty(z) + o(1).
\end{eqnarray}
In doing so, we have explicitly assumed~\footnote[4]{This assumption is not unreasonable in view of Remark~1.4 in Ref.~\cite{DIK-2014}.} that the differential operator $\partial/\partial z$ can directly act onto the asymptotic expansion of $I_{N,0}(z)$ given by Eq.~(\ref{LN0-T2}). As $N\rightarrow\infty$, this brings the statement of the Theorem.
\end{proof}

\section*{Acknowledgments}
The authors thank F.~Bornemann for providing us with the MATLAB package for numerical evaluation~\cite{B-2010} of Fredholm determinants. This work was supported by the Israel Science Foundation through the Grants No.~648/18 (E.K. and R.R.) and No.~2040/17 (R.R.). Some of the computations presented in this work were performed on the Hive computer cluster at the University of Haifa, which is partially funded through the ISF grant No.~2155/15.

\newpage
\renewcommand{\appendixpagename}{\normalsize{Appendices}}
\addappheadtotoc
\appendixpage
\renewcommand{\thesection}{\Alph{section}}
\renewcommand{\theequation}{\thesection.\arabic{equation}}
\setcounter{section}{0}

\section{Boundary conditions for Painlev\'e VI function ${\sigma}_N (t; \zeta)$ as $t \rightarrow \infty$} \label{A-1}
To derive the $t \rightarrow \infty$ boundary condition for ${\sigma}_N (t; \zeta)$ satisfying Eq.~(\ref{pvi}) of Theorem \ref{Th-PVI}, we make use of Eqs.~(\ref{phin-cue-painleve-6}) and (\ref{Phin-cue}) to observe
the relation
\begin{eqnarray}\label{sn-bc}
    {\sigma}_N(t;\zeta) = - 2 \frac{d}{d\varphi} \log \Phi_N((0,\varphi);\zeta)\Big|_{\varphi= 2\arctan(1/t)}
\end{eqnarray}
which holds true for $t >0$ and $0\le \varphi < \pi/2$. Since $\varphi \rightarrow 0$ as $t \rightarrow \infty$, we shall consider a small-$\varphi$
expansion of the generating function $\Phi_N((0,\varphi);\zeta)$ represented as
\begin{eqnarray} \label{ps-tcue-2A}
    \Phi_N((0,\varphi);\zeta) &=&
    \prod_{j=1}^N  \left( \int_0^{2\pi} - \zeta \int_0^\varphi \right) \frac{d\theta_j}{2\pi}
    P_N(\theta_1,\dots,\theta_N) \nonumber\\
    &=& 1+ \sum_{\ell=1}^N \frac{(-\zeta)^\ell}{\ell!} \left( \prod_{j=1}^\ell \int_{0}^{\varphi} \frac{d\theta_j}{2\pi} \right)
    R_{\ell,N} (\theta_1,\dots,\theta_\ell).
\end{eqnarray}
Here, the JPDF $P_N(\theta_1,\dots,\theta_N)$ is that of ${\rm CUE}(N)$ [Eq.~(\ref{jpdf-cue})], and $R_{\ell,N} (\theta_1,\dots,\theta_\ell)$ stands for
the $\ell$-th order correlation function [Eqs.~(\ref{RLN-CUE}) and (\ref{RLN-det})].

A small-$\varphi$ expansion of $\Phi_N((0,\varphi);\zeta)$ in Eq.~(\ref{ps-tcue-2A}) produces a sought boundary condition as $t\rightarrow \infty$. Several initial terms of the expansion read:
\begin{eqnarray} \fl \label{ps-tcue-2AA}
    \Phi_N((0,\varphi);\zeta) = 1 - \frac{\zeta}{2\pi} \left(
        R_{1,N} (0) \, \varphi + \frac{1}{2!} R_{1,N}^\prime (0) \, \varphi^2 + \frac{1}{3!} R_{1,N}^{\prime\prime} (0) \, \varphi^3
    \right) \nonumber\\
     +  \frac{1}{2!} \left(\frac{\zeta}{2\pi}\right)^2 \left( R_{2,N} (0,0) \, \varphi^2 +
        \frac{1}{2} \left[
            R_{2,N}^{[0,1]}(0,0) +  R_{2,N}^{[1,0]}(0,0)
        \right] \varphi^3
    \right) \nonumber\\
    - \frac{1}{3!} \left(\frac{\zeta}{2\pi}\right)^3  R_{3,N} (0,0,0) \, \varphi^3 +
    {\mathcal O}(\varphi^4).
\end{eqnarray}
Here, ${(\dots)}^\prime$ denotes a derivative; ${(\dots)}^{[\ell,m]}$ stands for a mixed derivative (of orders $\ell$ and $m$) of a function of two variables with respect to its first and second argument, respectively.

Equations (\ref{RLN-det}) and (\ref{CUEN-Kernel}) show that only two, out of seven, coefficients in the expansion are nontrivial,
\begin{eqnarray}
    R_{1,N}(0)&=&N, \nonumber\\
    R_{1,N}^{\prime\prime} (0) &=& -\frac{N(N^2-1)}{12}
\end{eqnarray}
yielding
\begin{eqnarray} \label{ps-tcue-2AAA}
    \Phi_N((0,\varphi);\zeta) = 1 - \frac{\zeta}{2\pi} \left( N\varphi -\frac{N(N^2-1)}{72}\varphi^3 \right) +
    {\mathcal O}(\varphi^4).
\end{eqnarray}
By virtue of Eq.~(\ref{sn-bc}), the boundary condition for ${\sigma}_N (t; \zeta)$ as $t \rightarrow \infty$ reads
\begin{eqnarray}\label{1texp} \fl \qquad\qquad
    {\sigma}_N (t; \zeta) =  \frac{N\zeta}{\pi} + \frac{N^2\zeta^2}{\pi^2 t} + \frac{N \zeta}{\pi t^2}\left[
    N^2 \left( \frac{\zeta^2}{\pi^2} - \frac{1}{6}\right) +\frac{1}{6} \right]
    + {\mathcal O}(t^{-3}).
\end{eqnarray}

\section{Boundary conditions for Painlev\'e V transcendent ${\sigma}_0 (t;\zeta)$ as $t \rightarrow 0$} \label{A-2}
To derive the $t \rightarrow 0$ boundary condition for ${\sigma}_0 (t;\zeta)$ satisfying Eq.~(\ref{PV-eq-0}), we make use of the celebrated formula relating $\sigma_0(t;\zeta)$ to the Fredholm determinant
\begin{eqnarray} \label{B-1}
    \exp\left({\int_{0}^{\lambda} \frac{dt}{t} \sigma_0(t;\zeta)}\right) = {\rm det\,} \left[
        \mathds{1} - \zeta \hat{K}_{1/2\pi,\infty}^{(0,\lambda)}
    \right],
\end{eqnarray}
see Eqs.~(\ref{relation-to-sine-process}), (\ref{sine-operator}) and (\ref{Sine-Kernel}) taken at $\rho=1/2\pi$. Notice that the Fredholm
determinant Eq.~(\ref{B-1}) is an entire function in $\zeta$ (see Ref.~\cite{F-1903}) as well as in $\lambda \in \mathbb{C}$ (see Ref.~\cite{AGZ-2010}).

Combined with the integral expansion Eq.~(\ref{FD-inf}) of the Fredholm determinant,
\begin{eqnarray} \fl \quad
   {\rm det\,} \left[
        \mathds{1} - \zeta \hat{K}_{1/2\pi,\infty}^{(0,\lambda)}
    \right] = 1 + \sum_{\ell=1}^{\infty} \frac{(-\zeta)^\ell}{\ell!} \prod_{i=1}^{\ell} \int_{0}^{\lambda} dx_i \,
    {\rm det}_{1 \le j,k \le \ell} \left[
        {K}_{1/2\pi,\infty} (x_j - x_k)
    \right]\nonumber\\
    \fl \quad\qquad = 1 -\zeta \int_{0}^{\lambda} dx_1 {K}_{1/2\pi,\infty} (0) \nonumber\\
    \fl \quad\qquad \qquad + \frac{\zeta^2}{2}
    \int_{0}^{\lambda} \int_{0}^{\lambda} dx_1 dx_2 \left[
       {K}_{1/2\pi,\infty}^2(0)  - {K}_{1/2\pi,\infty}^2(x_1 -x_2)
    \right] + \dots \nonumber\\
    \fl \quad\qquad \qquad = 1 - \frac{\zeta\lambda}{2\pi} + \frac{\zeta^2}{8\pi^2}
    \int_{0}^{\lambda} \int_{0}^{\lambda} dx_1 dx_2 \left[
       1 - \left(\frac{\sin[(x_1-x_2)/2]}{(x_1-x_2)/2}\right)^2
    \right] + \dots,
\end{eqnarray}
this brings a Taylor series expansion
\begin{eqnarray} \label{s_0_polynomial_ex} \fl\qquad
    \sigma_0(\lambda;\zeta) = \lambda \frac{d}{d\lambda} \log {\rm det\,} \left[
        \mathds{1} - \zeta \hat{K}_{1/2\pi,\infty}^{(0,\lambda)}
    \right] = -\frac{\zeta\lambda}{2\pi} - \frac{\zeta^2\lambda^2}{4\pi^2} +{\mathcal O}(\zeta^3\lambda^3)
\end{eqnarray}
as $\lambda\rightarrow 0$.

Calculation of higher order terms in Eq.~(\ref{s_0_polynomial_ex}) becomes increasingly cumbersome. This difficulty can nevertheless be circumvented by substitution of the sought Taylor series~\footnote[5]{Since the Fredholm determinant in Eq.~(\ref{B-1}) is an analytic function of $\lambda$ that equals unity at $\lambda=0$, its logarithm in Eq.~(\ref{s_0_polynomial_ex}) is well defined and admits a power series expansion in the vicinity of $\lambda=0$ with a non-zero convergence radius. This allows us to differentiate the log in that vicinity thus justifying a small-$\lambda$ expansion of $\sigma_0(\lambda;\zeta)$ in a vicinity of $\lambda=0$.}
\begin{eqnarray}\label{s0-pol}
    \sigma_0(\lambda;\zeta) = \sum_{k=1}^{\infty} \left(\frac{\lambda}{2\pi}\right)^k g_k(\zeta),
\end{eqnarray}
where
\begin{eqnarray}
g_1(\zeta)=-\zeta,\\
g_2(\zeta)=-\zeta^2,
\end{eqnarray}
see Eq.~(\ref{s_0_polynomial_ex}), into the Chazy form~\cite{C-1911,C-2000,OK-2010} of the Painlev\'e V equation [Eq.~(\ref{PV-eq-0})]
\begin{eqnarray} \label{PV-Chazy}
    t^2 \sigma_0^{\prime\prime\prime} + t \sigma_0^{\prime\prime} + 6 t (\sigma_0^\prime)^2 -4 \sigma_0\sigma_0^\prime +t^2 \sigma_0^\prime -t\sigma_0 = 0.
\end{eqnarray}
This yields a recurrence relation
\begin{eqnarray} \fl \quad
    k(k-1)^2 g_k(\zeta) + 4\pi^2 (k-3) g_{k-2}(\zeta) + 2 \sum_{j=1}^{k-1}(3j-2)(k-j) g_j(\zeta) g_{k-j}(\zeta) = 0
\end{eqnarray}
for the functions $g_k(\zeta)$ for $k\ge 3$. Below we quote several of them:
\begin{eqnarray}
    g_3(\zeta)=-\zeta^3,\\
    g_4(\zeta)=-\zeta^4 + \frac{\pi^2}{9}\zeta^2,\\
    g_5(\zeta)=-\zeta^5 + \frac{5\pi^2}{36}\zeta^3,
\end{eqnarray}
etc.

\section{Cumulants of the counting function for the ${\rm Sine}_2$ point process} \label{A-3}
\subsection{Small-$\omega$ expansion of Painlev\'e V transcendent ${\sigma}_0 (t;\zeta)$}\label{A-3-1}
To analyse the ${\rm CUE}(\infty)$ power spectrum at small frequencies, we follow Ref.~\cite{ROK-2020,ROK-2017} to adopt the small-$\omega$ ansatz~\footnote[6]{To justify the expansion Eq.~(\ref{sigma-expan-app}), let us notice the relation
$$
    \sigma_0(\lambda;\zeta) = \lambda \frac{d}{d\lambda} \log \Phi_\infty(\lambda;\zeta).
$$
Since $\Phi_\infty(\lambda;\zeta)$ is the Fredholm determinant Eq.~(\ref{FD-inf-FD}), it is an entire function in each variable~\cite{F-1903,AGZ-2010} $\zeta\in{\mathbb C}$ and $\lambda\in{\mathbb C}$, separately. By Hartogs' theorem and Osgood's lemma, the Fredholm determinant is a holomorphic function on ${\mathbb C}^2$. As such, it can be written as a power series in two variables,
$$
    \Phi_\infty(\lambda;\zeta) = 1 + \sum_{j,k=1}^{\infty} \alpha_{j,k} \lambda^j \zeta^k
$$
with an infinite convergence radius in both $\lambda$ and $\zeta$. Its logarithm, $\log \Phi_\infty(\lambda;\zeta)$, can then be written down as a double power
series with a non-zero convergence radius $R_{\lambda,\zeta}$ in both variables. Within the convergence domain, differentiation w.r.t. $\lambda$ is allowed; this will result in a double power series for $d/d\lambda \log \Phi_\infty(\lambda;\zeta)$ with the same convergence radius $R_{\lambda,\zeta}$. In view of the first relation of this footnote, this
implies that $\sigma_0(\lambda;\zeta)$ admits a power series expansion in $\zeta$ with a non-zero convergence radius $R_\zeta \neq 0$. This justifies a small-$\omega$ expansion Eq.~(\ref{sigma-expan-app}) as $\zeta= 1- e^{i\omega}$.
}
\begin{eqnarray} \label{sigma-expan-app}
    \sigma_0(t;\zeta) = \sum_{k=1}^\infty \tilde{\omega}^k f_k(t)
\end{eqnarray}
for the solution to the Painlev\'e V equation Eq.~(\ref{PV-eq-0}) supplemented by the boundary condition Eq.~(\ref{bc-zero-0}), see Theorem~\ref{cue-pv}.

The main objective of this subsection is to study the expansion coefficients $f_k(t)$. In particular, we shall calculate the functions $f_1(t)$, $f_2(t)$ and $f_3(t)$ explicitly, find an integral representation for $f_4(t)$, study asymptotic properties of $f_1(t)$, $f_2(t)$, $f_3(t)$ and $f_4(t)$ as $t\rightarrow 0$ and $t\rightarrow\infty$, and provide non-sharp estimates for an asymptotic behavior of $f_k(t)$ for $k\ge 5$.

To start with, we substitute the ansatz Eq.~(\ref{sigma-expan-app}) into the Chazy form [Eq.~(\ref{PV-Chazy})] of the fifth Painlev\'e transcendent to generate a set of linear differential equations for the coefficients $f_k(t)$ with $k=1, 2, \dots$,
\begin{eqnarray} \label{fk-eqns}
    t f_k^{\prime\prime\prime} + f_k^{\prime\prime} + t f_k^{\prime} -  f_k = \frac{F_k(t)}{t},
\end{eqnarray}
where
\begin{eqnarray}
    F_k(t) = 2 \sum_{\ell=1}^{k-1} \left\{
        \frac{d}{dt}(f_\ell f_{k-\ell}) - 3 t f_\ell^\prime f_{k-\ell}^\prime
    \right\}.
\end{eqnarray}
In particular,
\begin{eqnarray}\label{FL-1}
    F_1(t) &=& 0, \\
    \label{FL-2}
    F_2(t) &=& 4 f_1\, f_1^\prime
    - 6 t (f_1^\prime)^2, \\
    \label{FL-3}
    F_3(t) &=& 4 \left( f_1 f_2^\prime +  f_1^\prime f_2\right) - 12 t f_1^\prime f_2^\prime,\\
    \label{FL-4}
    F_4(t) &=& 4 \left( f_1 f_3^\prime + f_2 f_2^\prime + f_1^\prime f_3\right)
            -6 t \left(  (f_2^\prime)^2 + 2 f_1^\prime f_3^\prime\right).
\end{eqnarray}
Differential equations Eq.~(\ref{fk-eqns}) have to be supplemented by boundary conditions.

(i) The function $f_1(t)$ can be found from the general solution to Eqs.~(\ref{fk-eqns}) and (\ref{FL-1})
\begin{eqnarray} \label{fk-gs}
\fl \qquad
    f_G (t; c_1,c_2,c_3) = c_1 t  + c_2 a(t)
    + c_3 b(t),
\end{eqnarray}
where
\begin{eqnarray}\label{at}
a(t) &=&  \cos t + t {\rm Si} (t) -\frac{\pi t}{2}, \\
\label{bt}
b(t) &=&  \sin t - t {\rm Ci}(t).
\end{eqnarray}
The coefficients $c_{1}, c_2,$ and $c_3$ should be fixed by meeting specific boundary conditions. Here, ${\rm Si}(t)$ and ${\rm Ci}(t)$ are the sine and cosine integrals. Prior to doing so, let us notice the following limiting behaviors:
\begin{eqnarray}\label{Y2-L}
        a(t) =
        \left\{
                   \begin{array}{ll}
                     \displaystyle 1+ {\mathcal O}(t), & \hbox{$t\rightarrow 0$;} \\
                     \displaystyle -\frac{\sin t}{t}+{\mathcal O} (1/t^2), & \hbox{$t\rightarrow \infty$,}
                   \end{array}
                 \right.
\end{eqnarray}
and
\begin{eqnarray}\label{Y3-L}
        b(t)  =
        \left\{
                   \begin{array}{ll}
                     \displaystyle - t\log t +{\mathcal O}(t), & \hbox{$t\rightarrow 0$;} \\
                     \displaystyle \frac{\cos t}{t}+{\mathcal O}(1/t^2), & \hbox{$t\rightarrow \infty$.}
                   \end{array}
                 \right.
\end{eqnarray}
The boundary conditions for $f_1(t)$ can be read off from Eqs.~(\ref{eq:bc-inf-Th}) and (\ref{bc-zero-0}): $f_1(t)\rightarrow 0$ as $t\rightarrow 0$ and $f_1(t) = it + o(t)$ as $t\rightarrow\infty$. Owing to Eqs.~(\ref{Y2-L}) and (\ref{Y3-L}), the boundary condition at infinity yields $c_1 = i$ while that at zero brings $c_2=0$ leaving a constant $c_3$ unidentified. To fix it, we turn to our observation Eq.~(\ref{s0-pol}) implying that $\sigma_0(t)$, and thus $f_1(t)$, admits a {\it polynomial} in $t$ expansion in a vicinity of $t=0$. This makes us set $c_3=0$ to avoid appearance of $t\log t$ terms in $f_1(t)$ as $t\rightarrow 0$, see Eq.~(\ref{Y3-L}). As the result, we obtain
\begin{eqnarray} \label{f1-final}
    f_1(t) = f_G (t; i,0,0)= it.
\end{eqnarray}

(ii) The function $f_2(t)$ is a solution to the nonhomogeneous differential equation Eq.~(\ref{fk-eqns}) with $F_2(t) = 2t$, see Eqs.~(\ref{FL-2}) and (\ref{f1-final}), which can be represented as a sum of the particular solution
\begin{eqnarray}\label{f2-part}
    f_2^{(p)} (t) = -2+ 2 \left(
        \cos t + t{\rm Si}(t) -\frac{\pi t}{2}
    \right)
\end{eqnarray}
and of a general solution $f_G(t; c_1, c_2, c_3)$. To fix the three constants, we turn to the boundary conditions for $f_2(t)$ that follow from Eqs.~(\ref{eq:bc-inf-Th}) and (\ref{bc-zero-0}): $f_2(t)\rightarrow 0$ as $t\rightarrow 0$ and $f_2(t) \rightarrow -2$ as $t\rightarrow\infty$. Making use of the expansion Eq.~(\ref{Y2-L}), we see that they are satisfied by the particular solution Eq.~(\ref{f2-part}). This implies that the constants in the general solution $f_G(t; c_1, c_2, c_3)$ should be chosen in such a way that $f_G \rightarrow 0$ for both $t\rightarrow 0$ and $t\rightarrow\infty$, yet its small-$t$ behavior should be of the polynomial type. This brings $c_1=c_2=c_3=0$. Hence,
\begin{eqnarray}\label{f2-final}
    f_2 (t) = f_2^{(p)}(t).
\end{eqnarray}
For further reference, we quote asymptotic expansions of $f_2(t)$:
\begin{eqnarray}\label{f2-at-zero}
    f_2(t) = -\pi t + t^2 -\frac{1}{36}t^4 +{\mathcal O}(t^6)\quad {\rm as} \quad t\rightarrow 0,
\end{eqnarray}
\begin{eqnarray}\label{f2-at-infty}
    f_2(t) = -2 -2 \frac{\sin t}{t}+ 4\frac{\cos t}{t^2}  +{\mathcal O}(t^{-3}) \quad {\rm as} \quad t\rightarrow \infty.
\end{eqnarray}

(iii) The function $f_3(t)$ is a solution to the nonhomogeneous differential equation Eq.~(\ref{fk-eqns}) with $F_3(t)$ given by Eqs.~(\ref{FL-3}), (\ref{f2-part}) and (\ref{f1-final}) subject to the boundary conditions $f_3(t) \rightarrow 0$ for both $t\rightarrow 0$ and $t\rightarrow \infty$; yet its small-$t$ behavior should be of the polynomial type. Following the reasoning in (ii), we represent the sought solution as a sum of the particular solution
\begin{eqnarray}\label{f3-part}
    f_3^{(p)} (t) &=& i t^3 \, {}_3F_4\left(1,1,1;\frac{3}{2},2,2,2;-\frac{t^2}{4}\right) - 4i\pi \left(1 - \cos t \right) \nonumber \\
                  &+& 8 i \sin t \left(  {\rm Ci}(t) -\log t -\gamma \right) + 8 i {\rm Si} (t)\, \left(1 - \cos t \right) \nonumber \\
                  &-& 8 i t {\rm Ci}(t) \left( {\rm Ci}(t) -\log t -\gamma \right)\nonumber \\
                  &+& 4 i t \left( {\rm Ci}(t)^2 - {\rm Si}(t)^2 -\gamma^2 -\log^2 t - 2 \gamma\log t\right) \nonumber\\
                  &+& 4 i \pi  t {\rm Si}(t) -\frac{2}{3} i \pi ^2 t
\end{eqnarray}
and of a general solution $f_G(t; c_1, c_2, c_3)$. Here, ${}_p F_q$ is the hypergeometric function. Since $f_3^{(p)}(t) \rightarrow 0$ for both $t\rightarrow 0$ and $t\rightarrow \infty$,
we conclude that all three constants of a general solution should nullify. Hence,
\begin{eqnarray}\label{f3-final}
    f_3 (t) = f_3^{(p)}(t).
\end{eqnarray}
The following asymptotic expansions hold:
\begin{eqnarray}\label{f3-at-zero}\fl\quad
    f_3(t) = -\frac{2\pi^2}{3} i t + 2\pi i t^2 -i t^3 -\frac{\pi}{18}it^4+ {\mathcal O}(t^5) \quad {\rm as} \quad t\rightarrow 0,
\end{eqnarray}
\begin{eqnarray}\label{f3-at-infty}\fl\quad
    f_3(t) = \frac{4i}{t} -  8i\frac{\cos t}{t}(\gamma +\log t)  - 16 i \frac{\sin t}{t^2}(\log t +\gamma-1)
    +{\mathcal O}\left(\frac{\log t}{t^3}\right) \quad {\rm as} \quad t\rightarrow \infty. \nonumber\\
    {}
\end{eqnarray}

(iv) The function $f_4(t)$ can be calculated along the same lines albeit more technical effort is required. We present the final answer in the form
\begin{eqnarray} \fl \qquad \label{f4t-def}
    f_4(t) = \frac{\pi^3}{3} t + a(t)  \int_{0}^{t} \frac{dx}{x}\, F_4(x) \sin x - b(t) \int_{0}^{t} \frac{dx}{x}\, F_4(x) \cos x \nonumber\\
    \fl\qquad\qquad\qquad
    - t \int_{0}^{t} \frac{dx}{x}\, F_4(x) \left[
    \cos x {\rm Ci}(x) + \sin x {\rm Si}(x) - \frac{\pi}{2}\sin x\right],
\end{eqnarray}
where $a(t)$ and $b(t)$ exhibit an asymptotic behavior specified in Eqs.~(\ref{Y2-L}) and (\ref{Y3-L}), and $F_4(x)$ is defined by Eqs.~(\ref{FL-4}), (\ref{f3-part}), (\ref{f2-part}) and (\ref{f1-final}).

(iv-a)~The small-$t$ expansion of $f_4(t)$ can readily be evaluated:
\begin{eqnarray}\label{f4-at-zero}
    f_4(t) = \frac{\pi^3}{3} t - \frac{7\pi^2}{3} t^2 + \left(
        \frac{7\pi^2}{108} -1
    \right) t^4 + {\mathcal O}(t^5) \quad {\rm as} \quad t\rightarrow 0.
\end{eqnarray}

(iv-b) Finding an expansion of $f_4(t)$ at infinity requires some effort. To start with, we notice that
\begin{eqnarray}\label{F4-est-1}
    F_4(x) &=& F_\infty(x) - \frac{32 \cos x \log x}{x}+ {\mathcal O}\left(
        \frac{1}{x}
    \right) \quad {\rm as} \quad x\rightarrow \infty, \\
    \label{F4-inf}
    F_\infty(x) &=& 64  (\gamma+\log x) \sin x,
\end{eqnarray}
in order to re-group the terms in $f_4(t)$, Eq.~(\ref{f4t-def}), as follows:
\begin{eqnarray} \fl \qquad \label{f4t-02}
    f_4(t) = t\left\{ \frac{\pi^3}{3} -
    \int_{0}^{\infty} \frac{dx}{x}\, F_4(x) \left[
    \cos x {\rm Ci}(x) + \sin x {\rm Si}(x) - \frac{\pi}{2}\sin x\right] \right\}
    \nonumber\\
    + a(t)  \int_{0}^{\infty} \frac{dx}{x}\, \left[ F_4(x) -F_\infty(x)\right] \sin x  \nonumber\\
    - a(t)  \int_{t}^{\infty} \frac{dx}{x}\, \left[ F_4(x) -F_\infty(x)\right] \sin x\nonumber\\
    \fl\qquad\qquad\qquad
    + a(t) \int_{0}^{t} \frac{dx}{x}\, F_\infty(x) \sin x \nonumber\\
    - b(t) \int_{0}^{\infty} \frac{dx}{x}\, F_4(x) \cos x + b(t) \int_{t}^{\infty} \frac{dx}{x}\, F_4(x) \cos x\nonumber\\
    + t \int_{t}^{\infty} \frac{dx}{x}\, F_4(x) \left[
    \cos x {\rm Ci}(x) + \sin x {\rm Si}(x) - \frac{\pi}{2}\sin x\right].
\end{eqnarray}
Let us analyze each term in Eq.~(\ref{f4t-02}) separately.

The first, linear in $t$ term {\it must} vanish according to the boundary condition
Eq.~(\ref{eq:bc-inf-Th}). This implies, that
\begin{eqnarray}
    \int_{0}^{\infty} \frac{dx}{x}\, F_4(x) \left[
    \cos x {\rm Ci}(x) + \sin x {\rm Si}(x) - \frac{\pi}{2}\sin x\right]  = \frac{\pi^3}{3}.
\end{eqnarray}

The second term,
\begin{eqnarray}
    a(t)  \int_{0}^{\infty} \frac{dx}{x}\, \left[ F_4(x) -F_\infty(x)\right] \sin x  = C_1 \frac{\sin t}{t} + {\mathcal O}\left(
        \frac{1}{t^2}
    \right)
\end{eqnarray}
since the integral therein converges to some constant denoted by $-C_1$; for $a(t)$ as $t\rightarrow \infty$, see Eq.~(\ref{Y2-L}).

The third term,
\begin{eqnarray}
    a(t)  \int_{t}^{\infty} \frac{dx}{x}\, \left[ F_4(x) -F_\infty(x)\right] \sin x = {\mathcal O}\left(
        \frac{\log t}{t^2}\right),
\end{eqnarray}
where we have used the estimates Eqs.~(\ref{Y2-L}) and (\ref{F4-est-1}).

The integral in the fourth term can be evaluated explicitly and further expanded. This yields the estimate
\begin{eqnarray} \fl\quad
    a(t)\int_{0}^{t} \frac{dx}{x}\, F_\infty(x) \sin x =
    \frac{\sin t}{t}\left( C_2 (\log t)^2 + C_3 \log t + C_4\right) +{\mathcal O}\left(
    \frac{(\log t)^2}{t^2}
    \right),
\end{eqnarray}
where $C_2$, $C_3$ and $C_4$ are some constants.

The estimate of the fifth term reads
\begin{eqnarray}
    b(t) \int_{0}^{\infty} \frac{dx}{x}\, F_4(x) \cos x = C_5 \frac{\cos t}{t} + {\mathcal O}\left(\frac{1}{t^2}\right),
\end{eqnarray}
where we have used Eq.~(\ref{Y3-L}) and the fact that the integral itself converges to a constant $C_5$.

The sixth term brings
\begin{eqnarray}
    b(t) \int_{t}^{\infty} \frac{dx}{x}\, F_4(x) \cos x = {\mathcal O}\left(\frac{\log t}{t^2}\right).
\end{eqnarray}
Here, the expansion Eq.~(\ref{F4-est-1}) was used.

As $t\rightarrow \infty$, the last, seventh, term
\begin{eqnarray} \label{term-7}
    t \int_{t}^{\infty} \frac{dx}{x}\, F_4(x) \left[
    \cos x {\rm Ci}(x) + \sin x {\rm Si}(x) - \frac{\pi}{2}\sin x\right]
\end{eqnarray}
can be estimated by expanding the integrand
\begin{eqnarray}
    \cos x {\rm Ci}(x) + \sin x {\rm Si}(x) - \frac{\pi}{2}\sin x = {\mathcal O}(x^{-2}) \quad {\rm as} \quad x\rightarrow\infty,
\end{eqnarray}
see also Eqs.~(\ref{F4-est-1}) and (\ref{F4-inf}). Hence, as $t\rightarrow \infty$, we need to consider a simpler integral
\begin{eqnarray}
    t \int_{t}^{\infty} \frac{dx}{x^3} \log x \sin x = {\mathcal O}\left( \frac{\log t}{t^2}\right).
\end{eqnarray}

Finally, combining all the estimates together, we conclude that
\begin{eqnarray} \label{f4t-final-inf-est} \fl\qquad
    f_4(t) = \frac{\sin t}{t}\left( C_2 (\log t)^2 + C_3 \log t + (C_1+C_4)\right) \nonumber\\
    \fl \qquad \qquad\qquad\qquad\qquad\qquad + C_5 \frac{\cos t}{t} +{\mathcal O}\left(
    \frac{(\log t)^2}{t^2}
    \right)
    \quad {\rm as} \quad t\rightarrow \infty.
\end{eqnarray}

(v) The asymptotic behavior of $f_k(t)$ for $k\ge 5$ can also be obtained. As $t\rightarrow 0$, we make use of Eq.~(\ref{s0-pol}) to observe
\begin{eqnarray}\label{fkt-0}
    f_k(t) = {\mathcal O}(t) \quad {\rm as} \quad t\rightarrow 0.
\end{eqnarray}
At infinity, a non-sharp estimate of $f_k(t)$ follows from the boundary condition Eq.~(\ref{eq:bc-inf-Th}) for $\sigma(t)$ and the ansatz Eq.~(\ref{sigma-expan-app}),
\begin{eqnarray} \label{eq:fkinf}
	f_k(t)=\mathcal{O}\left(\frac{\log^k(t)}{t}\right) \quad {\rm as} \quad t\rightarrow \infty.	
\end{eqnarray}
\begin{remark}
As a consistency check, a reader may verify that small-$t$ expansions of $f_1$ [Eq.~(\ref{f1-final})], $f_2$ [Eq.~(\ref{f2-at-zero})], $f_3$ [Eq.~(\ref{f3-at-zero})] and $f_4$ [Eq.~(\ref{f4-at-zero})] are in concert with a small-$\omega$ expansion of Eq.~(\ref{s0-pol}).
\end{remark}

\subsection{Cumulants of the counting function}\label{A-3-2}
Complete information about the cumulants $\kappa_\ell(\lambda) = \langle\!\langle  n_{1/2\pi}^\ell(\lambda)\rangle\!\rangle$ of the counting function $n_{1/2\pi}(\lambda)$ is contained in the fifth Painlev\'e transcendent $\sigma_0(t;\zeta)$ specified in Theorem~\ref{cue-pv}. To extract it, we shall combine the cumulant expansion Eq.~(\ref{cum-exp}) with the small-$\omega$ ansatz Eq.~(\ref{sigma-expan-app}) to obtain
\begin{eqnarray}\label{Fk-integrated-kappa-L}
    \kappa_\ell(\lambda) = \frac{\ell!}{(2i\pi)^\ell} {\mathcal F}_\ell (\lambda),
\end{eqnarray}
where
\begin{eqnarray} \label{Fk-integrated}
     {\mathcal F}_\ell (\lambda)=\int_{0}^{\lambda} \frac{f_\ell(t)}{t}dt.
\end{eqnarray}
Notice that the cumulants $\kappa_\ell^{(\rho)}(\lambda) = \langle\!\langle  n_{\rho}^\ell(\lambda)\rangle\!\rangle$ of the counting function $n_{\rho}(\lambda)$ associated with the ${\rm Sine}_2$ determinantal point process with the mean local density $\rho$ can be calculated by simple rescaling,
\begin{eqnarray}\label{kappa-rho}
        \kappa_\ell^{(\rho)}(\lambda) = \kappa_\ell(2\pi \rho \lambda),
\end{eqnarray}
see Eq.~(\ref{CF-diff-densities}).

(i) For one, the values $\kappa_\ell(\infty)$ for $\ell \ge 3$, can readily be determined from the global integral condition Eq.~(\ref{eq:global}), rewritten in terms of the fifth Painlev\'e transcendent ${\sigma_0(t;\zeta)=\sigma(s=-it)}$,
\begin{eqnarray} \label{global-copy} \fl\qquad
    \lim_{\lambda\rightarrow \infty} \left(
        \int_{0}^{\lambda} \frac{dt}{t} \, \sigma_0(t;\zeta) - i\tilde{\omega} \lambda +2\tilde{\omega}^2 \log \lambda
    \right) =  2\log \left[G(1+\tilde\omega)G(1-\tilde\omega)\right].
\end{eqnarray}
Indeed, substituting the small-$\omega$ expansion Eq.~(\ref{sigma-expan-app}) into Eq.~(\ref{global-copy}), and taking into account
Eq.~(\ref{f1-final}) along with the definition Eq.~(\ref{Fk-integrated}), we obtain:
\begin{eqnarray} \fl\qquad\label{Fk-exp-w}
   \lim_{\lambda\rightarrow \infty}\left\{ \tilde\omega^2  \left( {\mathcal F}_2(\lambda) + 2 \log\lambda\right)
    + \sum_{\ell=3}^{\infty} \tilde\omega^\ell {\mathcal F}_\ell(\lambda)
    \right\} \nonumber\\
    \fl\qquad\qquad\qquad\qquad\qquad
    = -2(1+\gamma)\tilde\omega^2 - 2 \sum_{k=2}^{\infty} \frac{\zeta(2k-1)}{k}\tilde\omega^{2k}.
\end{eqnarray}
The r.h.s.~of~Eq.~(\ref{Fk-exp-w}) follows from the Taylor expansion of the Barnes $G$ function~\cite{S-1988}:
\begin{eqnarray} \fl \qquad
    \log G(1+w) = \frac{w}{2}\log(2\pi) -\frac{w+(1+\gamma)w^2}{2} +\sum_{j=2}^{\infty} (-1)^j \frac{\zeta(j)}{j+1} w^{j+1},\; |w|<1.\nonumber
    \\{}
\end{eqnarray}
Equation~(\ref{Fk-exp-w}) immediately implies that
\begin{eqnarray} \label{F2k1-inf}
    {\mathcal F}_{2\ell+1}(\infty) = 0, \;\; \ell=1, 2, \dots
\end{eqnarray}
whilst
\begin{eqnarray} \label{F2k-inf}
    {\mathcal F}_{2\ell}(\infty) = -\frac{2}{\ell}\zeta(2\ell-1), \;\; \ell=2, 3, \dots,
\end{eqnarray}
where $\zeta(x)$ is the Riemann zeta function. By virtue of Eq.~(\ref{Fk-integrated-kappa-L}), this translates to
\begin{eqnarray} \label{kappa-odd-inf}
    \kappa_{2\ell+1}(\infty) = 0, \;\; \ell=1, 2, \dots
\end{eqnarray}
and
\begin{eqnarray} \label{kappa-even-inf}
    \kappa_{2\ell}(\infty) = \frac{2}{\pi}(-1)^{\ell-1}\frac{(2\ell-1)!}{(2\pi)^{2\ell-1}}\, \zeta(2\ell-1), \;\; \ell=2, 3, \dots.
\end{eqnarray}

(ii) As far as the second cumulant $\kappa_2(\lambda)=-(1/2\pi^2) {\mathcal F}_2(\lambda)$ is concerned, due to its $\log$-divergency at infinity (see Eq.~(\ref{f2-integral-as-inf}) below), we found it useful to isolate this divergency by defining
the function
\begin{eqnarray} \label{f2-integrated-reg}
    \tilde{\mathcal F}_2(\lambda) = {\mathcal F}_2(\lambda) + 2\log\lambda
\end{eqnarray}
to deduce from $\tilde\omega^2$ terms in Eq.~(\ref{Fk-exp-w}) that
\begin{eqnarray} \label{F2tilde-inf}
    \tilde{{\mathcal F}}_{2}(\infty) = -2(1+\gamma).
\end{eqnarray}
In the language of cumulants this translates to
\begin{eqnarray}
    \tilde{\kappa}_2(\lambda) = \kappa_2(\lambda) - \frac{1}{\pi^2} \log\lambda
\end{eqnarray}
with
\begin{eqnarray} \label{F2tilde-inf}
    \tilde{\kappa}_{2}(\infty) = \frac{1}{\pi^2}(1+\gamma).
\end{eqnarray}

(iii) Some of the functions ${\mathcal F}_k(\lambda)$ can be calculated explicitly, based on Eqs.~(\ref{f1-final}), (\ref{f2-final}) and (\ref{f3-final}).

(iii-a) The function ${\mathcal F}_1(\lambda)$ is a linear function
\begin{eqnarray}\label{f1-integrated}
    {\mathcal F}_1(\lambda) = i\lambda.
\end{eqnarray}
Consequently,
\begin{eqnarray}
    \kappa_1^{(\rho)}(\lambda) = \langle n_{1/2\pi}(2\pi \rho\lambda) \rangle = \rho \lambda.
\end{eqnarray}

(iii-b) The function ${\mathcal F}_2(\lambda)$ equals
\begin{eqnarray} \label{f2-integrated}
    {\mathcal F}_2(\lambda) &=& -2(1-\cos \lambda) \nonumber\\
    &+& 2 \lambda \left( {\rm Si}(\lambda) - \frac{\pi}{2} \right) + 2 \left(
    {\rm Ci}(\lambda) - \gamma -\log \lambda
    \right).
\end{eqnarray}
Its behavior at zero and infinity is described by two expansions:
\begin{eqnarray}\label{f2-integral-as-0}
    \mathcal{F}_2(\lambda) = - \pi \lambda + \frac{\lambda^2}{2} +{\mathcal O}(\lambda^4) \quad {\rm as} \quad \lambda\rightarrow 0
\end{eqnarray}
and
\begin{eqnarray} \label{f2-integral-as-inf}
    {\mathcal F}_2(\lambda) &=& -2 \log \lambda -2(1+\gamma) + \frac{2\cos\lambda}{\lambda^2} \nonumber\\
    &+& \frac{8\sin\lambda}{\lambda^3} + {\mathcal O}
    \left(
        \frac{1}{\lambda^4}
    \right)
     \quad {\rm as} \quad \lambda\rightarrow \infty.
\end{eqnarray}

The results above imply that the second cumulant equals
\begin{eqnarray} \fl\qquad \label{2nd-cum}
    \kappa_2^{(\rho)}(\lambda) &=& \frac{1}{\pi^2} \big( 1+\gamma + \log(2\pi\rho\lambda) \big)
    \nonumber\\
    &-&\frac{1}{\pi^2} \left(
        \cos(2\pi\rho\lambda) + {\rm Ci}(2\pi\rho\lambda) + 2\pi\rho\lambda
        \left(
            {\rm Si}(2\pi\rho\lambda) -\frac{\pi}{2}
        \right)
    \right).
\end{eqnarray}
This is, of course, a well known result~\cite{M-2004} for the number variance. Its asymptotic behavior can be read off from Eqs.~(\ref{f2-integral-as-0}) and (\ref{f2-integral-as-inf}),
\begin{eqnarray} \fl\qquad\qquad\qquad \label{2nd-cum-as}
    \kappa_2^{(\rho)}(\lambda) = \left\{
                                         \begin{array}{ll}
                                           \displaystyle \rho\lambda +{\mathcal O}(\lambda^2), & \hbox{$\lambda \rightarrow 0$;} \\
                                          \displaystyle \frac{1}{\pi^2}\big[ \log(2\pi\rho\lambda) + 1+ \gamma
                                            \big] + {\mathcal O}(\lambda^{-2}), & \hbox{$\lambda \rightarrow\infty$.}
                                         \end{array}
                                       \right.
\end{eqnarray}

(iv) The function ${\mathcal F}_3(\lambda)$ is purely imaginary; it equals
\begin{eqnarray} \label{f3-integrated}
    {\rm Im\,}{\mathcal F}_3(\lambda) &=&  16 \lambda \; _2F_3\left(\frac{1}{2},\frac{1}{2};\frac{3}{2},\frac{3}{2},\frac{3}{2};-\frac{\lambda^2}{4}\right) \nonumber\\
    &+&\frac{1}{3} \lambda^3 \; _3F_4\left(1,1,1;2,2,2,\frac{5}{2};-\frac{\lambda^2}{4}\right) \nonumber\\
    &-& 2 \pi ^{3/2} G_{2,4}^{2,1}\left(\frac{\lambda^2}{4}\Bigg|
       \begin{array}{c}
    \frac{1}{2},1_{\Big.} \\
    {\frac{1}{2}}^{\,},{\frac{1}{2}}^{\,},0,0 \\
    \end{array} \right) - 4 \lambda {\rm Ci}(\lambda)^2 +4 \pi  {\rm Ci}(\lambda)\nonumber\\
    &-& 8  (1-\gamma)\lambda {\rm Ci}(\lambda)
    +8  \lambda {\rm Ci}(\lambda) \log \lambda + 8 {\rm Ci}(\lambda) \sin \lambda \nonumber \\
    &-& 4  \lambda {\rm Si}(\lambda)^2+4  \pi  \lambda {\rm Si}(\lambda)+8  (1-\gamma){\rm Si}(\lambda) \nonumber\\
    &-& 8  {\rm Si}(\lambda) \log \lambda -8  {\rm Si}(\lambda) \cos \lambda -4  \left( \gamma^2 + \frac{\pi^2}{6}\right) \lambda \nonumber\\
    &-& 8  (1-\gamma)\lambda -4  \lambda \log^2 \lambda +8 (1-\gamma) \lambda  \log \lambda
    \nonumber\\
    &-& 4  \pi  \log \lambda +8 (1-\gamma) \sin \lambda  +4  \pi  \cos \lambda \nonumber\\
    &-& 8  \sin\lambda \log \lambda  - 4 (1+\gamma) \pi.
\end{eqnarray}
Here, $G_{2,4}^{2,1}$ is the Meijer $G$-function. The behavior of ${\mathcal F}_3(\lambda)$ at zero and infinity is described by two expansions:
\begin{eqnarray} \label{f3-integral-as-0}
    {\rm Im\,}{\mathcal F}_3(\lambda) = - \frac{2}{3} \pi^2\lambda + \pi \lambda^2 +{\mathcal O}(\lambda^3) \quad {\rm as} \quad \lambda\rightarrow 0,
\end{eqnarray}
and
\begin{eqnarray} \label{f3-integral-as-inf}
    {\rm Im\,}{\mathcal F}_3(\lambda) &=& - \frac{4}{\lambda} - \frac{8}{\lambda^2} \sin\lambda \, (\gamma+\log \lambda) \nonumber\\
    &+& \frac{4}{\lambda^3} \left(
        \frac{1}{3} - 6 \cos\lambda + 8 \gamma \cos\lambda + 8 \cos\lambda \log\lambda
    \right) \nonumber\\
    &+& {\mathcal O} \left(
        \frac{\log\lambda}{\lambda^4}
    \right) \quad {\rm as} \quad \lambda\rightarrow \infty.
\end{eqnarray}

This can readily be translated to the third cumulant $\kappa_3^{(\rho)}(\lambda)$ of the ${\rm Sine}_2$ counting function. Equations (\ref{Fk-integrated-kappa-L}) and (\ref{kappa-rho}) yield:
\begin{eqnarray}\label{kappa-3-new}
    \kappa_3^{(\rho)}(\lambda) = -\frac{3}{4\pi^3} {\rm Im\,} {\mathcal F}_3(2\pi\rho\lambda).
\end{eqnarray}
Notice that Eqs.~(\ref{kappa-3-new}) and (\ref{f3-integrated}) provide an explicit formula for the third cumulant of the ${\rm Sine}_2$ counting function. To the best of our knowledge, this result has never been reported in the random-matrix-theory literature.

The asymptotic behavior of the third cumulant can be read off from Eqs.~(\ref{f3-integral-as-0}) and (\ref{f3-integral-as-inf}); it has been announced in the introductory Section~\ref{m-res-dis}, see Eq.~(\ref{3rd-cum-as-intro}) and a brief discussion therein.

(v) Since the function $f_4(t)$ is quite cumbersome, see Eq.~(\ref{f4t-def}), we failed to determine ${\mathcal F}_4(\lambda)$ explicitly. Yet, its asymptotic behavior at zero and infinity is clearly within the reach. In particular,
\begin{eqnarray} \label{f4-integral-as-0}
    {\mathcal F}_4(\lambda) =  \frac{\pi^3}{3}\lambda -\frac{7\pi^2}{6} \lambda^2 +\pi \lambda^3  +{\mathcal O}(\lambda^4) \quad {\rm as} \quad \lambda\rightarrow 0.
\end{eqnarray}
To estimate ${\mathcal F}_4(\lambda)$ as $\lambda\rightarrow\infty$, we re-write it as follows
\begin{eqnarray}\label{F4-2-parts}
    {\mathcal F}_4(\lambda) = \int_0^\lambda \frac{f_4(t)}{t} dt = {\mathcal F}_4(\infty) - \int_\lambda^\infty \frac{f_4(t)}{t} dt,
\end{eqnarray}
where ${\mathcal F}_4(\infty)=-\zeta(3)$, see Eq.~(\ref{F2k-inf}). Further, we substitute Eq.~(\ref{f4t-final-inf-est}) into Eq.~(\ref{F4-2-parts}) to derive:
\begin{eqnarray} \label{f4-integral-as-inf}
    {\mathcal F}_4(\lambda) =  -\zeta(3) + {\mathcal O}\left(\left(\frac{\log \lambda}{\lambda}\right)^2 \right) \quad {\rm as} \quad \lambda\rightarrow \infty.
\end{eqnarray}

(vi) For $k\ge 5$, we notice that
\begin{eqnarray} \label{eq:Fk0}
\mathcal{F}_k(\lambda)=\mathcal{O}(\lambda) \quad {\rm as} \quad \lambda\rightarrow 0, 	
\end{eqnarray}
see Eq.~(\ref{fkt-0}).  At infinity, it holds:
\begin{eqnarray} \label{eq:Fkinf} \fl \qquad
	\mathcal{F}_k(\lambda)=\mathcal{F}_k(\infty)-\int_\lambda^\infty \frac{f_k(t)}{t}dt=\mathcal{F}_k(\infty)+\mathcal{O}\left(\frac{\log^k(\lambda)}{\lambda}\right), \qquad \lambda\rightarrow \infty. 	
\end{eqnarray}
Here we have used a non-sharp estimate Eq.~(\ref{eq:fkinf}).

\smallskip\smallskip\smallskip

\newpage

\section*{References}
\fancyhead{} \fancyhead[RE,LO]{References}
\fancyhead[LE,RO]{\thepage}

\end{document}